%% file: main.tex
\documentclass[twocolumn,nofootinbib,preprintnumbers,superscriptaddress]{revtex4-1}
\usepackage{amsmath}
\usepackage{mathrsfs} 
\usepackage{amssymb}
\usepackage{graphicx}
\usepackage{hyperref}
\usepackage{xspace}
\usepackage{bm}
\usepackage{slashed}
\usepackage{dsfont}
\usepackage{float} 
\usepackage{caption}
\usepackage{subcaption}
\usepackage{xcolor}
\usepackage{natbib}
\usepackage{url}
\usepackage{mathtools}
\usepackage{multirow}
\graphicspath{{./figs/}}

\definecolor{light-gray}{gray}{0.8}

\newcommand{\EDaff}{Higgs Centre for Theoretical Physics, School of Physics \& Astronomy,
  University of Edinburgh, Edinburgh EH9 3FD, United Kingdom.}
\newcommand{\IGMMaff}{MRC Human Genetics Unit, Institute of Genetics \& Molecular Medicine, University of Edinburgh, Edinburgh EH4 2XU, United Kingdom.}
\newcommand{\EDMathaff}{School of Mathematics, University of Edinburgh,
    Edinburgh EH9 3FD, United Kingdom.}
\newcommand{\Bonnaff}{Hausdorff Center for Mathematics, Universität Bonn, \\
    Endenicher Allee 60, D-53115 Bonn, Germany.}

\input{macros}

\usepackage{todonotes}

\begin{document}

\title{Higher-order interactions in statistical physics and machine learning: \\ A model-independent solution to the inverse problem at equilibrium}

\preprint{xxx}

\author{Sjoerd Viktor Beentjes}
\affiliation{\Bonnaff}
\affiliation{\EDMathaff}
\email{Email: sjoerd.beentjes@ed.ac.uk}

\author{Ava Khamseh}
\affiliation{\IGMMaff}
\affiliation{\EDaff}
\email{Email: ava.khamseh@ed.ac.uk}

\begin{abstract}
\vspace{0.5cm}
The problem of inferring pair-wise and higher-order interactions in complex systems involving large numbers of interacting variables, from observational data, is fundamental to many fields. Known to the statistical physics community as the inverse problem, it has become accessible in recent years due to real and simulated ‘big’ data being generated. Current approaches to the inverse problem rely on parametric assumptions, physical approximations, e.g. mean-field theory, and ignoring higher-order interactions which may lead to biased or incorrect estimates. We bypass these shortcomings using a cross-disciplinary approach and demonstrate that none of these assumptions and approximations are necessary: We introduce a universal, model-independent, and fundamentally unbiased estimator of all-order symmetric interactions, via the non-parametric framework of Targeted Learning, a subfield of mathematical statistics. Due to its universality, our definition is readily applicable to \emph{any} system at equilibrium with binary and categorical variables, be it magnetic spins, nodes in a neural network, or protein networks in biology. Our approach is targeted, not requiring fitting unnecessary parameters. Instead, it expends all data on estimating interactions, hence substantially increasing accuracy. We demonstrate the generality of our technique both analytically and numerically on (i) the 2-dimensional Ising model, (ii) an Ising-like model with 4-point interactions, (iii) the Restricted Boltzmann Machine, and (iv) simulated individual-level human DNA variants and representative traits. The latter demonstrates the applicability of this approach to discover epistatic interactions causal of disease in population biomedicine. 
\end{abstract}

\pacs{}

\maketitle


\newpage
\section{Introduction}


Starting from microscopic laws of Nature, the aim of statistical physics is to provide a macroscopic description of Nature by deriving observable quantities from the underlying laws.
In the \emph{inverse problem}, the starting point is observations for which the underlying microscopic properties, such as interactions within the constituents of the system of interest, are unknown and to be inferred.
Taking the Ising model of binary magnetic spins as an example, the goal of the forward problem is to obtain observables such as magnetisation, energy and correlation, given the Hamiltonian with its parameters.
Conversely, the goal of the inverse problem is to derive unknown interactions within spins directly from data. 

In recent years, the inverse problems are often motivated by challenges in `big data' biology due to modern high-throughput sequencing experiments and large scale patient databases.
There is a rich literature for inverse problems with the aim of inferring model parameters describing a system, \eg, via a Hamiltonian, from observational data (see, \eg, \cite{Nguyen2017} and the references therein).
Most of these methods rely on making assumptions about the parametric form of the Hamiltonian, which may not accurately reflect the true distribution of the data.
For instance, a misspecified parametric form often results in biases in the estimation of the quantities of interest when sample sizes grow without the variance in the estimation decreasing sufficiently fast.
Furthermore, in most real world settings such as interactions in biomedical data, there is no heuristic, let alone a theory, suggesting that the effects of higher-order interactions are negligible and can be ignored without consequence.
Most methods in the literature simply truncate the problem by allowing for at most pair-wise interactions \cite{Nguyen2017,PhysRevLett.112.070603,PhysRevLett.108.090201,ravikumar2010,published_papers/7111360}.
This in turn results in biased estimates, even for 2-point interactions.


The aim of this work is to introduce a universal, unbiased, and targeted framework in which symmetric $2$-point and higher-order interactions can be estimated from \emph{any} discrete data set.
We propose a \emph{model-independent} definition of $n$-point interaction amongst binary and categorical random variables. 
In contrast to earlier approaches to the inverse problems in the literature, our definition is fully \emph{non-parametric}: we make no assumptions on the parametric form of the joint or marginal probability distributions of the random variables.
Moreover, in contrast to other approaches, which consider pair-wise interactions only, ours can access higher-order interactions \cite{Nguyen2017,PhysRevLett.112.070603,PhysRevLett.108.090201,ravikumar2010,published_papers/7111360}.
We note that the non-parametric approach in Ref.~\cite{Lu14424}, although pair-wise, does incorporate dynamical interactions.
From a theoretical perspective, our definition benefits from the following three properties: (i) it is unbiased by construction and hence converges to the ground truth in the infinite data limit, (ii) it provides a natural, model-independent interpretation of higher-order interactions, and (iii) it reduces to well-known intuitive notions of interaction in parametric statistical physics models described by a Hamiltonian.
From a computational point of view, our definition of $n$-point interaction may be directly estimated from observational data by simply taking suitable combinations of expectation values.
The variance on the resulting estimate solely depends on how deeply relevant states are sampled, and it can be substantially improved when (conditional) independence between variables is known or derived.
In most practical situations where the Markovian condition is assumed, \eg, for causal identifiability \cite{10.5555/1642718}, (conditional) independence may be derived using causal structure learning algorithms such as~\cite{fastparallelpc,Kuipers2018EfficientSL,10.3389/fgene.2019.00524}.
\\


Our non-parametric definition of $n$-point interactions amongst binary random variables fits in the Targeted Learning framework of~\cite{MR2867111}, a subfield of mathematical statistics.
Targeted Learning is a probabilistic framework to estimate (causal) quantities of interest directly from a data set $\CO$, without the need to successfully estimate the true (but unknown, and often unknowable) joint probability distribution $p_{0}$ that generated $\CO$, or to expend data on estimating parameters $\theta$ of a potentially misspecified parametric model $p_{\theta}$.
Crucially, the framework requires a \emph{model-independent} definition of the (causal) quantity of interest $\alpha$, known as the \emph{target parameter}, as a functional of any candidate probability distribution $p$, not in terms of a parameter of a parametric ansatz.
This eliminates bias due to the choice of model whilst safeguarding the interpretation of $\alpha$ as a meaningful statistical quantity revealing true knowledge about the ground truth $p_0$.
Once the target parameter is established, all statistical power is used for its estimation.
The Targeted Learning framework has already been successfully applied in biomedicine and epidemiological studies~\cite{MR2867111}.
\\

This paper is structured as follows.
We discuss the non-parametric formulation of interactions using the Targeted Learning framework in Sec.~\ref{sec:non-param-int}, for the case of binary and categorical variables. We propose two definitions of interaction, namely \textit{additive} and \textit{multiplicative}, and illustrate their relation.
For a given data set and application, one choice may be more intuitive than the other, but the information they hold is equivalent.
The additive formulation in Sec.~\ref{sec:add_interaction} applies to scenarios where the subject expert takes one of the variables in the system as the `outcome' variable and is interested in estimating the effect of the interaction amongst other variables on this outcome.
The multiplicative formulation in Sec.~\ref{sec:mult_interaction} treats the variables on the same footing, and instead considers their effect (via interactions) on the energy function, and hence the joint probability distribution.
The former is more used in biomedical applications when a treatment-outcome relationship is set out at the beginning, whereas the latter is more relevant for statistical physics and, \eg, molecular networks in biology. 

Next, we provide a general formula for extracting $n$-point interactions and their interpretation directly from data.
We conclude Sec.~\ref{sec:non-param-int} by discussing how establishing conditional independence amongst variables, \eg, via the non-parametric $\chi$-squared test or more sophisticated state-of-the-art algorithms such as~\cite{fastparallelpc,Kuipers2018EfficientSL}, leads to improved estimates of the $n$-point interaction.

As a first result, we provide a concrete biological example in Sec.~\ref{sec:linear_regression_analytic}, based on interactions amongst DNA variants (epistasis) contributing to trait or disease, with data generated using a linear model.
We demonstrate analytically and numerically, that the Targeted Learning estimator obtains the correct ground truth interaction, even though it is \emph{entirely agnostic} to both the data generating process and its linearity.
This simplified example is used to guide the reader through the theoretical concepts introduced in Sec.~\ref{sec:non-param-int}.

To demonstrate universal applicability of our estimator, in Sec.~\ref{sec:RBM_analytic}, we consider a more complex Hamiltonian, namely that of the Restricted Boltzmann Machine (RBM), and analytically obtain its all-order couplings without the need for an asymptotic expansion and resummation as originally employed in \cite{PhysRevB.100.064304}.
In Sec.~\ref{sec:Ising_RBM_numerical}, we consider the $2$D Ising model and show how the \emph{same} estimator is able to predict 2-point interactions amongst nearest and non-nearest neighbour spin pairs, at various temperatures and lattice sizes.
Moreover, it correctly predicts that $3$-point and $4$-point interactions vanish.
We compare our estimations to predictions from an RBM, on data generated from the $2$D Ising model.
We limit our comparisons to the RBM as, unlike other parametric methods, it does not truncate higher-order interactions and hence does not bias lower-order interactions.

Finally, in Sec.~\ref{sec:4pt_Hamiltonian}, we generate data from a Hamiltonian with self, $2$-point, $3$-point, and $4$-point interactions and show that our Targeted Learning estimator accurately predicts higher-order interactions.
We  present numerical results at various temperatures.
This indicates that the TL estimator can be applied to obtain higher-order interactions in the case of biological networks, such as biomarker and gene expression networks.
For instance, this method is applicable to modern biomedical data sets, such as large-scale patient databases, \eg, UKBiobank, containing half a million patient samples \cite{Sudlow-ukbb}, or high-throughput sequencing experiments, \eg, 10X 1.3 million cell experiment \cite{10XMillionCells} and the Human Cell Atlas project, so far containing 4.5 million cells \cite{HumanCellAtlas}.


\section{Non-parametric formulation of interaction}\label{sec:non-param-int}

\subsection{Targeted Learning}
\label{sec:targeted_learning}
Let $\CO$ be a data set of $n$ observations $\CO_{i}$ generated by an experiment with random variable $O$, and let $p_{0}$ denote its probability distribution $O \sim p_{0}$.
The fundamental goal in probabilistic modelling is to obtain an estimate $\bar{p}$ of $p_{0}$ given the data $\CO$.
With $\bar{p}$ in hand, a relevant quantity $\alpha$ concerning the data set $\CO$ can then be estimated, such as a moment, an interaction coefficient, or a (causal) effect.

In typical situations however, given the data $\CO$ the ground truth $p_{0}$ is completely out of reach due to, \eg, a small sample size $n$ as compared to the dimensionality of the data.
To remedy this, a parametric form $p_{\theta}$ of $\bar{p}$ may be proposed, and the data may be used to fit unknown parameters $\theta$, but this often leads to an incorrect ansatz for the parametric model due to bias.
Alternatively, one may use model selection based on the data $\CO$, but will subsequently suffer from overconfidence in reporting the estimate $\bar{\alpha}$ of the quantity of interest $\alpha$.

Targeted Learning~\cite{MR2867111} is a probabilistic framework to estimate (causal) quantities of interest directly, without the need to successfully estimate $p_{0}$ or to expend data on estimating parameters $\theta$ of a (misspecified) parametric model $p_{\theta}$.
As such, it avoids the above pitfalls of the estimation problem.
Targeted Learning consists of the following steps:
\begin{enumerate}
    \item Define the \emph{statistical model} $\CM$: this is the, in general infinite dimensional, space of candidate probability distributions, b
        \begin{equation*}
            \CM = \{p \mid p \text{ a probability compatible with } \CO\},
        \end{equation*}
        based on the data $\CO$.
        By compatibility, we mean that the statistical model accommodates for \emph{a priori} knowledge regarding the data and how it is generated.
        For example, if $\CO$ is generated by $n$ binary random variables, then $\CM$ only contains $p = p(T_1,\ldots,T_n)$ with $T_i$ binary variables. Similarly, if the expectation value $\BE(T_i)$ of a variable is known to be positive, or if one or more variables are known to be (conditionally) independent, this true knowledge can be incorporated. 
        Finally, the statistical model contains the true probability distribution $p_{0} \in \CM$ by definition.
    \item Define the \emph{target mapping} $\Phi \colon \CM \to \BR^{d}$ that expresses the quantity of interest $\alpha$ as a function of the distribution $p$.
        In particular, $\alpha_{0} = \Phi(p_{0})$ is the ground truth for $\alpha$.
        For example, $\Phi$ could be a (conditional) expectation value over some or all of the variables.
        As another example, suppose that $\CO$ is generated by a random variable $O = (Y,T,W)$ where $Y$ is a continuous outcome, $T$ is a binary random variable which we will call treatment, and $W$ is a covariate.
        The treatment effect,
        \begin{equation*}
            \Phi(p) = \BE_{W}[\BE(Y \mid T=1,W) - \BE(Y \mid T = 0, W)],
        \end{equation*}
        is another example of a target parameter, often used in epidemiological studies to estimate the causal effect of a drug or treatment $T$ on health outcome $Y$ whilst correcting for confounding effects due to the covariate $W$.
    \item Apply statistical methods to obtain an estimate $\bar{\alpha}$ of the target parameter.
        We indicate a method for obtaining improved estimates of $n$-point interaction in Sec.~\ref{sec:conditional_independence}, but otherwise refer the reader to~\cite{MR2867111}.
\end{enumerate}
There are a number of important remarks to be made regarding the Targeted Learning paradigm.
First of all, the \emph{definition} of the quantity of interest $\alpha$ and its subsequent \emph{estimation} are two separate steps.
On the one hand, the quantity of interest is no longer a parameter in a potentially misspecified parametric model $p_{\theta}$, but is associated to a candidate probability distribution $p$ via the map $\Phi$ as $\Phi(p)$; thus, the quantity of interest needs to be expressed \emph{non-parametrically} as a function of $p$ forcing one to re-evaluate the interest of said quantity.
On the other hand, the method of estimation may be chosen independently from either model or target parameter.
Secondly, by expressing the quantity of interest $\alpha$ as a target parameter $\Phi(p)$ one avoids introducing bias by making an incorrect parametric ansatz $p_{\theta}$ whilst safeguarding the interpretation of $\alpha$ as a meaningful statistical quantity revealing true knowledge about the ground truth $p_0$.
And thirdly, due to bias every misspecified parametric model will not converge to the ground truth as sample size increases and variance shrinks.
Thus a non-parametric definition of a quantity of interest is essential to make full use of big data. \\

In this paper, we apply the framework of Targeted Learning to our quantity of interest, $n$-point interaction, and illustrate its application on data generated from various models. 

\subsection{Additive interaction}\label{sec:add_interaction}
Consider a random variable $O = (Y,T_1,\ldots,T_r,W)$ where $Y$ is a discrete or continuous outcome, the $T_i$ are binary random variables causally leading to the outcome $Y$, and $W$ is a covariate.
In this section, we wish to causally infer the effect of the interaction of the treatment variables $T_i$ on the outcome $Y$, for simplicity having already corrected for confounding effects $W$.
In other words, we implicitly take expectation values over strata of the covariate $W$. For example, we abbreviate
\begin{equation}
   \BE(Y \mid T_1 = 1) = \BE_{W}\left[\BE(Y \mid T_1 = 1, W) \right],
\end{equation}
where $\BE$ denotes the expectation value over $Y \mid T_1 = 1$, and $\BE_{W}$ denotes the expectation value over $W$.
Note however, that all definitions and results hold in the more general case of a fixed value $W = w$ of the covariate.

First of all, we define the statistical model, incorporating all \emph{a priori} knowledge, as in Sec.~\ref{sec:targeted_learning}:
\begin{equation*}
    \CM = \bigl\{p(Y,T_1,T_2,\ldots,T_r,W) \mid \substack{Y \text{ continuous},\, T_i \text{ binary}, \\ W \text{ a covariate}} \bigr\}.
\end{equation*}
Before defining the target parameter, we introduce some notation that will be used throughout the paper.
If a subset $T_{i_1},\ldots,T_{i_n}$ of the variables $T_1,\ldots,T_r$ is specified, then we write $\uT$ for all of the remaining variables.
For example, $\BE(T_1 \mid T_3 = 1, \uT = 0)$ denotes the conditional expectation value of $T_1$, given $T_3 = 1$ and $\uT = 0$, meaning $T_2 = T_4 = T_5 = \ldots = T_r = 0$.
We abbreviate $(T_i,T_j) = (a,b)$ to $T_{ij} = (a,b)$.

In biomedicine and epidemiological studies, a particular quantity of interest to be estimated is the causal effect of a treatment on an outcome, the \emph{average treatment effect}, \eg, the effect of a drug on health.
We express our additive notion of interaction with notation compatible with the existing literature~\cite{MR2867111,10.5555/1642718,imbens_rubin_2015}.
The \emph{average treatment effect} (ATE) of $T_i$ on $Y$ is given by
\begin{equation}\label{eq:average_treatment_effect}
    \ATE_{T_i}(Y) =\BE(Y \mid T_i = 1) -\BE(Y \mid T_i = 0).
\end{equation}
This expression is the first order derivative with respect to $T_i$ evaluated at $T_i = 0$ of the function $T_i \mapsto \BE(Y \mid T_i)$.
Indeed, for a function $f$ of a binary variable $T$ we have $\partial_T f = f(1)-f(0)$.

Next, given two binary variables $T_i, T_j$ encoding two different treatments, we obtain the ATE of treatment $T_i$ on $Y$ and the ATE of treatment $T_j$ on $Y$.
A natural question is \emph{how do these treatments interact?}
In words, how does applying treatment $T_i$ affect the effect of treatment $T_j$ on $Y$, and vice versa?
In order to isolate the effects of $T_i$ and $T_j$ on $Y$, the other treatments are not applied, \ie, we condition on $\uT = 0$.
We now define the first target mapping, $\Phi^{a}_{i,j}$, which is our non-parametric additive formulation  of $2$-point interaction between binary random variables.
The \emph{additive interaction} $I^{a}_{i,j}$ between the binary variables $T_i$ and $T_j$, is given by the difference of the effect of changing $T_i \colon 0 \to 1$ on $Y$ given $T_j = 1$, and the effect of changing $T_i \colon 0 \to 1$ on $Y$ given $T_j = 0$, \ie,
\small
\begin{equation}\label{eq:interaction_add_2}
\begin{split}
    &\CM \ni p \mapsto \Phi^{a}_{i,j}(p) \coloneqq I^{a}_{i,j} \\
        = &\bigl[\BE(Y \mid T_{ij} = (1,1), \uT = 0) -\BE(Y \mid T_{ij} = (0,1), \uT = 0)\bigr] \\
        -&\bigr[\BE(Y \mid T_{ij} = (1,0), \uT = 0) -\BE(Y \mid T_{ij} = (0,0), \uT = 0)\bigr].
\end{split}
\end{equation}
\normalsize
Note that interaction is a difference of ATEs, \ie, $I^{a}_{i,j} = \ATE_{T_i}(Y \mid T_j = 1, \uT = 0) - \ATE_{T_i}(Y \mid T_j = 0, \uT = 0)$.
Thus, the interaction $I^{a}_{i,j}$ is the change of effect of $T_i$ on $Y$ when changing $T_j$, conditioned on $\uT = 0$.
This change of effect may be expressed as the (symmetric) double derivative with respect to $T_i$ and $T_j$, and so $I^{a}_{1,2}$ is also the change of effect of $T_j$ on $Y$ when changing $T_i$.
Formally, this reads
\begin{equation}\label{eq:interaction_is_symmetric}
    I^{a}_{i,j} = I^{a}_{j,i},
\end{equation}
as one readily deduces from Eq.~\ref{eq:interaction_add_2}. Indeed, given a function $f \colon \{0,1\}^2 \to \BR$ of two binary variables $x$ and $y$, $\partial_{x}\partial_{y}f = \partial_{y}\partial_{x}f$.

Although numerically, the effect of $T_i$ on the effect of $T_j$ on $Y$ is the same as the effect of $T_j$ on the effect of $T_i$ on $Y$, only one direction might admit a sensible interpretation. This is similar to the causal interpretation of the set of equations $Y=mX+b$ or $X = m'Y+b'$ that is provided by a directed acyclic graph (DAG)~\cite{10.5555/1642718} and is not captured by the equation alone.
In contrast, note that the \emph{sign} of the interaction is uniquely determined since a \emph{direction} is specified: it is the effect on $Y$ of changing $T_i$ from $0$ to $1$, not from $1$ to $0$, that we compare to the effect on $Y$ of changing $T_j$ from $0$ to $1$.
Both the symmetry and the sign of $I^{a}_{i,j}$ are illustrated in the following diagram:
\begin{equation}\label{eq:diagram_2pt_interaction}
\begin{tikzpicture}[baseline=(current  bounding  box.center)]
  \matrix (m) [matrix of math nodes, row sep=1.5em,
    column sep=1.5em]{
    (1,1) & \textcolor{white}{A} & (1,0) \\
    \textcolor{white}{A} & & \textcolor{white}{A} \\
    (0,1) & \textcolor{white}{A} & (0,0) \\};
  \path[-stealth]
    (m-1-3) edge (m-1-1)
    (m-3-1) edge (m-1-1)
    (m-3-3) edge (m-3-1)
            edge (m-1-3)
    (m-2-3) edge [densely dotted] (m-2-1)
    (m-3-2) edge [densely dotted] (m-1-2);
\end{tikzpicture}
\end{equation}
We introduce the shorthand $A(t_i,t_j) =\BE(Y \mid T_{ij} = (t_i,t_j), \uT = 0)$ where $t_i,t_j \in \{0,1\}$.
In the diagram, vertex $(t_i,t_j)$ represents the expected outcome $A(t_i,t_j)$.
An arrow represents the average treatment effect of the variable of which the value changes, where the sign is dictated by `target minus source'.
For example, the left vertical arrow encodes the average treatment effect of $T_i \colon 0 \to 1$ on $Y$ given $T_j = 1$, \ie,
\begin{equation}
    A(1,1) - A(0,1) = \ATE_{T_i}(Y \mid T_j = 1, \uT = 0).
\end{equation}
Finally, either dotted arrow encodes the interaction between the effects of $T_i$ and $T_j$ on the outcome $Y$, together with its inherent symmetry.
Indeed, via the sign convention `target minus source', the diagram yields relations,
\small
\begin{align*}
    I^{a}_{i,j} &= \ATE_{T_i}(Y \mid T_j = 1, \uT = 0) - \ATE_{T_i}(Y \mid T_j = 0, \uT = 0), \\
    I^{a}_{j,i} &= \ATE_{T_j}(Y \mid T_i = 1, \uT = 0) - \ATE_{T_j}(Y \mid T_i = 0, \uT = 0),
\end{align*}
\normalsize
where the first line is encoded by the horizontal arrow and the second line by the vertical arrow.

Next, we define  the additive $n$-point interaction on the outcome $Y$. 
Whereas the $2$-point interaction is a difference of two ATEs, hence a sum of $2^2 = 4$ expectation values, the $3$-point interaction involves $2^3 = 8$ such terms and, more generally, the $n$-point interaction involves $2^n$ terms.
We introduce notation in order to state the formula of a general $n$-point interaction.

Consider a subset $K= \{i_1,\ldots,i_{\ell(K)}\} \subset \{1,\ldots,r\}$ of the indices for the treatment variables $T_1,\ldots,T_r$ in the random variable $O$.
Here, in general, given a further subset $J \subset K$ we denote its number of elements by $\ell(J)$.
We write $e^{(\ell(K))}_J$ for the $\ell(K)$-tuple of elements,
\begin{equation}
   e^{(\ell(K))}_J = (e_{i_1},\ldots,e_{i_{\ell(K)}}),
\end{equation}
where $e_{i_j}$ equals $1$ if $i_j \in J$ and $0$ if $i_j \not\in J$.
For example, if $J = \{2,7\} \subset \{1,2,4,5,7\} = K$, then
\begin{equation}
    e^{(\ell(K))}_{J} = e^{(5)}_{J} = (0,1,0,0,1).
\end{equation}
Finally, we write $T_K = (T_{i_1},\ldots,T_{i_{\ell(K)}})$ where $i_j \in K$ for all $1 \leq j \leq \ell(K)$.
Continuing the previous example, we have $\ell(K) = 5$ and $\ell(J) = 2$.
The $5$-point interaction between the variables $T_K = (T_1,T_2,T_4,T_5,T_7)$ is a sum of $2^5 = 32$ terms, and it will involve the expectation value
\small
\begin{equation}
\begin{split}
   &\BE\bigl(Y | T_K = e^{(5)}_J, \uT = 0\bigr) = \\
   &\BE\bigl(Y | (T_1,T_2,T_4,T_5,T_7) = (0,1,0,0,1), \uT = 0\bigr).
\end{split}
\end{equation}
\normalsize
The next target mapping, $\Phi^{a}_{i_1,\ldots,i_n}$, is our non-parametric additive formulation of $n$-point interaction.
\begin{defn}\label{def:interaction_add}
    Let $K = \{i_1,\ldots,i_n\} \subset \{1,\ldots,r\}$ be a subset of indices.
    The additive $n$-point interaction amongst the effects of the binary treatments $T_K = (T_{i_1},\ldots,T_{i_n})$ on the outcome $Y$, is
    \small
    \begin{equation}\label{eq:interaction_add}
    \begin{split}
        &\CM \ni p \mapsto \Phi^{a}_{i_1,\ldots,i_n}(p) \coloneqq I^{a}_{i_1,\ldots,i_n} = \\
        &\sum_{j=0}^{n} (-1)^{n-j} \left(\sum_{J \subset K \colon \ell(J) = j}\BE\bigl(Y \mid T_K = e^{(n)}_J, \uT = 0 \bigr) \right),
    \end{split}
    \end{equation}
    \normalsize
    where the internal sum runs over all subsets $J \subset K$ of length $\ell(J) = j$.
\end{defn}
This is the $n$th order boolean derivative of the function $(T_1,\ldots,T_n) \mapsto \BE(Y \mid T_1,\ldots,T_n)$. 
As an example, consider the $3$-point interaction $I^{a}_{1,2,3}$ amongst the effects of the binary random variables $T_1,T_2,T_3$ on the outcome $Y$.
Then $T_K = (T_1,T_2,T_3)$ with $K = \{1,2,3\}$, and $I^{a}_{1,2,3}$ consists of $2^3 = 8$ terms. 
Explicitly, the interaction reads
\small
\begin{equation*}
\begin{split}
    I^{a}_{1,2,3} = 
    &\BE(Y | T_K=(1,1,1),\uT = 0) - \BE(Y | T_K=(1,1,0),\uT = 0) \\
    - &\BE(Y | T_K=(1,0,1),\uT = 0) - \BE(Y | T_K=(0,1,1),\uT = 0) \\
    + &\BE(Y | T_K=(1,0,0),\uT = 0) + \BE(Y | T_K=(0,1,0),\uT = 0) \\
    + &\BE(Y | T_K=(0,0,1),\uT = 0) - \BE(Y | T_K=(0,0,0),\uT = 0).
\end{split}
\end{equation*}
\normalsize
Note that the four terms with a `$+$' are those for which an \emph{odd} number of variables satisfies $T_i = 1$, whereas the four terms with a `$-$' are those for which an \emph{even} number of variables satisfies $T_i = 1$. This is the other way around for $2$-point interactions, see Eq.~\ref{eq:interaction_add_2}, and depends on the parity of the number $n$ in general as follows from Eq.~\ref{eq:interaction_add}.
    
For a diagrammatic relation between the $3$-point interaction and the $2$-point interactions from which it is built, as in Eq.~\ref{eq:diagram_2pt_interaction}, together with an interpretation of $n$-point interaction in general, we refer the reader to section~\ref{sec:interpretation}.
Finally, we show in Cor.~\ref{cor:symmetry_interaction_add} that $I^{a}_{i_1,\ldots,i_n}$ is symmetric under any permutation of its indices $i_1,\ldots,i_n$.
\\

Our additive notion of $n$-point interaction amongst binary random variables readily generalizes to the setting of categorical variables.
Recall that a \emph{categorical random variable} $X$ distinguishes $k+1$ categories, typically labelled by integers $0,1,\ldots,k$, where the probability of being in category $i$ equals $p(X = i) = p_i$ and the $p_i \in [0,1]$ sum to $1$.
If $k = 1$ then $X$ is a binary random variable.
The categorical case leads to new phenomena, most importantly the dependence of the interaction $I^{a}_{i_1,\ldots,i_n}$ on the particular categories of $T_{i_1},\ldots,T_{i_n}$ one considers.
Indeed, \eg, $I^{a}_{i,j}$ in the binary case has a unique double derivative whereas in general a derivative is a \emph{function} that needs to be evaluated at a point (\ie, a category) in order to obtain a \emph{value}.

Before we define interaction as a target parameter, we again specify the statistical model:
\begin{equation*}
    \CM = \bigl\{p(Y,T_1,T_2,\ldots,T_r,W) \mid \substack{Y \text{ continuous},\, T_i \text{ categorical with} \\ k_i \in \BN \text{ categories},\, W \text{ a covariate}} \bigr\}
\end{equation*}
Let $t_i,t_i'$ and $t_j,t_j'$ be categories of $T_i$ and $T_j$ respectively.
First, we define the interaction between the effects of $T_i$ on $Y$ as $T_i$ changes from $t_i$ to $t_i'$ and the effect of $T_j$ on $Y$ as $T_j$ changes from $t_j$ to $t_j'$.
We write $T_i \colon t_i \to t_i'$ to mean that $T_i$ changes from $t_i$ to $t_i'$. 
For example, the average treatment effect of $T_i \colon t_i \to t_i'$ on $Y$, given $T_j = t_j$, reads
\small
\begin{equation}
\begin{split}
    &\ATE_{T_i \colon t_i \to t_i'}(Y \mid T_j = t_j) = \\
    &\BE(Y \mid T_i = t_i', T_j = t_j) - \BE(Y \mid T_i = t_i, T_j = t_j),
\end{split}
\end{equation}
\normalsize
The target mapping for the additive interaction between the effects of $T_i$ and $T_j$ on the outcome $Y$ is the following.
The \emph{additive interaction} $I^{a}_{i,j}(t_i t_i';t_j t_j')$ between the effect of the categorical variables $T_i \colon t_i \to t_i'$ on $Y$ and the effect of $T_j \colon t_j \to t_j'$ on $Y$, is given by the difference of their respective treatment effects, \ie, 
\begin{equation}\label{eq:interaction_add_2_categorical}
\begin{split}
    I^{a}_{i,j}(t_i t_i';t_j t_j') &= \ATE_{T_i \colon t_i \to t_i'}(Y \mid T_j = t_j', \uT = 0) \\
        &- \ATE_{T_i \colon t_i \to t_i'}(Y \mid T_j = t_j, \uT = 0).
\end{split}
\end{equation}
This definition reduces to that of Eq.~\ref{eq:interaction_add_2} in the case where both $T_i$ and $T_j$ are binary with labels $\{0,1\}$, \ie,
\begin{equation}
    I^{a}_{i,j}(01;01) = I^{a}_{i,j}.
\end{equation}
For properties of $n$-point interaction in this more general setting, such as transitivity, see~App.~\ref{app:categorical_variables}.

\subsection{Multiplicative interaction}\label{sec:mult_interaction}
In this section, we define the multiplicative interaction amongst $n$ binary random variables $X_i$ forming part of a random variable $O = (X_0,\ldots,X_r)$ with joint probability density function $p_0$.
First of all, we specify the statistical model as in Sec.~\ref{sec:targeted_learning}:
\small
\begin{equation*}
    \CM = \bigl\{p(X_0,X_1,\ldots,X_r) \mid X_i \text{ binary random variables} \bigr\}.
\end{equation*}
\normalsize
The target map, $\Phi^{m}_{i,j}$, is our non-parametric multiplicative formulation of $2$-point interaction between the binary random variables $X_i$ and $X_j$:
\begin{equation}\label{eq:interaction_mult_2}
\begin{split}
    &\CM \ni p \mapsto \Phi^{m}_{i,j}(p) \coloneqq I^m_{i,j} = \\
    &\frac{p(X_{ij} = (1,1) \mid \uX = 0)}{p(X_{ij} = (1,0) \mid \uX = 0)} \frac{p(X_{ij} = (0,0) \mid \uX = 0)}{p(X_{ij} = (0,1) \mid \uX = 0)}.
\end{split}
\end{equation}
The above ratios of conditional probability distributions may be expressed in terms of the joint probability distribution $p$ since all are conditioned on $\uX = 0$.
As a result, the $2$-point interaction  between, \eg, $X_1$ and $X_2$ can be directly estimated from the data, as it reduces to
\begin{equation}\label{eq:interaction_mult_2_no_conditional}
    I^m_{1,2} = \frac{p(1,1,0,\ldots,0)}{p(1,0,0,\ldots,0)} \frac{p(0,0,0,\ldots,0)}{p(0,1,0,\ldots 0)}.
\end{equation}
Moreover, if a variable $X_k$ appearing in the $\uX$ is independent of both $X_i$ and $X_j$, then one need not condition on $X_k$. In this case, statistics may be improved as $X_k$ drops out of the conditional joint distribution $p(X_i,X_j | \uX)$ for $(X_i,X_j)$. See Sec.~\ref{sec:conditional_independence} where this argument is explained in detail. 
\\

The multiplicative $2$-point interaction $I^{m}_{i,j}$ of Eq.~\ref{eq:interaction_mult_2} between the binary random variables $X_i, X_j$ can also be expressed in terms of their (conditional) expectation values.
Numerically, this re-formulation allows one to obtain uncertainties on the estimates of $I^m_{i,j}$ using, \eg, the empirical bootstrap procedure, see Sec.~\ref{sec:Ising_RBM_numerical}.
The expression of $I^{m}_{i,j}$ in terms of expectation values is derived via the product rule for probabilities, which yields 
\begin{align*}
    \frac{p(X_{ij} = (0,0) \mid \uX = 0)}{p(X_{ij} = (1,0) \mid \uX = 0)}
    = \frac{1-\BE(X_i\mid X_j=0, \uX = 0)}{\BE(X_i\mid X_j=0, \uX = 0)},
\end{align*}
and similarly for the remaining two probabilities. Therefore, the multiplicative 2-point interaction Eq.~\ref{eq:interaction_mult_2} can be written as a combination of expectation values:
\small
\begin{align}\label{eq:interaction_mult_2_expectation}
     I^m_{i,j} = \frac{\BE(X_i | X_j = 1, \uX = 0)}{\BE(X_i| X_j=0, \uX = 0)}
        \frac{\bigl(1-\BE(X_i | X_j = 0, \uX = 0)\bigr)}{\bigl(1-\BE(X_i | X_j = 1, \uX = 0)\bigr)}.
\end{align}
\normalsize
It is not hard to see that this expression is symmetric under $X_i \leftrightarrow X_j$.
For a general statement, see Prop.~\ref{prop:symmetry_interaction_mult}.

The following is the target map for our non-parametric multiplicative formulation of $n$-point interaction.
\begin{defn}\label{def:interaction_mult}
    Let $K = \{i_1,\ldots,i_n\} \subset \{0,1,\ldots,r\}$ be a subset of indices.
    The multiplicative $n$-point interaction amongst the binary random variables $X_K = (X_{i_1},\ldots,X_{i_n})$ is defined as
    \small
    \begin{equation}\label{eq:interaction_mult}
    \begin{split}
        &\CM \ni p \mapsto \Phi^{m}_{i_1,\ldots,i_n}(p) \coloneqq I^m_{i_1,\ldots,i_n} = \\
        &\prod_{j=0}^{n} \left(\prod_{J \subset K \colon \ell(J) = j} p\bigl(X_K = e^{(n)}_J  \mid \uX = 0\bigr)^{(-1)^{n-j}} \right),
    \end{split}
    \end{equation}
    \normalsize
    where the internal product runs over all subsets $J \subset K$ of length $\ell(J) = j$.
\end{defn}
As an example, consider the $3$-point interaction $I^m_{1,2,3}$ amongst the binary random variables $X_1,X_2,X_3$.
It consists of $2^3 = 8$ terms.
Writing $X_K = X_{1,2,3}$ for the triple $(X_1,X_2,X_3)$, the interaction reads
\small
\begin{equation}\label{eq:interaction_mult_3}
\begin{split}
    I^m_{1,2,3} =
    &\frac{p(X_K = (1,1,1) \mid \uX = 0)}{p(X_K = (1,1,0) \mid \uX = 0)}
    \frac{p(X_K = (1,0,0) \mid \uX = 0)}{p(X_K = (1,0,1) \mid \uX = 0)} \\
    \cdot &\frac{p(X_K = (0,1,0) \mid \uX = 0)}{p(X_K = (0,1,1) \mid \uX = 0)}
    \frac{p(X_K = (0,0,1) \mid \uX = 0)}{p(X_K = (0,0,0) \mid \uX = 0)}.
\end{split}
\end{equation}
\normalsize
Note that the four terms in the numerator are those for which an \emph{odd} number of variables satisfies $X_i = 1$, whereas the four terms in the denominator are those for which an \emph{even} number of variables satisfies $X_i = 1$.
This is the other way around for $2$-point interactions, see Eq.~\ref{eq:interaction_mult_2}, and depends on the parity of the number $n$ in general as follows from Eq.~\ref{eq:interaction_mult}.
There is a large amount of symmetry in this expression:
\small
\begin{equation}
    I^{m}_{1,2,3}
    = \frac{I^{m}_{1,2}(X_3 = 1)}{I^{m}_{1,2}(X_3 = 0)}
    = \frac{I^{m}_{1,3}(X_2 = 1)}{I^{m}_{1,3}(X_2 = 0)}
    = \frac{I^{m}_{2,3}(X_1 = 1)}{I^{m}_{2,3}(X_1 = 0)},
\end{equation}
\normalsize
where $I^{m}_{1,2}(X_3 = 1)$ means that all instances of $X_3$ are conditioned as $X_3 = 1$, as opposed to $X_3 = 0$.
The fact that all three expressions (and the remaining three) are equal follows from the $3! = 6$ symmetries of $I^{m}_{1,2,3}$ of Prop.~\ref{prop:symmetry_interaction_mult} below.
We also remark that $I^m_{1,2,3}$ can be readily computed from data since the ratios of conditional probability distributions appearing in this equation may be expressed in terms of the joint probability distribution $p$ of $O$.
As for the $2$-point interaction, a general $3$-point interaction $I^{m}_{i,j,k}$ can be expressed in terms of expectation values:
\begin{equation}\label{eq:3pt_mult_expectation_values}
\begin{split}
    I^{m}_{i,j,k} &= \frac{R_{i;jk}(1,1)}{R_{i;jk}(1,0)} \frac{R_{i;jk}(0,0)}{R_{i;jk}(0,1)}, 
\end{split}
\end{equation}
where we have defined, for any variable $X_i$ conditioned on $X_{jk} = (X_j,X_k) = (a,b)$, the following expression,
\begin{equation}
    R_{i;jk}(a,b) = \frac{\BE(X_i \mid X_{jk} = (a,b), \uX = 0)}{1-\BE(X_i\mid X_{jk}=(a,b), \uX = 0)}.
\end{equation}
For any binary variable $T$ with $p(T=1) = p$, this fraction encodes the ratio $p/(1-p)$.
The expression of the $3$-point interaction $I^{m}_{i,j,k}$ in terms of expectation values over binary random variables is used in Sec.~\ref{sec:Ising_RBM_numerical} for the purposes of numerical estimation via statistical bootstrap.
It is straightforward to write down an expression similar to that of Eq.~\ref{eq:3pt_mult_expectation_values} for any $n$-point interaction, making statistical bootstrap applicable in general.

Finally, we make explicit a basic and natural symmetry that is inherent in our non-parametric formulation of $n$-point interaction $I^{m}_{i_1,\ldots,i_n}$ amongst the binary random variables $X_{i_1},\ldots,X_{i_n}$: $n$-point interaction is invariant under any permutation $\sigma$ of the $n$ variables, namely
\begin{equation}\label{eq:symmetry_interaction_mult_in_paper}
    I^{m}_{i_1,\ldots,i_n} = I^{m}_{\sigma(i_1,\ldots,i_n)}.
\end{equation}
We refer the interested reader to Prop.~\ref{prop:symmetry_interaction_mult} for a proof.

\subsection{Relating additive and multiplicative formulations}
\label{sec:relating_add_mult}
Consider binary random variables $X_i$ forming part of a random variable $O = (X_0,\ldots,X_r)$ with joint probability density function $p$.
In this section, we show that the non-parametric formulation of multiplicative $n$-point interaction amongst the variables $X_{i_1},\ldots,X_{i_n}$ is equivalent to the additive $n$-point interaction amongst the effects of the variables $X_{i_1},\ldots,X_{i_n}$ on a particular outcome canonically related to $p$; in fact, when both interactions are defined, they are related by a logarithm.
This outcome is the negative of the energy function $E(\uX)$, obtained from the joint distribution $p$ via
\small
\begin{equation}
    p(\uX) = \exp\bigl(-(-\ln p(\uX))\bigr) \quad \text{and } E(\uX) = -\ln p(\uX).
\end{equation}
\normalsize
Note that the expectation value of $E(\uX)$ is the Shannon entropy of the probability distribution $p$.
More precisely, the additive and multiplicative $n$-point interactions amongst the $X_{i_1},\ldots,X_{i_n}$ are related via
\begin{equation}\label{eq:equivalence_interactions_add_mult}
    \ln \bigl(I^{m}_{i_1,\ldots,i_n}\bigr) = I^{a}_{i_1,\ldots,i_n},
\end{equation}
where the additive $n$-point interaction is computed with respect to the outcome $Y = -E(\uX)$.
Indeed, this follows directly as taking the logarithm of Eq.~\ref{eq:interaction_mult} yields Eq.~\ref{eq:interaction_add}.
Here we have used that
\small
\begin{equation}
    \frac{p(X_{i_1,\ldots,i_n} = e^{(n)}_{J} \mid \uX = 0)}{p(X_{i_1,\ldots,i_n} = e^{(n)}_{J'} \mid \uX = 0)}
    = \frac{p(X_{i_1,\ldots,i_n} = e^{(n)}_{J}, \uX = 0)}{p(X_{i_1,\ldots,i_n} = e^{(n)}_{J'}, \uX = 0)},
\end{equation}
\normalsize
\ie, a ratio of conditional probabilities is equal to the corresponding ratio of joint probabilities, together with the fact that an expectation value of the number
\begin{equation*}
    \alpha = \ln p(X_{i_1,\ldots,i_n} = e^{(n)}_{J}, \uX = 0)
\end{equation*}
equals the number itself: $\BE(\alpha) = \alpha$.
Take, as an example, the $2$-point interaction $I^m_{1,2}$ between $X_1$ and $X_2$ of Eq.~\ref{eq:interaction_mult_2}:
\small
\begin{equation*}
\begin{split}
    I^m_{1,2} =
    &\frac{p(X_{12} = (1,1) \mid \uX = 0)}{p(X_{12} = (1,0) \mid \uX = 0)}
    \frac{p(X_{12} = (0,0) \mid \uX = 0)}{p(X_{12} = (0,1) \mid \uX = 0)} \\
    =
    &\frac{p(X_{12} = (1,1), \uX = 0)}{p(X_{12} = (1,0), \uX = 0)}
    \frac{p(X_{12} = (0,0), \uX = 0)}{p(X_{12} = (0,1), \uX = 0)}.
\end{split}
\end{equation*}
\normalsize
Taking the logarithm, and simplifying notation to $p_{12}(X_1,X_2) = p(X_1,X_2, \uX = 0)$, yields
\begin{equation*}
\begin{split}
    \ln I^m_{1,2}   &= \ln p_{12}(1,1) - \ln p_{12}(1,0) \\
                    &-\ln p_{12}(0,1) + \ln p_{12}(0,0) = I^{a}_{1,2},
\end{split}
\end{equation*}
as claimed.
Note that we recognise the canonical outcome $Y = - E(\uX) = \ln p(\uX)$.

As a corollary, we deduce the general permutation symmetry of the additive $n$-point interaction, namely
\begin{equation}
    I^{a}_{i_1,\ldots,i_n} = I^{a}_{\sigma(i_1,\ldots,i_n)}
\end{equation}
for any permutation $\sigma$; see Cor.~\ref{cor:symmetry_interaction_add} for a proof.

\subsection{Interpreting higher-order interactions}\label{sec:interpretation}
The non-parametric $n$-point interaction consists of $2^n$ terms, as it involves $n$ binary variables turning on or off.
Consequently, the interpretation of such higher-order interactions is somewhat delicate.
To fix ideas, we focus on the case of additive $3$-point interactions, the discussion readily generalises to $n$-point interactions.

Let $T_1, T_2, T_3$ be three binary random variables and let $Y$ denote the outcome.
The interpretation of the 3-point interaction $I^{a}_{1,2,3}$ of Sec.~\ref{sec:add_interaction} is similar to that of the $2$-point interaction in Eq.~\ref{eq:diagram_2pt_interaction}. Consider the following diagram:

\vspace{-0.5cm}
\begin{equation}\label{eq:diagram_3pt_interaction}
\begin{tikzpicture}
  \matrix (m) [matrix of math nodes, row sep=1.0em,
    column sep=1.0em]{
    & (1,1,0) & & (1,0,0) \\
    (1,1,1) & & (1,0,1) & \\
    & (0,1,0) & & (0,0,0) \\
    (0,1,1) & & (0,0,1) & \\};
  \path[-stealth]
    (m-1-2) edge (m-2-1)
    (m-1-4) edge (m-1-2) edge (m-2-3)
    (m-3-2) edge [densely dotted] (m-1-2)
            edge [densely dotted] (m-4-1)
    (m-2-3) edge [-,line width=6pt,draw=white] (m-2-1)
            edge (m-2-1)
    (m-3-4) edge (m-1-4) edge (m-4-3)
            edge [densely dotted] (m-3-2) 
    (m-4-1) edge (m-2-1)
    (m-4-3) edge (m-4-1)    
            edge [-,line width=6pt,draw=white] (m-2-3)
            edge (m-2-3);
\end{tikzpicture}
\end{equation}

We have introduced the shorthand
\begin{equation}
    A(t_1,t_2,t_3) = \BE(Y \mid T_{123} = (t_1,t_2,t_3), \uT = 0),
\end{equation}
where $t_1,t_2,t_3 \in \{0,1\}$.
Vertex $(t_1,t_2,t_3)$ represents the expected outcome $A(t_1,t_2,t_3)$.
An arrow represents the ATE of the variable of which the value changes, where the sign is again dictated by `target minus source'.
For example, the front left vertical arrow encodes the ATE:
\small
\begin{equation*}
    A(1,1,1) - A(0,1,1) = \ATE_{T_1}(Y \mid T_{23} = (1,1), \uT = 0).
\end{equation*}
\normalsize
The twelve arrows along the six faces of the cube (one horizontal and one vertical each) encode the six additive $2$-point interactions between the effects of two out of the three variable $T_1,T_2,T_3$ on the outcome $Y$, with the third variables fixed to $0$ or $1$, together with their inherent symmetry as discussed in Sec.~\ref{sec:add_interaction}.
Either of the three arrows through the sides of the cube, depicted in the figure below, encodes the additive $3$-point interaction between the effects of $T_1,T_2,T_3$ on the outcome $Y$. 
\begin{equation}
\begin{tikzpicture}
  \matrix (m) [matrix of math nodes, row sep=0.7em,
    column sep=-2em]{
    & & I_{23}(T_1 = 1) & & \\
    & & & I_{12}(T_3 = 0) & \\
    I_{13}(T_2 = 1) & & & & I_{13}(T_2 = 0) \\
    & I_{12}(T_3 = 1) & & & & \\
    & & I_{23}(T_1 = 0) & & \\};
  \path[-stealth]
    (m-2-4) edge (m-4-2)
    (m-3-5) edge (m-3-1)
    (m-5-3) edge (m-1-3);
\end{tikzpicture}
\end{equation}
We have the relations `target minus source':
\begin{equation}
\begin{split}
    I^{a}_{1,2,3}
        &= I^{a}_{1,2}(T_3 = 1) - I^{a}_{1,2}(T_3 = 0) \\
        &= I^{a}_{1,3}(T_2 = 1) - I^{a}_{1,3}(T_2 = 0) \\
        &= I^{a}_{2,3}(T_1 = 1) - I^{a}_{2,3}(T_1 = 0).
\end{split}
\end{equation}
This is our three-fold interpretation of $3$-point interaction: it is the change in the $2$-point interaction between $T_1$ and $T_2$, \ie, $I^{a}_{1,2} = I_{1,2}^{a}(T_3 = 0)$, as $T_3$ is turned on $T_3 \colon 0 \to 1$, yielding $I^{a}_{1,2}(T_3 = 1)$.
In other words, $I^{a}_{1,2,3}$ captures the dependence of the $2$-point interaction between $T_1$ and $T_2$ as a function of $T_3$. We conclude that the sign and magnitude of a $3$-point interaction can be interpreted relative to any of the $2$-point interactions between two out of the three variables. 

As an illustration, we present the natural interpretation of symmetric higher-order interactions in the following real-world examples:
\begin{enumerate}
	\item Genomic variant-interaction leading to disease: The additive $2$-point interaction answers the question \textit{Does variant $i$ influence disease differently depending on the status of variant $j$, and by how much?} The $3$-point interaction answers the question \textit{Does the interaction between variant $i$ and variant $j$ influence disease differently depending on the status of variant $k$, and by how much?} The same interpretation applies to combination therapy where the effects of multiple drug-interactions on health are examined.
	\item Molecular networks: The multiplicative $2$-point interaction answers the question \textit{Does the likelihood of gene $i$ being on increase or decrease depending on whether gene $j$ is on or off, and by how much?} Similarly, the $3$-point interaction answers the question \textit{Does the interaction between gene $i$ and gene $j$ influence outcome differently, depending on the status of gene $k$, and by how much?}
\end{enumerate}
The cause-effect directionalities are either provided by subject experts, discovered by perturbation experiments, or derived by causal discovery algorithms.


\subsection{Improving statistics via (conditional) independence}\label{sec:conditional_independence}
The non-parametric formulations of $n$-point interaction amongst the random variables $X_{i_1},\ldots,X_{i_n}$, Eq.~\ref{eq:interaction_add} and Eq.~\ref{eq:interaction_mult}, require conditioning on all remaining variables in the system.
In order to improve statistical power when estimating interactions directly from data, this requirement can be relaxed under the assumption that the system is \emph{Markovian}.
Then, one need only condition on the \emph{parents} of the variables $X_{i_1},\ldots,X_{i_n}$ involved in the interaction.
A finite collection of categorical random variables $\{X_i\}_{i=1}^r$ is a \emph{Markov random field} if
\begin{enumerate}
    \item the joint distribution is strictly positive, \ie, $p(X_i = x_i \text{ for } 1 \leq i \leq r) > 0$, and
    \item for each  $X_i$ there exists a set of \emph{parents} $\CP_i \subset \{1,2,\ldots,r\}$, not including $i$, which is the minimal set such that the following condition holds:
    \small
    \begin{equation*}
        p\bigl(X_i = x_i \mid \uX = \underline{x} \bigr) = p\bigl( X_i = x_i \mid X_j = x_j \text{ for } j \in \CP_i \bigr).
    \end{equation*}
    \normalsize
    In words, the conditional probability of $X_i = x_i$ only depends on its parents $X_j = x_j$, $j \in \CP_i$.
\end{enumerate}
It is not hard to see that the set of parents $\CP_i$ of the variable $i$ is unique.
To any Markov random field one can associate a finite undirected graph with a vertex for each variable $X_i$ and an edge connecting $X_i$ and $X_j$ if $j \in \CP_i$, \ie, $X_j$ is a parent of $X_i$.
The Hammersley--Clifford Theorem \cite{HamCliff1971} (see also~\cite{MR329039}) states that $\{X_i\}_{i=1}^r$ is a Markov random field if and only the joint probability distribution $p(X_1,\ldots,X_r)$ is a \emph{Gibbs ensemble}, \ie, there exists a Hamiltonian $E(X_1,\ldots,X_r)$ such that
\begin{equation}
    p(X_1,\ldots,X_r) = \frac{1}{\CZ} \exp\bigl(- E(X_1,\ldots,X_r) \bigr),
\end{equation}
where $\CZ$ denotes the partition function normalising the distribution.
As a result, \emph{all} energy-based models of binary and categorical random variables are Markov random fields, and may thus benefit from the aforementioned improvement in statistical power when computing $n$-point interactions directly from data.
These facts are leveraged in the numerical sections~\ref{sec:Ising_RBM_numerical} and~\ref{sec:H-with-4pt_numerical} below.
We also remark that we regard the assumption that $\{X_i\}_{i=1}^r$ be a Markov random field as minimal in the context of inverse problems, since it is a basic axiom in any treatment of causality, \eg, in the works of Pearl~\cite{10.5555/1642718} or Rubin~\cite{imbens_rubin_2015}.
In practice, it may be the case that the parent structure of a Markov random field $\{X_i\}_{i=1}^r$ is not \emph{a priori} known and is to be inferred from data.
This can be achieved by applying algorithms designed to estimate conditional independence amongst variables in a given system, from data.
These algorithms use parametric or non-parametric statistical methods, such as Pearson's $\chi$-squared test, to establish conditional independence amongst categorical random variables \cite{fastparallelpc,Kuipers2018EfficientSL,10.3389/fgene.2019.00524}. 

As an example of a structure discovery algorithm, the PC algorithm only scales exponentially in the \emph{worst} case scenario.
The sparser the ground truth network structure is, the faster the algorithm will converge.
In Ref.~\cite{fastparallelpc}, parallelised PC is benchmarked for constructing gene network neighbouring structures for yeast (5361 variables), a bacterium (2810 variables) and DREAM5-Insilico dataset (1643 variables).
The algorithm was shown to converge in less than 12 hours in all cases, on a personal computer with 8-cores. 
Once the graph structure is known or learned, estimating interactions scales as efficiently as computing averages over the data. The algorithm is therefore approximately as fast as estimating the bootstrap error on the interaction estimates. 

As a simple illustration, in Sec.~\ref{sec:cond-indept-numerical} we demonstrate the results of conditional independence tests on data generated by the $2$-dimensional Ising model, using the $\chi$-squared test, and discuss the improved statistics of the interaction estimates.

\vspace{-0.5cm}

\section{Results I: analytical map to regression and numerical results for the UK Biobank simulation} \label{sec:linear_regression_analytic}
As an elementary and concrete example, in this section we show that the non-parametric additive definition of interactions (Def.~\ref{def:interaction_add}) reduces to an interaction coefficient in a linear regression model. We illustrate this example in the context of a biomedical application. \\

\vspace{-0.5cm}
\subsection{Application: Interactions in biomedicine}
\label{sec:Interactions-in-biomedicine}
Genome-wide association studies (GWAS) are methods to identify genetic variants in the genome of individuals in a population, that could be associated with a disease or trait.
In case-control GWAS, one searches for variants, a collection of single nucleotide changes in the DNA, that occur more frequently in people with a particular disease (cases) as compared to those without the disease (controls).
The goal of GWAS is to find candidate genes that could potentially increase the risk of a certain disease, with the medical aim of identifying potential drug targets.
Currently, one of the main aims of this field of study is to move away from associational to causal variant-trait relations.
For the magnitude of causal effects of genomic variants on traits to be inferred accurately, one is required to (i) relax parametric assumptions such as the linear dependencies of the traits on the variants, and (ii) take into account interactions amongst the variants affecting traits, known as \emph{epistasis}.
In contrast to the methods used in some of the key literature in the field \cite{LIU20191022,Claussnitzer2020}, our definition of interaction via the Targeted Learning framework satisfies requirement (i) by removing the need for parametric assumptions altogether, and incorporates (ii) by taking into account epistatic interactions.

\subsection{Epistatic interactions}
Consider (i) a transcription factor protein which modifies gene expression by binding the DNA.
The degree of binding, however, depends on the underlying DNA variants to which the transcription factor is binding. 
Now suppose that (ii) there are multiple other variants across the genome that regulate the effect of another transcription factor protein, hence changing levels of gene expressions.
Then, (i) and (ii) have downstream interactions that affect particular traits or diseases in humans. 
As the considerations of genetics and causality are beyond the scope of this work, we limit ourselves here to a sample application of our techniques in extracting such epistatic interactions, using simulated data of trait and disease representative of the summary-level UK BioBank population \cite{Sudlow-ukbb}. 
We consider the case of a complex continuous trait, height, as an example. \\

There are many variants across the genome contributing a small fraction to a complex trait such as height; this is known as the omnigenic model~\cite{Boyle2017-omnigenic}.
Suppose that we have an a priori understanding of which genomic variants are relevant to consider, \eg, those in the vicinity of bone developmental genes. 
Consider the following linear ground truth, involving six variants, $\text{V}_j$ for $j=1,2,\ldots,6$, across the genome each contributing via a positive or negative coefficient to the value of height. Without loss of generality, suppose that only two of them also have a non-zero interaction (the generalisation to more interactions is trivial):
\begin{align}\label{eq:hight-model}
        \text{Height}^{(i)} \sim \alpha_0 + \sum_{j=1}^6 \alpha_j\cdot\text{V}_j^{(i)} + \gamma \cdot \text{V}_1 \cdot \text{V}_2 + \epsilon,
\end{align}
where $i$ represents an individual, $\epsilon$ is the noise in height and $\alpha_0$ corresponds to unobserved, but independent, variants contributing to height. \\

We use our model-agnostic non-parametric additive $2$-point interaction estimator $I^{a}_{1,2}$, Eq.~\ref{eq:interaction_add_2}, to show we recover the coefficient $\gamma$ representing the ground truth interaction between $V_1$ and $V_2$.
To see this, we simply compute the four expected outcomes in Eq.~\ref{eq:interaction_add_2}:
\small
\begin{align*}
    \BE(H \mid V_1 = 1, V_2 = 1, V_\text{3,4,5,6}=0) &= \alpha_0 + \alpha_1 + \alpha_2 + \gamma, \\
    \BE(H \mid V_1 = 1, V_2 = 0, V_\text{3,4,5,6}=0) &= \alpha_0 + \alpha_1, \\
    \BE(H \mid V_1 = 0, V_1 = 0, V_\text{3,4,5,6}=0) &= \alpha_0 + \alpha_2,  \\
    \BE(H \mid V_1 = 0, V_2 = 0, V_\text{3,4,5,6}=0) &= \alpha_0 
\end{align*}
\normalsize
We obtain the following expressions for the four average treatment effects:
    \begin{equation}
    \begin{split}
        \ATE_{V_1}(H \mid T_V = 1) &= \alpha_1 + \gamma,  \\ 
        \ATE_{V_1}(H \mid T_V = 0) &= \alpha_1 \\
        \ATE_{V_2}(H \mid T_V = 1) &= \alpha_2 + \gamma, \\ 
        \ATE_{V_2}(H \mid T_V = 0) &= \alpha_2.
    \end{split}
    \end{equation}
    The interactions both ways around are $I^{a}_{1,2} = \gamma = I^{a}_{2,1}$, as expected since interaction is symmetric by Cor.~\ref{cor:symmetry_interaction_add}.
    In conclusion, we have $I^{a}_{1,2} = \gamma$ as claimed.
Generalisations to higher-point interactions are trivial.
For a numerical example with 3-point interactions, see App.~\ref{sec:linear_regression_numerical}.

\subsection{Numerical simulations based on \\ the UK BioBank traits} \label{sec:numerics-UKBB}
 We generate data from the above ground truth, Eq.~\ref{eq:hight-model}.
The coefficients are chosen without loss of generality to reproduce a realistic distribution of heights which is representative of the UK BioBank population \cite{Sudlow-ukbb}, with approximately the same mean (168.5 cm) and standard deviation (9.3 cm) (\href{http://biobank.ndph.ox.ac.uk/showcase/field.cgi?id=50}{\texttt{UK BioBank, standing height}}).

The male and female populations are generated separately and merged to form the full distribution of height, consisting of 20,000 individuals, as presented in Fig.~\ref{fig:ukbb-height}.
More explicitly, WLOG, $\alpha_0=154$ for females and $\alpha_0=166$ for males, together with $\{\alpha_1, \cdots, \alpha_6\}=\{2,6,-3,6,-1.5,6\}$ with $\gamma=\epsilon=5$.
Notice that the 2-point interaction, $\gamma$, between the two aforementioned variants is chosen to approximately equal the level of noise in height across the population.
The variant allele frequencies for $\text{V}_1$, $\text{V}_2 \sim \Binom(0.8), \Binom(0.7)$ respectively, and for $\text{V}_3, \ldots, \text{V}_6 \sim \Binom(0.5)$.

\begin{figure}[!htb]
    \includegraphics[width=\linewidth]{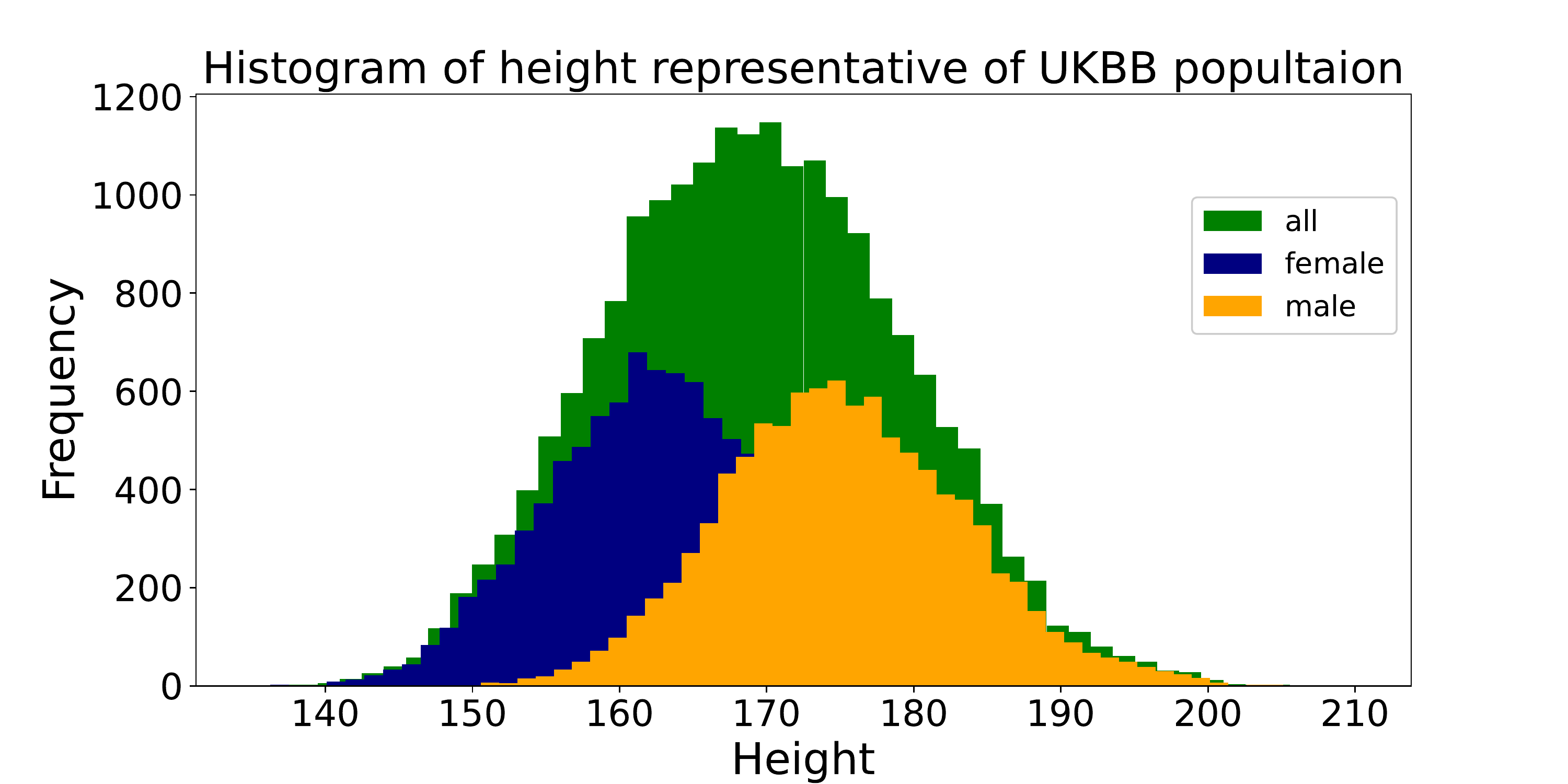}
    \caption{Histogram of female, male and combined heights on simulated data, such that it is representative of the UK BioBank population (\href{http://biobank.ndph.ox.ac.uk/showcase/field.cgi?id=50}{\texttt{UK BioBank, standing height}}). }
    \label{fig:ukbb-height}
\end{figure}

We apply the additive Targeted Learning estimator of interaction Eq.~\ref{eq:interaction_add} to the data.
We obtain the Targeted Learning prediction $\gamma = 4.77(1.36)$ which agrees with the ground truth value $\gamma=5$, within statistics.
\\

{\bf N.B.} Since the Targeted Learning (TL) estimator is non-parametric, it is completely agnostic to form, \eg, linearity or non-linearity, of the data generating process.
In particular, in the case of categorical variants, there is no biological basis for the linearity assumption often used in modelling variant-trait relations.
The above example merely serves to illustrate that \emph{if} the underlying truth were to be linear, then the TL estimator correctly recovers this linearity.
In fact, TL can be used to \emph{test} if the effect of variants on trait is linear. \\

The Targeted Learning estimator of epistatic interactions applies to all scenarios, be they linear, non-linear or non-monotonic, without requiring any parametric ansatz regarding the form of the fit function.
This generality is of crucial importance since transcription factors often consist of large protein complexes that can introduce highly non-trivial behaviour as well as other higher-order interactions. 
Such scenarios will be missed by standard linear parametric fits.
Using individual-level DNA variant and trait population data, our estimator's agnosticism and flexibility allows for new discoveries of novel and more complex interaction networks.

\section{Results II: analytical map and numerical results of the 2D Ising model and Restricted Boltzmann machines (RBM)}
\label{sec:results1}
In this section, we discuss Boltzmann probability distributions.
In Sec.~\ref{sec:Ising_model_analytic}, we recover the $2$-point couplings in an Ising Hamiltonian from the multiplicative formulation, Eq.~\ref{eq:interaction_mult_2}. 
In Sec.~\ref{sec:RBM_analytic}, we consider a more complex Hamiltonian: the Restricted Boltzmann Machine (RBM).
We analytically obtain its all-order couplings \emph{without} any need for an asymptotic expansion and resummation as originally employed in \cite{PhysRevB.100.064304}, using the \emph{same} universal multiplicative estimator, Eq.~\ref{eq:interaction_mult_2}.
In Sec.~\ref{sec:Ising_RBM_numerical}, we compare numerical results and finally, in Sec.~\ref{sec:cond-indept-numerical}, we evaluate the improvement in the numerical results when applying Markovian conditional independence criteria.

\subsection{Two-dimensional Ising model}
\label{sec:Ising_model_analytic}
We briefly recall the $2$-dimensional Ising model.{}
Consider a $2$-dimensional square lattice of size $L^2$ with periodic boundary conditions, with a spin $\tilde{v}_i$ on each lattice point $i$ taking on values $\tilde{v}_i = \pm 1$.
A \emph{state} of the Ising model is the assignment $\tilde{\bf v}$ of a value $+1$ or $-1$ to each of the $L^2$ spins.
Given a temperature $T$, the Boltzmann distribution describes the probability $p(\tilde{\bf v}|T)$ that the system takes on a particular state $\tilde{\bf v}$ at temperature $T$.
Explicitly,
\small
\begin{equation}\label{eq:H-ising}
    p(\tilde{\bf v}|T) = \frac{1}{\CZ(T)} e^{-E(\tilde{\bf v})} \quad \text{where }
        E(\tilde{\bf v}) = -\sum_{i,j} J_{i,j} \tilde{v}_i \tilde{v}_j,
\end{equation}
\normalsize
where the sum runs over all pairs of lattice sites $(i,j)$, where $J_{i,j}$ is the coupling between spins $\tilde{v}_i$ and $\tilde{v}_j$, the external magnetic field is zero, and $\CZ(T)$ is the partition function that normalises this probability distribution.

In the basic version of the Ising model, the interaction between non-nearest neighbour spins is put to zero, and $J_{i,j} = \frac{1}{2T}$ for all nearest neighbour spins $\tilde{v}_i, \tilde{v}_j$; this is not required in general.
However, $J_{i,j} = J_{j,i}$ is symmetric. \\

The \emph{inverse Ising problem} is concerned with estimating the coupling $J_{i,j}$ from data.
Our non-parametric definition Eq.~\ref{eq:interaction_mult_2} of multiplicative $2$-point interaction between the binary random variables $v_i$ and $v_j$ recovers the coupling coefficient $J_{i,j}$ directly from the probability distribution, after applying $\ln(-)/8$; the factor of $8$ is due to double counting as explained below.
To see this, we first apply the bijective transformation $\tilde{v}_i = 2v_i-1$ expressing the values of a spin $v_i$ in terms of $\{0,1\}$ as opposed to $\{-1,1\}$ in order to use our definition of multiplicative $2$-point interaction Eq.~\ref{eq:interaction_mult_2}.
Thus, $\tilde{v}_i = -1$ corresponds to $v_i = 0$, whereas $\tilde{v}_i = 1$ corresponds to $v_i = 1$.
The energy function corresponds to
\small
\begin{equation*}
    E({\bf v}) = - 4 \sum_{i,j} J_{i,j} v_i v_j + 4 \sum_{i} \left( \sum_{j} J_{i,j} \right)v_i - \left( \sum_{i,j} J_{i,j} \right),
\end{equation*}
\normalsize
where we have used the symmetry $J_{i,j} = J_{j,i}$.
    
Next, we compute the multiplicative $2$-point interaction $I^m_{i,j}$ between two spins.
Without loss of generality, we do this for spins $v_1$ and $v_2$.
We compute the probabilities that $(v_1,v_2)$ takes on the values $\{(1,1),(1,0),(0,1),(0,0)\}$ with all other spins being zero, \ie, $\underline{v} = 0$.
We find
\begin{align}
    \frac{p(1,1,\underline{v} = 0)}{p(1,0,\underline{v} = 0)}
    &= \exp\biggl(4J_{1,2} + 4J_{2,1} - 4\sum_{j\neq 1} J_{1,j}\biggr) \\
    \frac{p(0,0,\underline{v} = 0)}{p(0,1,\underline{v} = 0)}
    &= \exp\biggl(4\sum_{j\neq 1} J_{1,j}\biggr),
\end{align}
and multiplying both yields $I^m_{1,2} = \exp(8J_{1,2})$.
Hence $\ln(I^m_{1,2})/8 = J_{1,2}$ as claimed.

Whether or not $I^m_{1,2}$ is smaller or larger than $1$ is due to the interpretation of the interaction.
In this case, it is the $2$-point interaction between \emph{turning on} both spins, \ie, $v_1 \colon 0 \to 1$ and $v_2 \colon 0 \to 1$, not turning them off.
Alternatively, computing the additive interaction between $v_1 \colon 0 \to 1$ and $v_2 \colon 0 \to 1$ on the \emph{outcome} $-E({\bf v})$ is easily seen to be $I^a_{1,2} = 8J_{1,2}$.
The factor of $8$ is due to the change of variables $\tilde{v}_i \mapsto v_i$ and a double counting in Eq.~\ref{eq:H-ising}.
Finally, the coupling $J_{i,j}$ can be obtained directly by taking the double derivative of the outcome $-E({\bf v})$ with respect to $v_1$ and $v_2$.
\\

In Sec.~\ref{sec:Ising_RBM_numerical}, we extract $J_{i,j}$ directly from data.
In order to improve the estimate of the $2$-point interaction $I^{m}_{i,j}$ from data, one may appeal to the Hammersley--Clifford Theorem of Sec.~\ref{sec:conditional_independence} to increase statistics by only conditioning on the relevant \emph{parent} variables, \ie, in this case the nearest neighbours of $v_i$ and $v_j$.
In fact, the Monte Carlo algorithm, \eg, Metropolis, generating Ising configurations uses this feature in its update step by computing the change in energy only using nearest neighbour spins.
For completeness, we analytically demonstrate that the Hammersley--Clifford Theorem applies to the Ising model in App.~\ref{app:Hammersley_Clifford}.

\subsection{Restricted Boltzmann Machine}
\label{sec:RBM_analytic}
A Restricted Boltzmann Machine (RBM) is a type of undirected Markov random field (MRF) with a two layer architecture.
An RBM consists of $m$ visible nodes $v_j$, $j \in \{1,\ldots,m\}$, collectively denoted by ${\bf v}$ and representing the observed input data, and $n$ hidden nodes $h_i$, $i \in \{1,\ldots,n\}$, collectively denoted by ${\bf h}$.
We consider binary variables, i.e. $v_j, h_i \in \{0,1\}$.
The energy of the joint state $\{\bf v, \bf h\}$ of the machine is as follows:
\begin{equation}
    E({\bf v},{\bf h};\theta) = - \sum_{i=1}^n \sum_{j=1}^m h_i w_{ij} v_j - \sum_{j=1}^m b_j v_j - \sum_{i=1}^n c_i h_i,
\end{equation}
and we collectively call $\theta = \{\bf w,b,c\}$ the model parameters.
The RBM is used to encode the joint conditional probability distribution of a state $\{\bf v, \bf h\}$ given a set of parameters $\theta$:
\begin{equation}\label{eq:RBM_probability_visible}
    p({\bf v,h}|\theta) = \frac{1}{\CZ(\theta)} e^{-E(\bf{v,h};\theta)},
\end{equation}
where the partition function $\CZ(\theta)$ normalises the probability distribution.
Marginalising over the binary hidden variables $h_i$ yields the probability distribution of the variables in the visible layer \cite{FISCHER201425}:
\begin{equation}
    p({\bf v}|\theta) = \frac{1}{\CZ(\theta)} \prod_{j=1}^m \left( e^{b_j v_j} \right) \prod_{i=1}^n \left(1 + e^{c_i + \sum_{j=1}^m w_{ij} v_j} \right).
\end{equation}
By equating the RBM energy function to the $2$-dimensional Ising energy function, the expression
\begin{equation}
    J_{j_1,j_2} = \frac{1}{8} \ln \prod_{i = 1}^n \frac{(1+e^{c_i + w_{ij_1} + w_{ij_2}})(1+e^{c_i})}{(1+e^{c_i+w_{ij_1}})(1+e^{c_i+w_{ij_2}})}
\end{equation}
is obtained in~\cite{PhysRevB.100.064304}.
This expresses the Ising coupling $J_{j_1,j_2}$ in terms of the model parameters of the RBM.
The proof uses an asymptotic expansion and a resummation.
Computing the non-parametric $2$-point interaction, as in Eq.~\ref{eq:interaction_mult_2}, of the RBM readily yields the above formula:
\begin{equation}
    \frac{1}{8} \ln \bigl(I^m_{j_1,j_2}\bigr) = J_{j_1,j_2},
\end{equation}
where $I^m_{j_1,j_2}$ is computed from equation Eq.~\ref{eq:RBM_probability_visible}.
Indeed, this follows from Eq.~\ref{eq:interaction_mult_2} by a direct computation, since
\small
\begin{equation*}
\begin{split}
    I^m_{j_1,j_2}= 
    &\frac{p(v_{j_1 j_2} = (1,1),\underline{v} = 0)}{p(v_{j_1 j_2} = (1,0),\underline{v} = 0)} \frac{p(v_{j_1 j_2} = (0,0),\underline{v} = 0)}{p(v_{j_1 j_2} = (0,1),\underline{v} = 0)} \\
    = &\prod_{i = 1}^n \frac{(1+e^{c_i + w_{ij_1} + w_{ij_2}})(1+e^{c_i})}{(1+e^{c_i+w_{ij_1}})(1+e^{c_i+w_{ij_2}})}.
\end{split}
\end{equation*}
\normalsize
Indeed, both the partition functions and the $b_j$ coefficients cancel out.
By the same argument, one immediately recovers the closed form expression for the $3$-point interaction between $v_{j_1},v_{j_2},v_{j_3}$ as derived in~\cite[Eq.~(66)]{PhysRevB.100.064304}, and the closed form expressions for all $n$-point interactions, without having to resolve to an asymptotic expansion and resummation as in \cite{PhysRevB.100.064304}.

\subsection{Numerical results for the Ising model and comparisons with the RBM}
\label{sec:Ising_RBM_numerical}
In this section, we generate $2$-dimensional Ising configurations at various values of temperature using {\tt Magneto} \cite{Magneto}, a fast parallel C++ Monte Carlo code available online.
We set $J_{ij} = 1/2T$ in Eq.~\ref{eq:H-ising}.
We then use the non-parametric multiplicative definition of interactions, Sec.~\ref{sec:mult_interaction}, to extract the couplings $J_{ij}$ directly from the data, \ie, we solve the inverse problem.
We demonstrate agreement with the ground truth and compare the performance of the estimation of interactions directly from the data with the estimates obtained via the RBM \cite{PhysRevB.100.064304}. 
Ising states generated by {\tt Magneto} consist of spins $\pm 1$.
Note that these are converted to 0, 1 as input to both the multiplicative interaction formulation and the RBM, as already discussed in Sec.~\ref{sec:Ising_model_analytic}. Before delving into the numerical analysis, our main results are summarised in the paragraph below. \\

In general, the non-parametric interaction converges to the true value in the infinite data limit as it is unbiased, whereas the RBM need not do so as the original data is almost surely not generated from an RBM distribution.
However, for finite sample sizes, the direct computation may become noisy and unstable without additional information, such as conditional independence amongst the variables. 
Take, for example, the case of the Ising configuration in different temperature regimes.
At low temperatures the system is highly coupled and symmetric with respect to configurations mostly containing spin zeros and those mostly containing spin ones. 
In this regime, there are enough samples to estimate conditional probabilities appearing in Eq.~\ref{eq:interaction_mult_2}.
On the other hand, it is harder to train an RBM in highly coupled systems, \eg, in \cite{PhysRevB.100.064304} more precise hyperparameter tuning and longer training was required.
This behaviour of the RBMs has been reported previously in the literature \cite{FISCHER201425} and is due to the machine remaining in local minima of the activation function.
To avoid this problem, the RBM needs to be trained using more advanced algorithms such as Parallel Tempering \cite{FISCHER201425} which allows the machine to exit potential local minima.
Of course, this in turn requires tuning of extra hyperparameters and results in longer training times.
For temperatures above the critical temperature, the system becomes weakly coupled and moves towards more randomly distributed zero and one spin configurations. In this scenario, conditioning on all but two variables in the system results in very low sample sizes and unstable estimates of the interactions unless the total sample size is very large. 
The RBM, on the other hand, captures the interactions well given a comparable sample size. If however, information about conditional independence amongst the variables in the system is used, the non-parametric estimates perform better than the RBM in terms of bias, variance and compute time. 
In what follows, we quantify the above statements explicitly.\\

Before we present numerical results, we note that excluding higher-order interaction terms from the outset necessarily results in biased or incorrect estimates of even the 2-point and self-couplings. To give a simple example, consider the following formula;
\begin{align}
	E 	&= E_0 + h_1v_1 + h_2v_2 + J_{12}v_1v_2 + J_{123}v_1v_2v_3 \\ \nonumber
		&= E_0 + h_1v_1 + h_2v_2 + \Bigl(J_{12}+J_{123}v_3 \Bigr)v_1v_2.
\end{align}
Thus, any parametric fit ignoring third order (and higher) interactions will incorrectly report $J_{12}+J_{123}\mathbb{E}(v_3)$ as the $2$-point interaction.
More disturbingly, in a situation where the ground truth satisfies $J_{12} = 0$ but $J_{123} \neq 0$, a truncated parametric fit will incorrectly produce the non-existent $2$-point interaction $J_{123}\mathbb{E}(v_3)$.
Our method avoids this problem entirely. \\

Using the TL universal estimator, Eq.~\ref{eq:interaction_mult_2} directly, it is possible to obtain an accurate estimate of the couplings at cold temperatures, \emph{without} conditioning on the Markovian parents or using translational invariance. Unlike Refs.~\cite{Nguyen2017,PhysRevLett.112.070603,PhysRevLett.108.090201,ravikumar2010,published_papers/7111360} no parametric assumptions, regularisation, truncation of higher-order interactions or other approximations are required. The results are shown in Fig.~\ref{fig:ising_two_point_noMarkovian_noTranslational}. 

\begin{figure}[!htb]
\begin{center}
    \minipage{0.50\textwidth}
	\includegraphics[width=\linewidth]{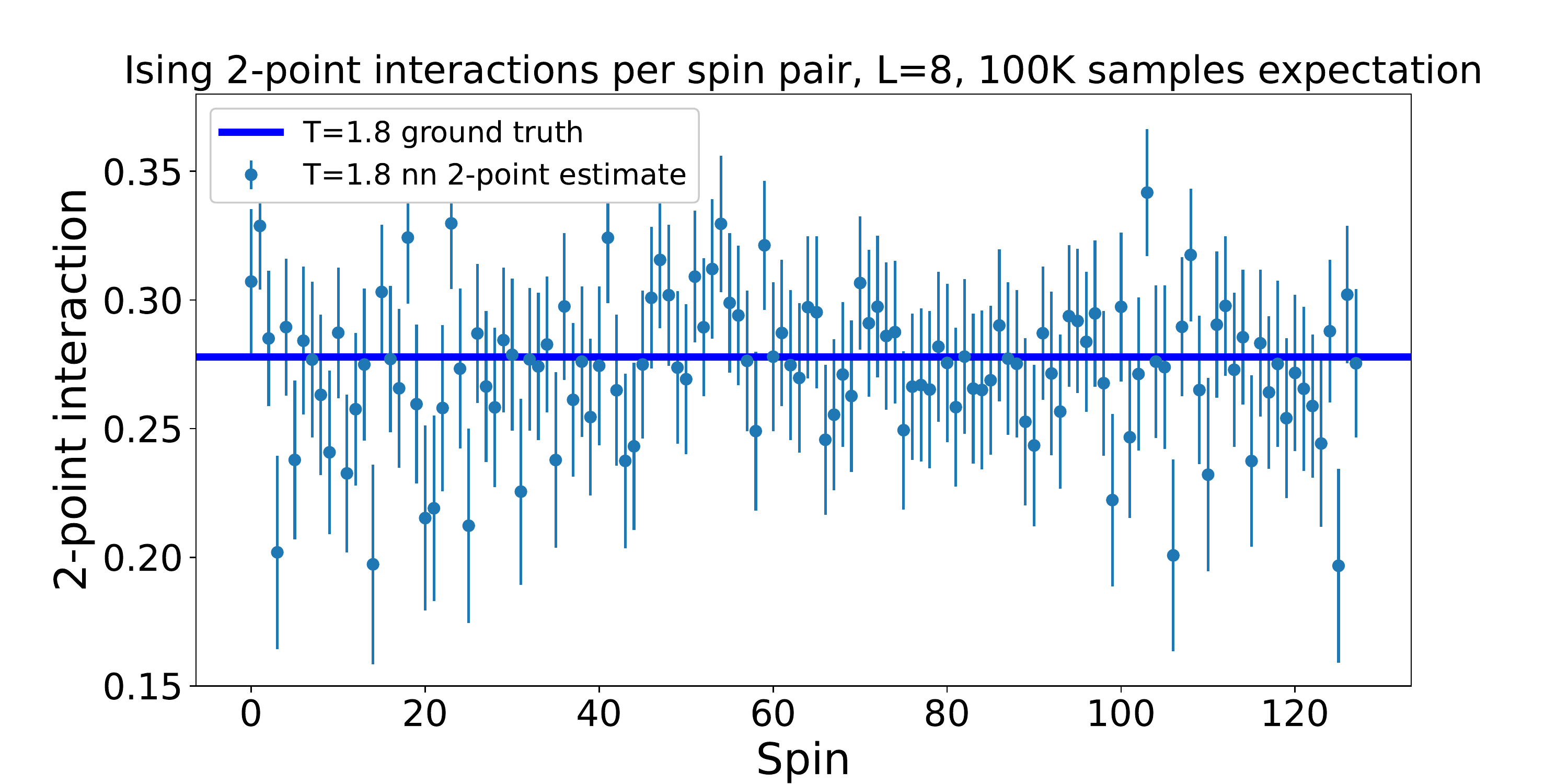}
	\endminipage\hfill
    \caption{All non-zero $2$-point interaction estimates using Eq.~\ref{eq:interaction_mult_2} directly, at temperature $T = 1.8$, in an Ising system of size $L^2 = 8^2$ with periodic boundary conditions. 100K samples are used for this estimation. No conditioning on the Markovian parents is performed, no translational invariance assumptions are made.}
    \label{fig:ising_two_point_noMarkovian_noTranslational}
\end{center}
\end{figure}

Above the critical temperature, however, TL estimation requires larger samples sizes. More explicitly, beyond $T=2.4$, the states become more random, and conditioning on all $v_i$'s to be zero, apart from the two spins whose interaction is to be estimated, results in low sample sizes and unstable predictions of the conditional probabilities appearing in Eq.~\ref{eq:interaction_mult_2}.
This is demonstrated by plotting the bin sizes used to estimate the probabilities at various values of temperature in Fig.~\ref{fig:ising_two_point_bin_vs_temp}. 

\begin{figure}[!htb]
\begin{center}
    \minipage{0.5\textwidth}
	\includegraphics[width=\linewidth]{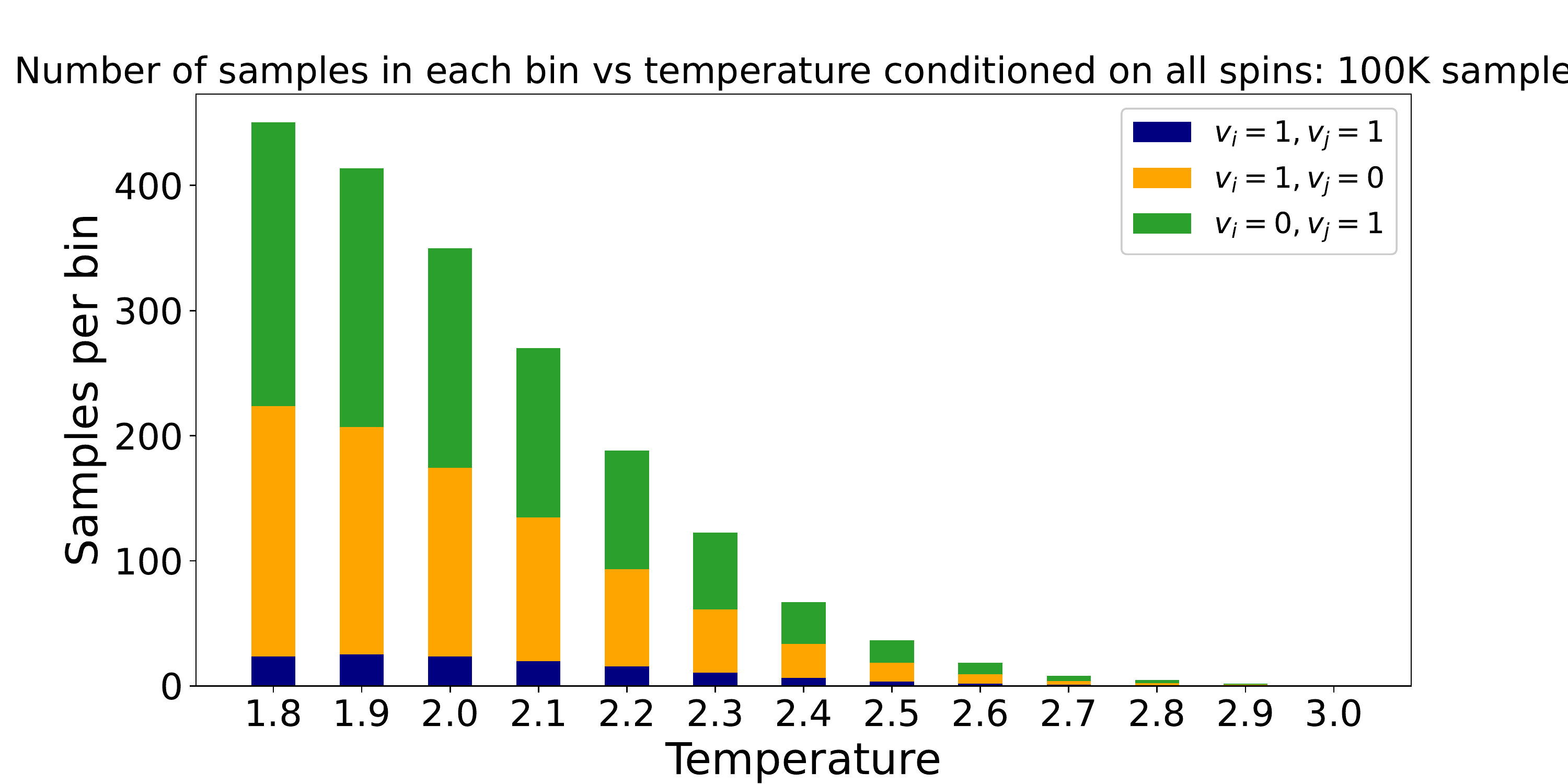}
	\endminipage\hfill
	\minipage{0.5\textwidth}
	\includegraphics[width=\linewidth]{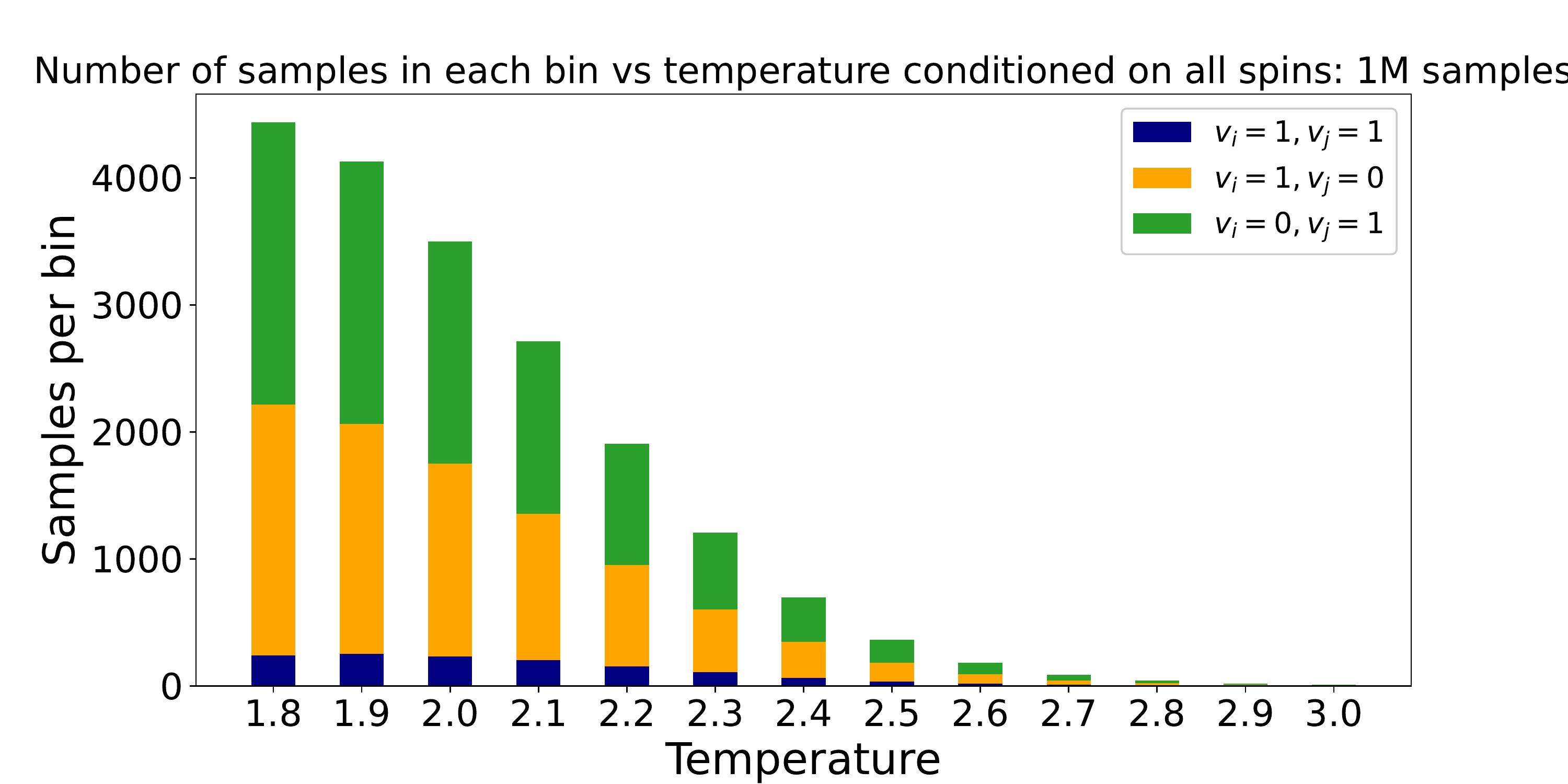}
	\endminipage\hfill
    \caption{Average sample sizes for conditional probabilities entering the computation of the 2-point interaction for the nearest neighbour pairs in an $L^2=8^2$ lattice. These values are obtained by conditioning on all other spins. The bin $v_i=v_j=0$ is left out as it has the largest size as compared to the other three. The top plot is from 100K samples, and the bottom is from 1M samples. Notice that each of the bin sizes increases 10-fold as we go from 100K to 1M samples, as expected. Observing the 100K plot, it is clear that above $T=2.6$, there are not enough samples in the $v_i=v_j=1$ bin to yield reliable estimates of the interactions, with $T=2.6$ containing approximately 9 samples on average. With 1M total samples, one can obtain estimates for $T=2.7$, which on average contain 10 samples in the $v_i=v_j=1$ bin respectively. Beyond this temperature, one has to again increase the sample size to 2M or more. }
    \label{fig:ising_two_point_bin_vs_temp}
\end{center}
\end{figure}

Note that, as mentioned earlier, the non-parametric approach of estimating coupling from the data is an unbiased estimator and only limited by the amount of data. Therefore, larger samples sizes are required, if one wishes to make no physical approximation or further assumptions about, \eg, conditional independence amongst the variables. Fig.~\ref{fig:ising_two_point_interaction_vs_temp} indicates this requirement: Above the critical temperatures, the sample sizes need to be increased from 100K to 1M and 10M, at very hot temperatures, in order to estimate the couplings. As expected, in Fig.~\ref{fig:ising_two_point_interaction_vs_temp} the estimates converge to the theoretical ground truth when the samples sizes are sufficiently increased. Note that translational invariance is not a requirement and is merely used as a summary to illustrate convergence of the non-zero couplings to the correct ground truth value.

\begin{figure}[!htb]
\begin{center}
    \minipage{0.50\textwidth}
	\includegraphics[width=\linewidth]{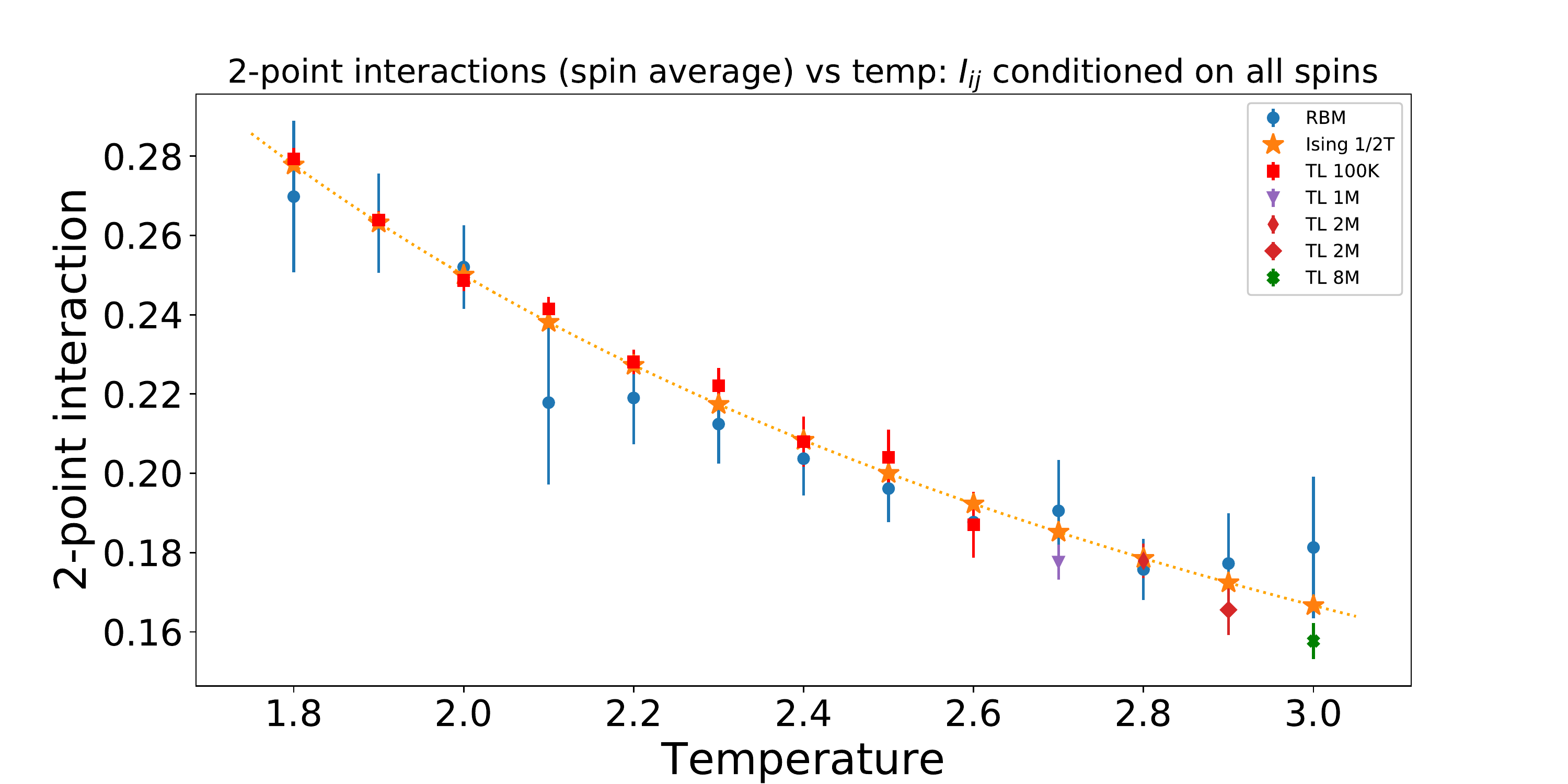}
	\endminipage\hfill
    \caption{A comparison of estimates of the $2$-point interaction amongst nearest neighbour spins as the temperature varies, in an Ising system of size $L^2 = 8^2$ with periodic boundary conditions, averaged over all 128 pairs of nearest neighbours for summary illustration.
    Each point represents a bootstrap average with error bar given by the bootstrap error. For $T\leq2.6$, 100K samples are enough to estimate the nearest neighbour interactions. For $T>2.6$ substantially more samples are required for stable estimates of the interactions. At $T=3.0$, 8M samples are required for a stable estimate.}
    \label{fig:ising_two_point_interaction_vs_temp}
\end{center}
\end{figure}


We now demonstrate improvements in the estimates of interactions at all values of temperatures, by using information on conditional independence  amongst the spins.
This allows for a substantial reduction in the sample sizes required, especially at high temperature.
As discussed earlier in Sec.~\ref{sec:conditional_independence} and will be further explained in Sec.~\ref{sec:cond-indept-numerical}, to obtain correct estimates of interaction amongst spins of interest, it is sufficient to condition on their parents, \ie, nearest neighbour spins, as opposed to all other spins in the rest of the lattice.
For interactions between pairs of nearest neighbour spins, we condition on their 6 nearest neighbours, while for interactions between pairs of non-nearest neighbour spins we condition on their 4+4 nearest neighbour spins. \\

The \emph{individual} per spin pair results, \emph{without} using translational averaging, for $T=1.8, 2.2, 3.0$ are shown in Fig.~\ref{fig:L8_T18_T22_T30_2pt_conditioned_PerSpin_expectation_100K}. Individual vanishing per spin triplet and quadruplet 3- and 4-point interactions are presented in App.~\ref{app:spin_pair_exp_L8_T18_T22_T30}, Fig.~\ref{fig:L8_T18_per_spin_3pt_4pt_cond_200K} with $T=1.8$ as an example. 
Fig.~\ref{fig:cond_ising_two_point_bin_vs_temp} indicates an increase in the smallest bin size, \ie, $v_i=v_j=1$, at all temperatures. 
This results in more precise estimates for the couplings, presented in Fig.~\ref{fig:cond_ising_two_point_interaction_vs_temp_100K}\footnote{All run times are measured on a MacBook Pro (2018) machine, 6-Core Intel i9 with 16GB memory.}, by using translational invariance. Again, note that translational invariance used in Fig.~\ref{fig:cond_ising_two_point_interaction_vs_temp_100K} is not a requirement and is merely used as a summary for comparison with the RBM results in \cite{PhysRevB.100.064304}.

\begin{figure}[!htb]
\begin{center}
    \minipage{0.50\textwidth}
	\includegraphics[width=\linewidth]{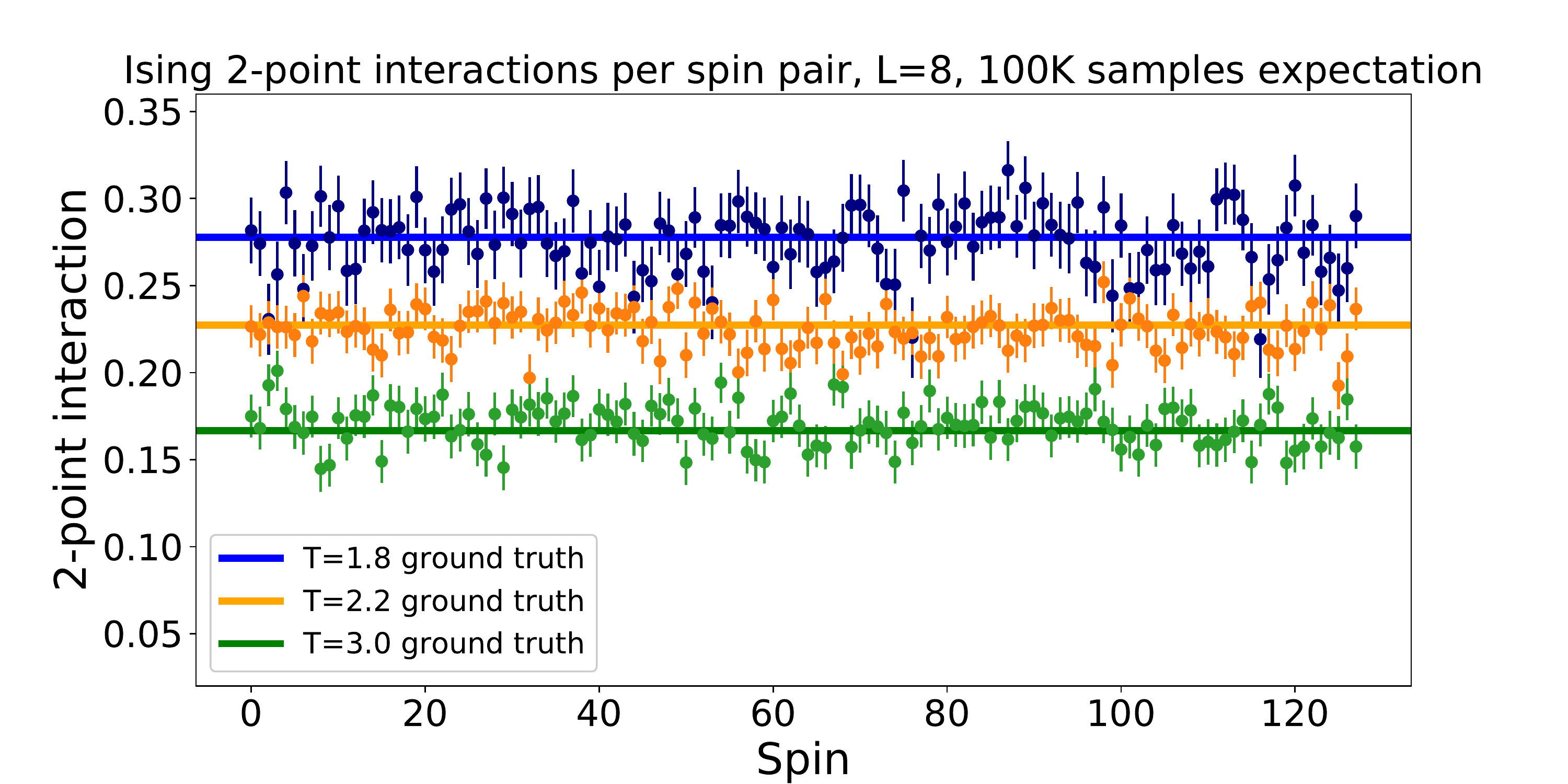}
	\endminipage\hfill
    \caption{Conditioning on the nearest neighbours (as prior information) to estimate $I_{ij}^m$ substantially improves the estimates. 100K samples for estimations at $T=1.8, 2.2, 3.0$, $L^2=8^2$.}
    \label{fig:L8_T18_T22_T30_2pt_conditioned_PerSpin_expectation_100K}
\end{center}
\end{figure}

\begin{figure}[!htb]
\begin{center}
    \minipage{0.45\textwidth}
	\includegraphics[width=\linewidth]{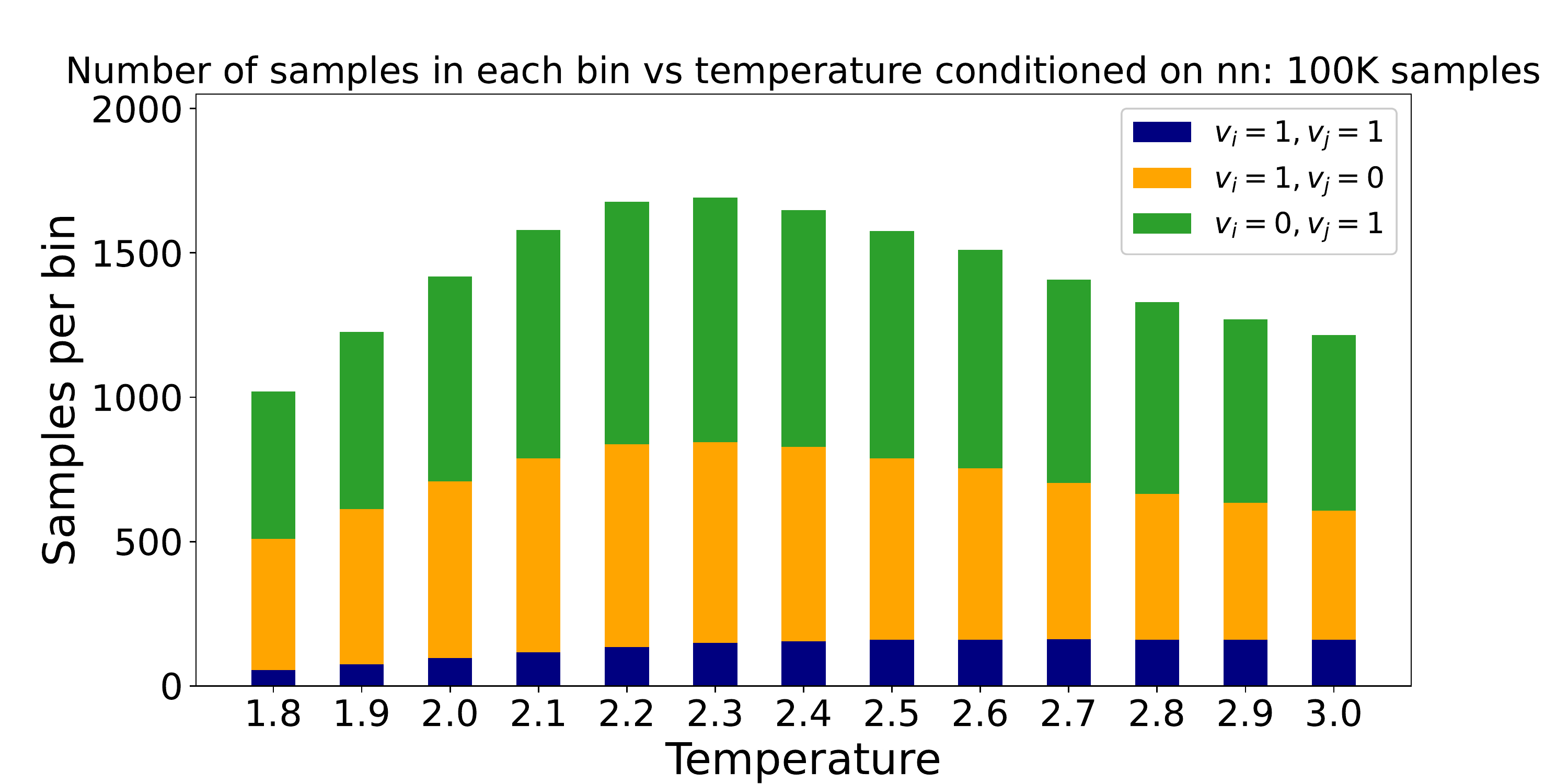}
	\endminipage\hfill
    \caption{Average sample sizes for conditional probabilities entering the computation of the 2-point interaction for the nearest neighbour pairs in an $L^2=8^2$ lattice.  These values are obtained by conditioning on the nearest neighbour spins only. The bin $v_i=v_j=0$ is left out as it has the largest size as compared to the other three. There are enough samples in each bin to yield stable estimates of each conditional probability/expectation value.}
    \label{fig:cond_ising_two_point_bin_vs_temp}
\end{center}
\end{figure}

\begin{figure}[!htb]
\begin{center}
    \minipage{0.50\textwidth}
	\includegraphics[width=\linewidth]{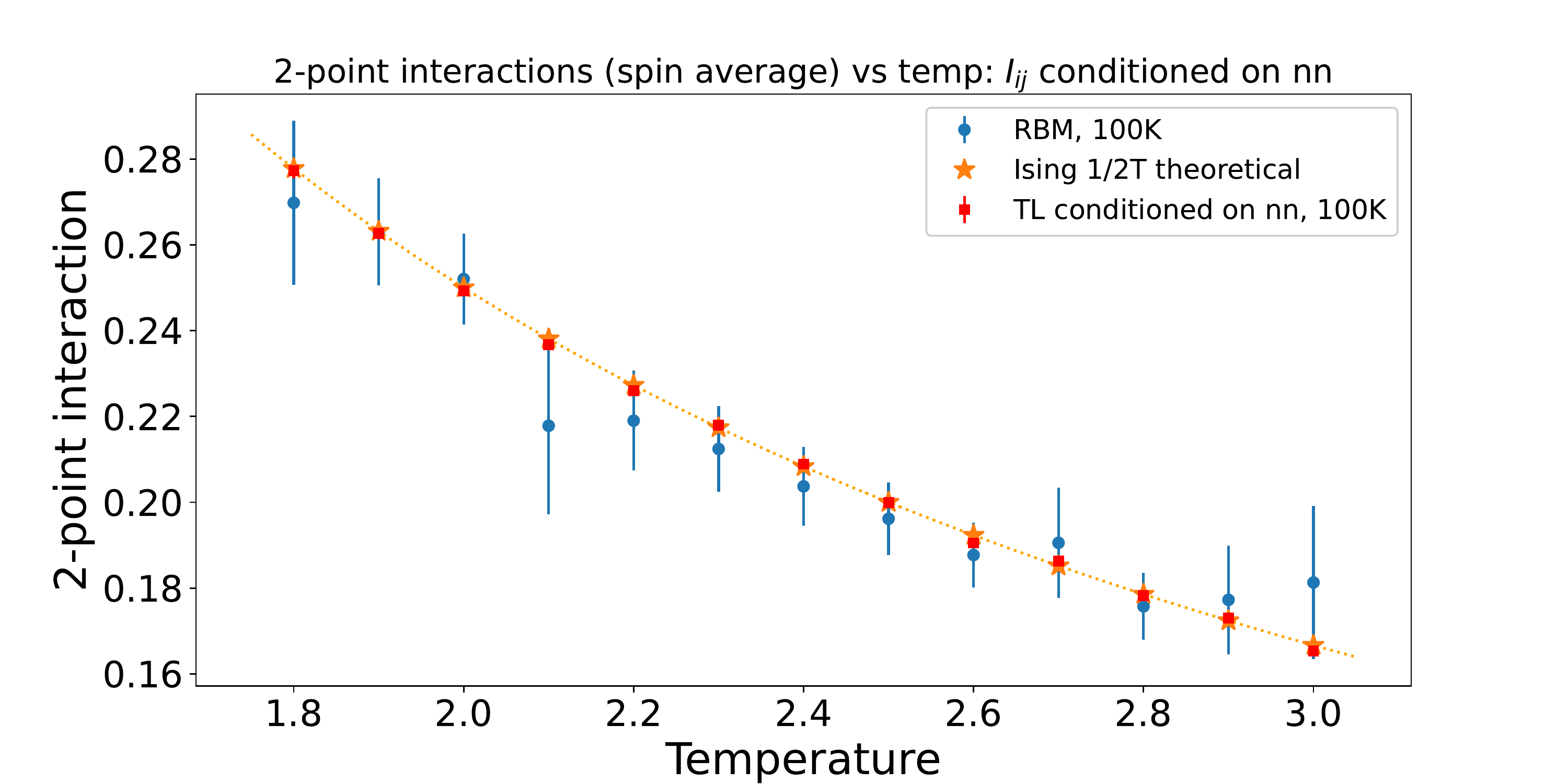}
	\endminipage\hfill
    \caption{Conditioning on the nearest neighbours (as prior information) to estimate $I_{ij}^m$ substantially improves the estimates as compared to Fig.~\ref{fig:ising_two_point_interaction_vs_temp}. 100K samples are used for both training the RBM and estimating the interactions directly using TL. See Fig.~\ref{fig:L8_T23_per_spin_indept_and_expect_10K} for the successful estimation of interactions and their uncertainty using TL, with 10K samples.
    The run time for each estimation using TL is at the order of a few seconds.}
    \label{fig:cond_ising_two_point_interaction_vs_temp_100K}
\end{center}
\end{figure}


Fig.~\ref{fig:L8_T30_per_spin_indept_and_expect} (upper), indicates individual spin pair couplings $I_{ij}^m$, estimated using Eq.~\ref{eq:interaction_mult_2_expectation} over 100K samples as compared to 20K (lower) for both nearest and non-nearest neighbour spins.
The latter results are more noisy as expected.
As compared to the 100K, 20K total samples approximately had 2\% of spin pairs with no samples in the $p_{11}$ bin.
This is due to the fact that it is unlikely that 2 spins having value one, whilst their 8 nearest neighbours all have spin zero.
This scenario is observed more often at colder temperatures, see Figs.~\ref{fig:L8_T18_per_spin_indept_and_expect},~\ref{fig:L8_T22_per_spin_indept_and_expect} in App.~\ref{app:spin_pair_exp_L8_T18_T22_T30}.
Note that the non-parametric method of estimation, combined with information on conditional independence amongst the variables, has nevertheless enabled us to obtain accurate estimates of the interactions relying on a smaller number of samples in total.
For example, using this method, there is enough power to estimate all the nearest neighbour spin pair interactions and approximately 83\% of the non-nearest neighbour spin pair interactions for temperature $T=2.2$ using 10K sample only, as demonstrated in Fig.~\ref{fig:L8_T30_per_spin_indept_and_expect}.
In contrast, \eg, the RBM does not train well on Ising data with 10K samples, see~\cite[Fig.~31]{PhysRevB.100.064304}, and therefore is not able to provide accurate estimates of the interactions at low sample sizes.
\\

\begin{figure}[!htb]
\begin{center}
    \minipage{0.45\textwidth}
	\includegraphics[width=\linewidth]{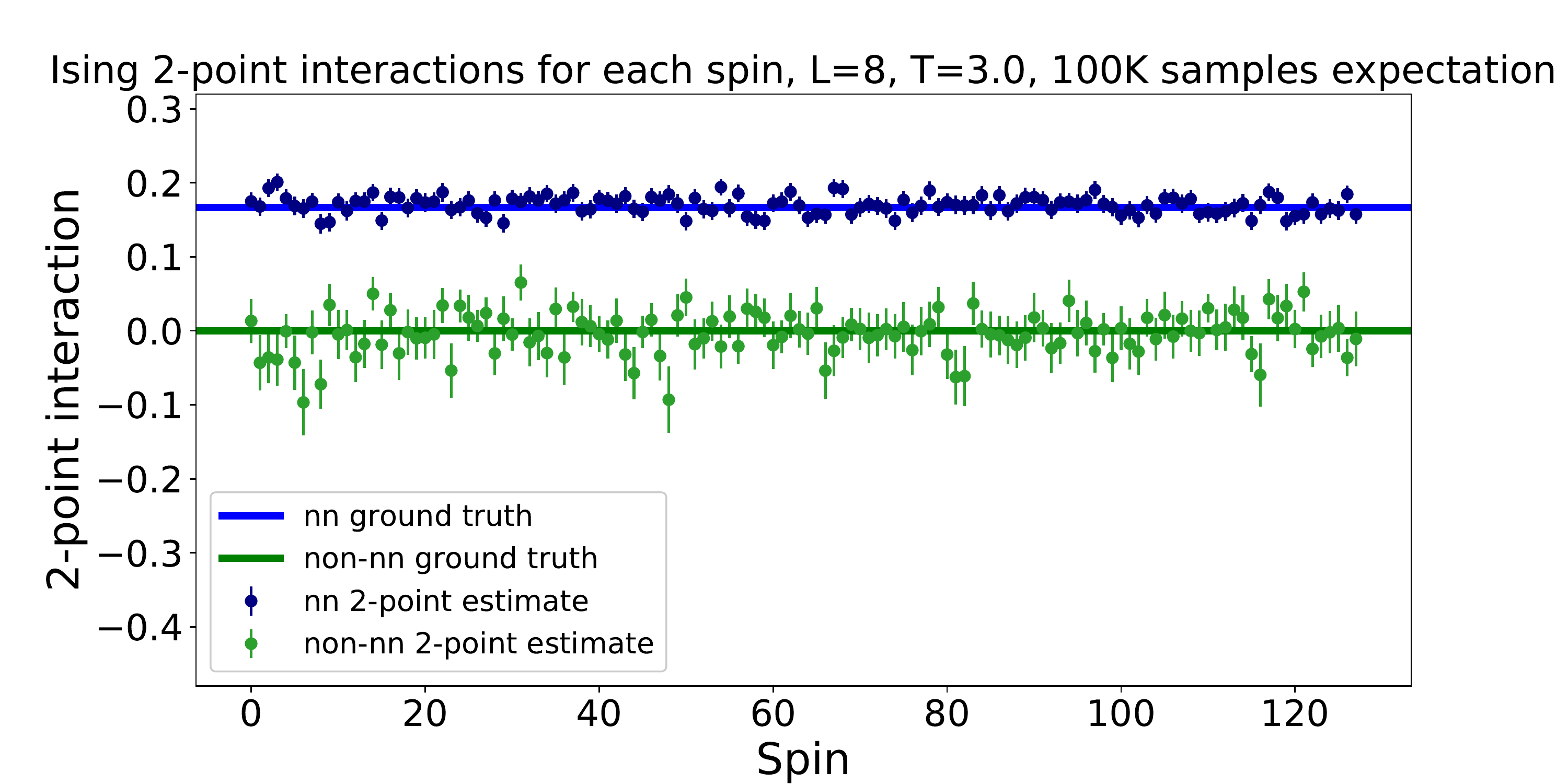}
	\endminipage\hfill
	\minipage{0.45\textwidth}
	\includegraphics[width=\linewidth]{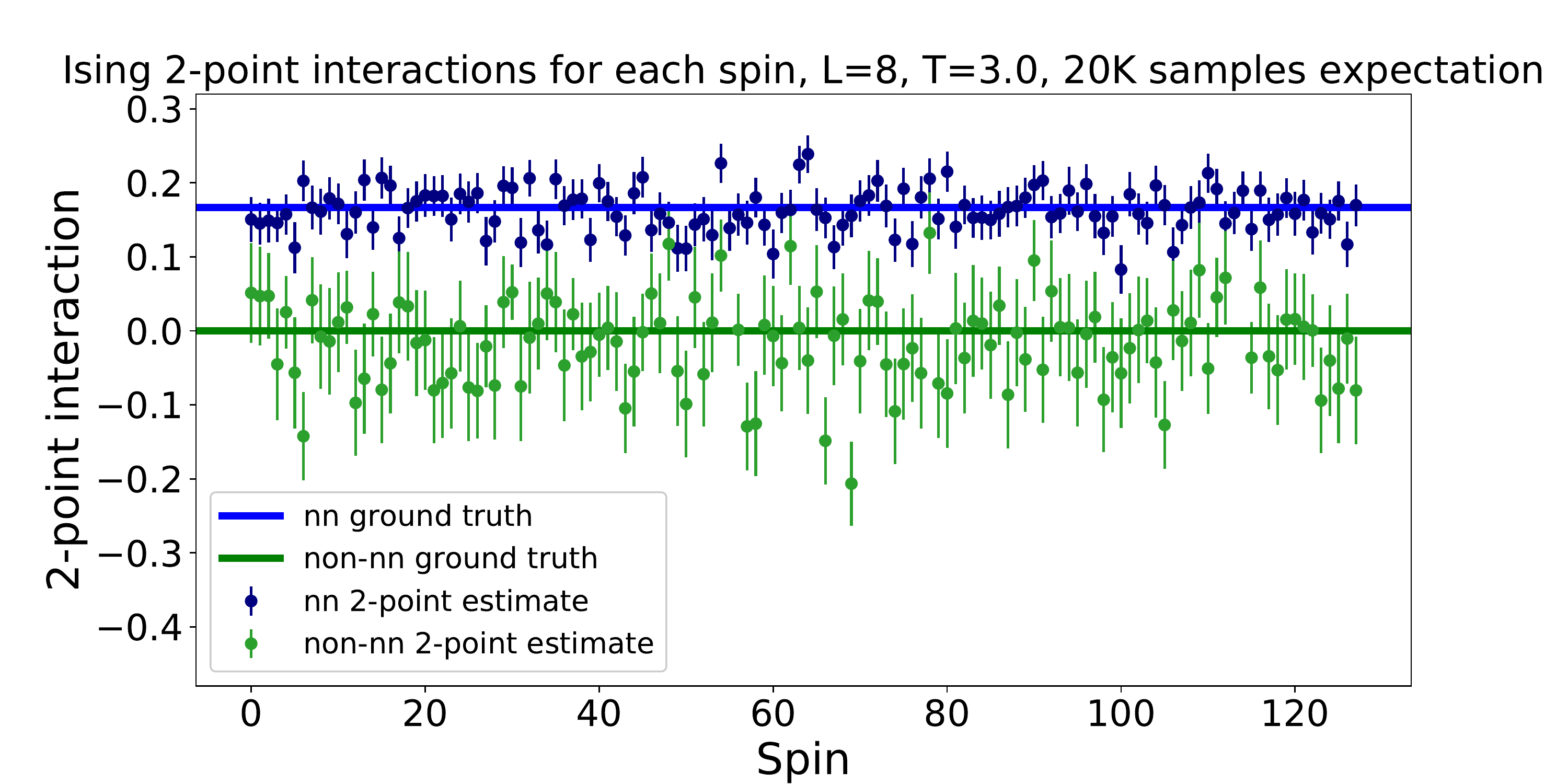}
	\endminipage\hfill
    \caption{$L^2=8^2$, $T=3.0$, with conditioning on the nearest neighbours to estimate $I_{ij}^m$ for both nearest and non-nearest neighbour spin pairs. In order to reduce clutter, the same number of non-nearest as nearest neighbour couplings are shown (128). No translational invariance is used. Top: The results are computed over a total of 100K samples, using Eq.~\ref{eq:interaction_mult_2_expectation} and statistical bootstrap, as compared to bottom: The results are computed over a total of 20K samples. For the latter, approximately 2\% of spins had no samples in the $p_{11}$ bin. This is because it is unlikely that 2 spins have value one, whilst their 8 nearest neighbours all have spin value zero, as the total sample size reduces. }
    \label{fig:L8_T30_per_spin_indept_and_expect}
\end{center}
\end{figure}

Finally, we present the results of estimating the 2-point interactions per individual spin pair, for a $L^2=32^2$ lattice at temperature $T=3.0$, in Fig.~\ref{fig:L32_T30_per_spin_indept_and_expect}.
As expected, the results for the case of 20K total samples is more noisy, however, the signal is clearly distinguishable from background with most of the nearest-neighbour interactions being more than $3\sigma$ away from the zero line.
We note that training an RBM on a lattice of this size, if possible, is expected to be computationally expensive and not possible for low numbers of sample sizes.
This is due to the fact that a $L^2=32^2$ lattice contains 1024 spins which would correspond to an RBM with $1024\times1024 \text{ weights } +2\times1024 \text{ bias terms }$, \ie, 1,050,624 parameters to be determined, when the number of hidden nodes (1024) is set equal to the visible nodes (1024).
The run time of the non-parametric approach is of the order of minutes on a local computer.  

\begin{figure}[!htb]
\begin{center}
    \minipage{0.45\textwidth}
	\includegraphics[width=\linewidth]{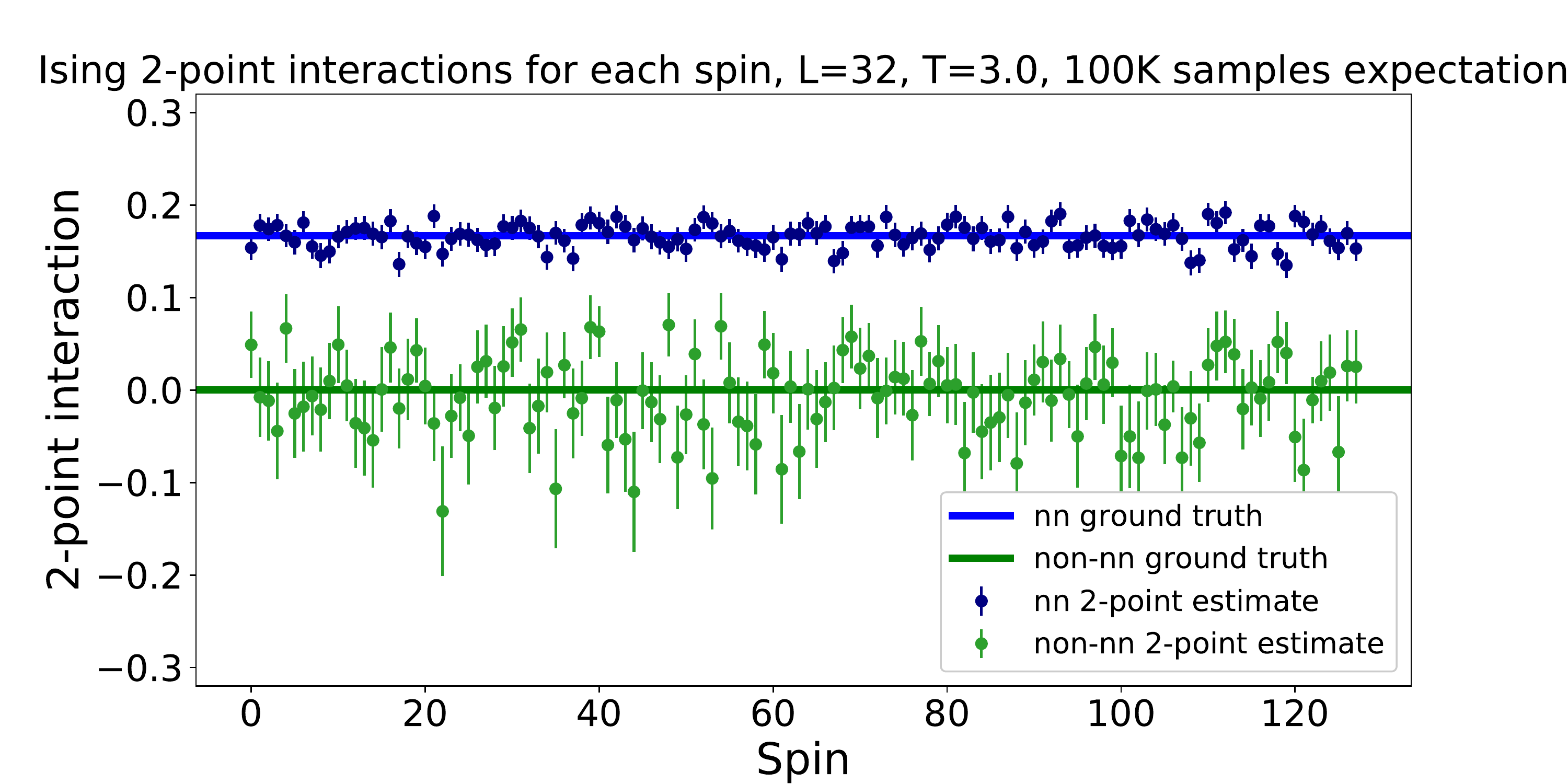}
	\endminipage
	\vspace{0.2cm}
	\minipage{0.45\textwidth}
	\includegraphics[width=\linewidth]{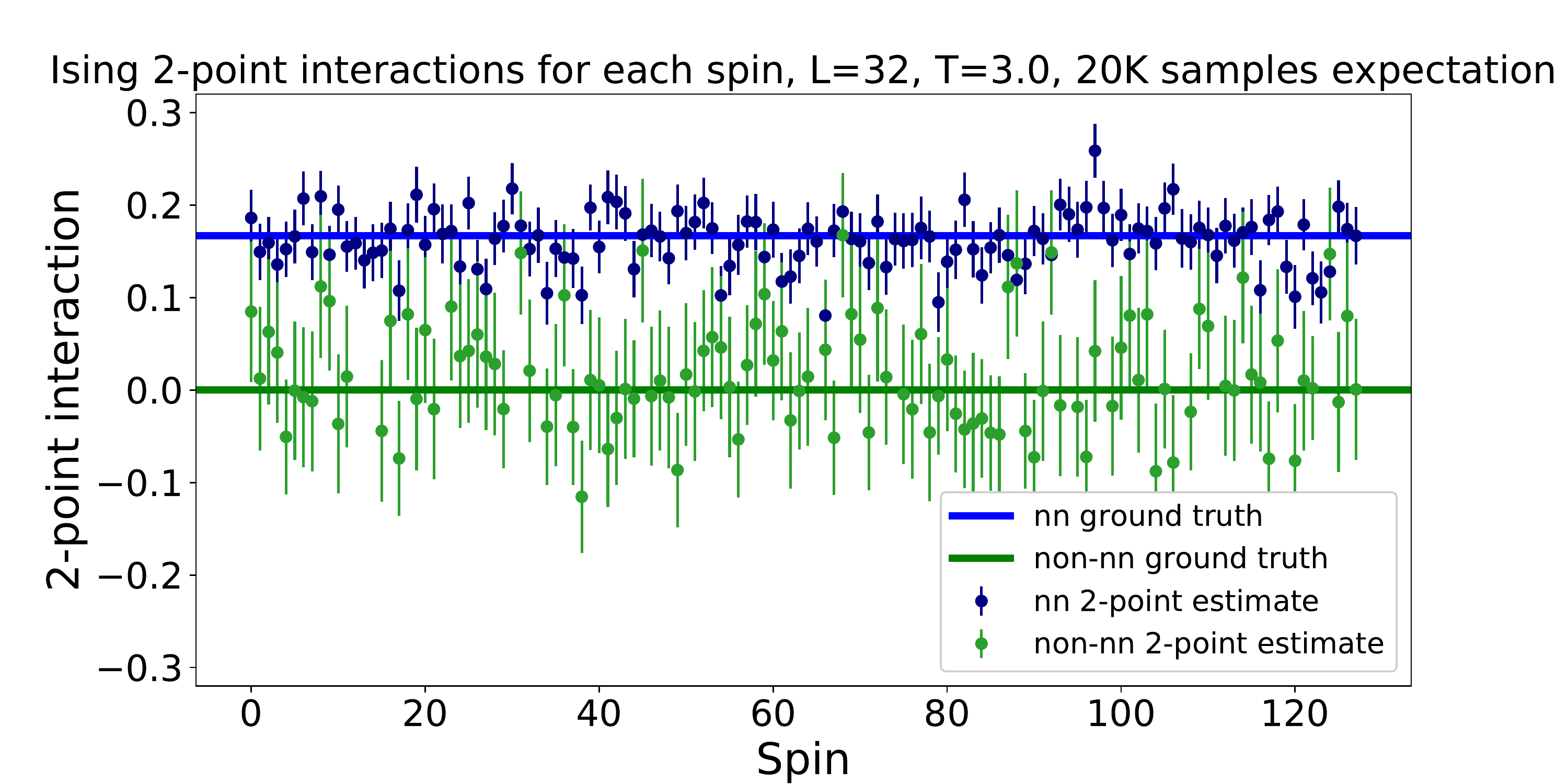}
	\endminipage\hfill
    \caption{$L^2=32^2$, $T=3.0$, with conditioning on the nearest neighbours to estimate $I_{ij}^m$ for both nearest and non-nearest neighbour spin pairs. In order to reduce clutter $2\times128$ interactions are shown. No translational invariance is used. Top: The results are computed over a total of 100K samples, using Eq.~\ref{eq:interaction_mult_2_expectation} and statistical bootstrap, as compared to bottom: The results are computed over a total of 20K samples. For the latter, there is sufficient power to accurately estimate all the nearest neighbour interactions, as well as approximately 98\% of non-nearest neighbour interactions.}
    \label{fig:L32_T30_per_spin_indept_and_expect}
\end{center}
\end{figure}

\subsection{Numerical evidence for conditional independence}
\label{sec:cond-indept-numerical}
In the first step of the Targeted Learning road map stated in Sec.~\ref{sec:targeted_learning}, we select the set of probability distributions $p$ that are compatible with \emph{a priori} knowledge regarding the data and how it is generated.
For example, in the case of the Ising model, this knowledge could include information regarding the nearest neighbour structure, namely, that by conditioning on the parental spins of two spins, the two spins become independent of each other and the rest of the spins if they are non-nearest neighbours. If they are nearest neighbours, then they only become independent of the rest of the spins but not of each other. 
Then using the Markovian property and the Hammersley--Clifford theorem of Sec.~\ref{sec:conditional_independence}, to obtain the interactions between pairs of spins, it suffices to condition on their nearest neighbours to be zero, rather than all the rest of the spins (see App.~\ref{app:Hammersley_Clifford} for a proof). 
This results in improved statistical estimates, as the number of samples that satisfy the latter condition will be significantly larger than the former.
The Markovian parent structure of nearest and non-nearest neighbours in the $2$-dimensional Ising model are presented in Fig.~\ref{fig:Markov_structure_2D_Ising}.

\begin{figure}[!htb]
\begin{center}
    \minipage{0.25\textwidth}
	\includegraphics[width=0.9\linewidth]{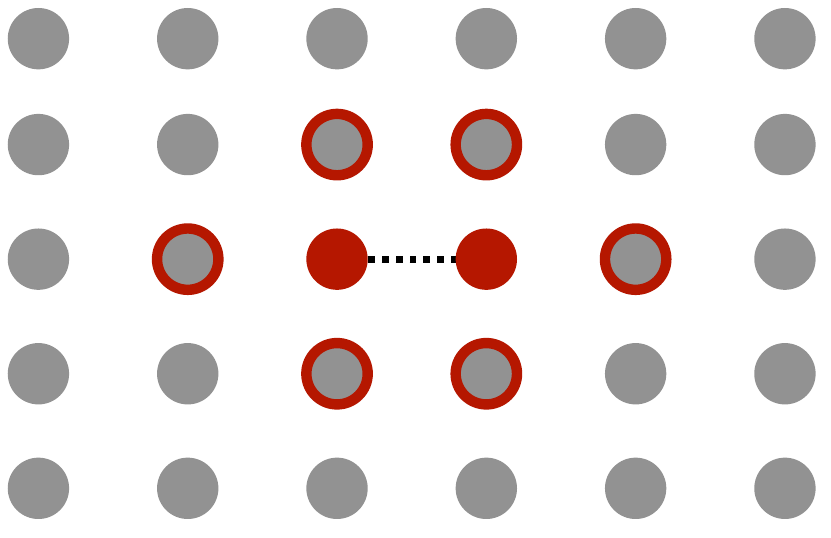}
	\endminipage 
	\vspace{0.5cm}
	\minipage{0.25\textwidth}
	\includegraphics[width=0.9\linewidth]{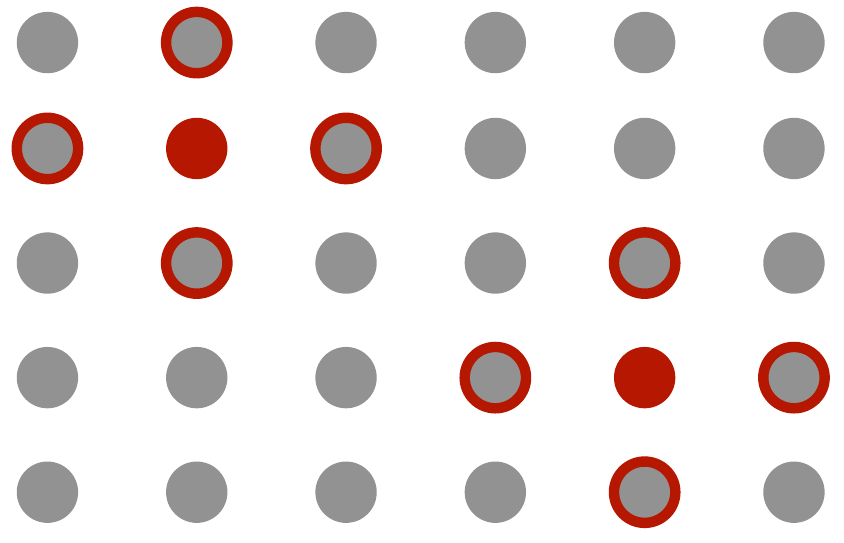}
	\endminipage\hfill
    \caption{Nearest neighbour structure in the $2$-dimensional Ising model. Parents of the pairs of interest required for conditional independence: the $6$ parents of a nearest neighbour pair (top), and the $8$ parents of a non-nearest neighbour pair (bottom).}
    \label{fig:Markov_structure_2D_Ising}
\end{center}
\end{figure}

If \emph{a priori} information on conditional independence is not known one can use non-parametric statistical testing to determine such independence criteria, in order to improve the estimates of interactions.
The $\chi$-squared test of independence can be used for the case of binary or categorical variables and, \eg, an information-theoretic independence criterion for continuous variables~\cite{MR3992389}.
Algorithms such as Peter--Clark can then be employed to automatically detect (conditional) independence using a given test in an efficient way \cite{pcReview2019}.
Discussion on the latter is beyond the scope of this work, and we only briefly present results on applying a $\chi$-squared test directly on Ising data as an example. 

We perform the $\chi$-squared test of independence on Ising configurations generated at the critical temperature which is approximately $T=2.3$.
The null hypothesis $H_0$ of $\chi$-squared is that the variables are \emph{independent} of each other. 
Given a particular threshold, if the computed $p$-values becomes less than the threshold, we reject the null hypothesis in favour of the alternative hypothesis $H_1$, \ie, that the variables in question are indeed dependent. 
For the $2$-dimensional Ising model at the critical temperature we expect the correlation length to diverge, and therefore to observe a large degree of dependence amongst all spins. 
Therefore, taking pairs of spins, while conditioning on no other spins in the system, we expect the $\chi$-squared test to result in small $p$-values, indicating dependence amongst the spins. 
Indeed, we observed $p\approx0$ for all pairs of spins in this case.  
If, on the other hand, we condition on all $8$ nearest neighbour spins of any non-nearest neighbour spin pair, we observed that most of the $p$-values are large, indicating independence as expected.
However, the test does result in less than 10\% of the non-nearest neighbour spin pairs having small $p$-values, namely less than the chosen threshold of $0.1$, see Fig.~\ref{fig:hist-chisq-nonnn-8parentscond}.
These are the result of a type I error, or false claim of dependence, which do not bias the estimation of the interactions but merely render the procedure more conservative than necessary, at the cost of larger variance.  

\begin{figure}[!htb]
\begin{center}
    \minipage{0.45\textwidth}
	\includegraphics[width=\linewidth]{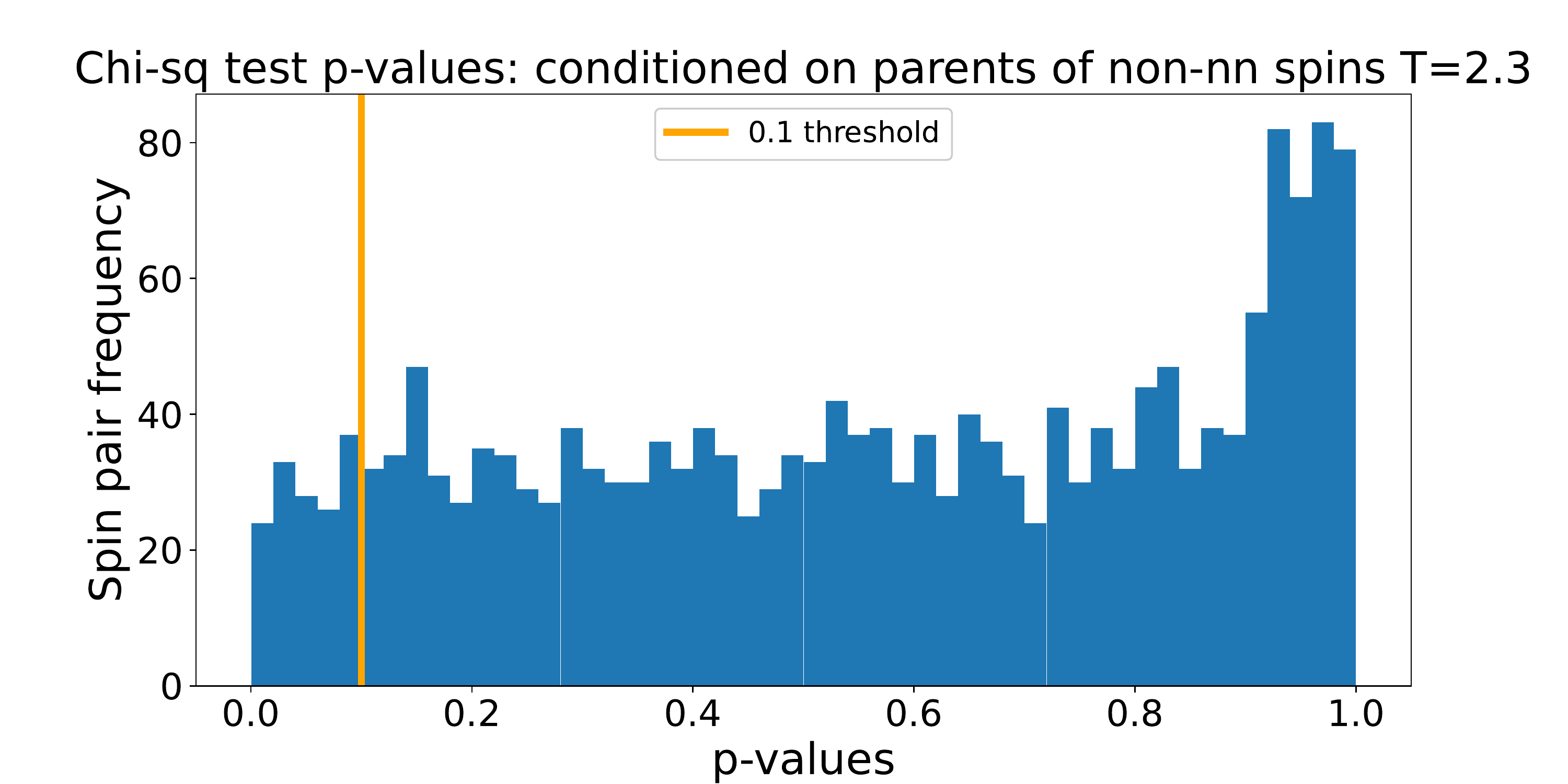}
	\endminipage\hfill
    \caption{Histogram of $\chi$-squared test $p$-values for non-nearest neighbour spins pairs, conditioned on all of the 8 parents, for the $T=2.3$ Ising model. We expect the null hypothesis of independence not to be rejected, \ie, high $p$-values. This is indeed observed with less than 10\% of the $p$-values being less than the chosen threshold 0.1. The $\chi$-squared test has incorrectly taken these as dependent, however, taking more spins into account when conditioning does not introduce any bias in the estimation of the interactions.  }
    \label{fig:hist-chisq-nonnn-8parentscond}
\end{center}
\end{figure}

Next, we observe what happens if we, wrongly, do not condition on all the parents of variables that $\chi$-squared otherwise declares as dependent.
As an example, conditioning on only $2$ of the total of $8$ nearest neighbours, the $\chi$-squared test declares all $p\approx0$. 
Estimating the interaction between non-nearest neighbour spin pairs, whilst conditioning on two parents only, results in highly biased estimates of the interactions, as expected, as indicated on the right hands side of Fig.~\ref{fig:chisq-ignore-condition-on-2parents-only}. 

\begin{figure}[!htb]
\begin{center}
    \minipage{0.45\textwidth}
	\includegraphics[width=\linewidth]{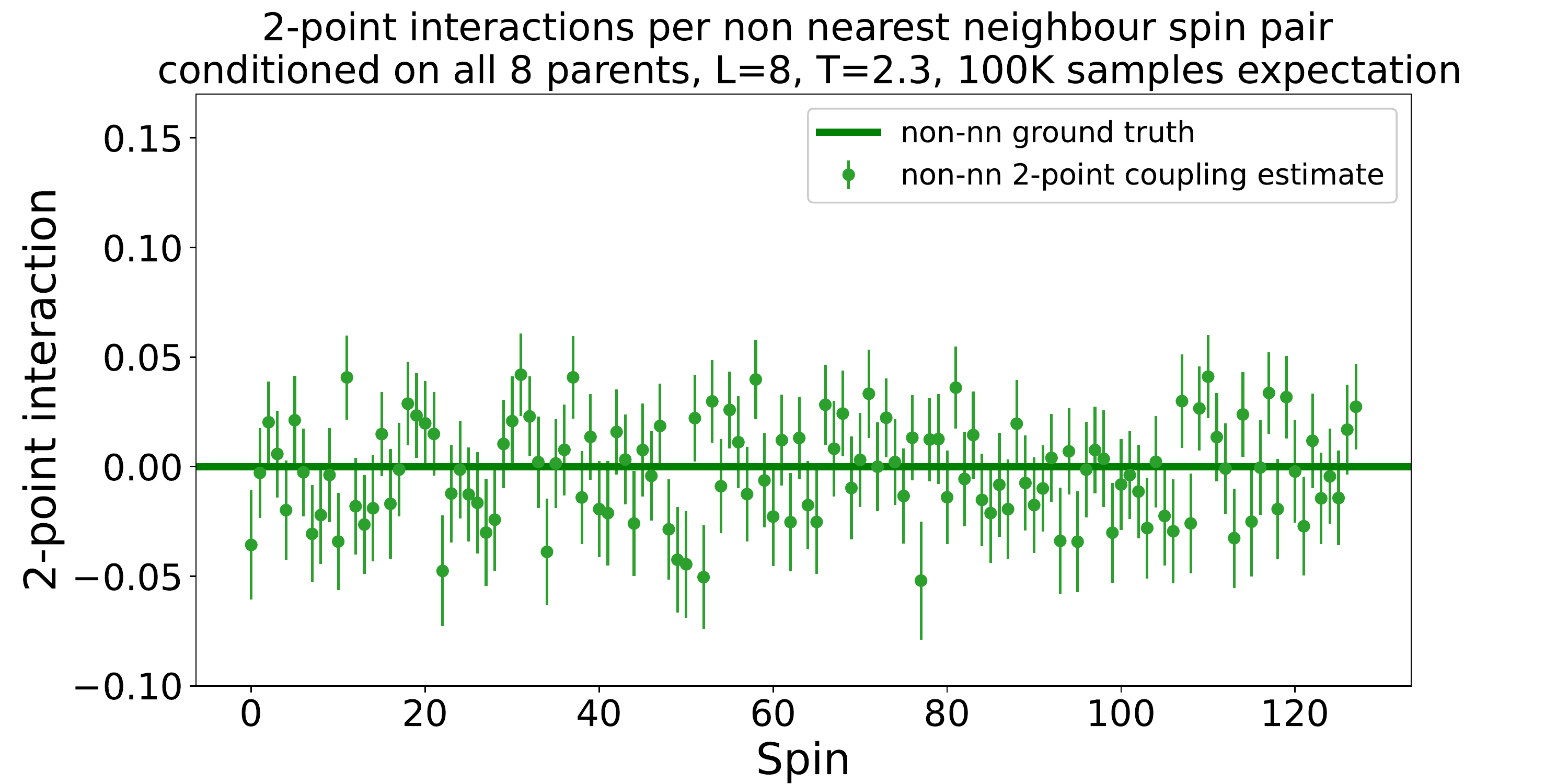}
	\endminipage
	\vspace{0.2cm}
	\minipage{0.45\textwidth}
	\includegraphics[width=\linewidth]{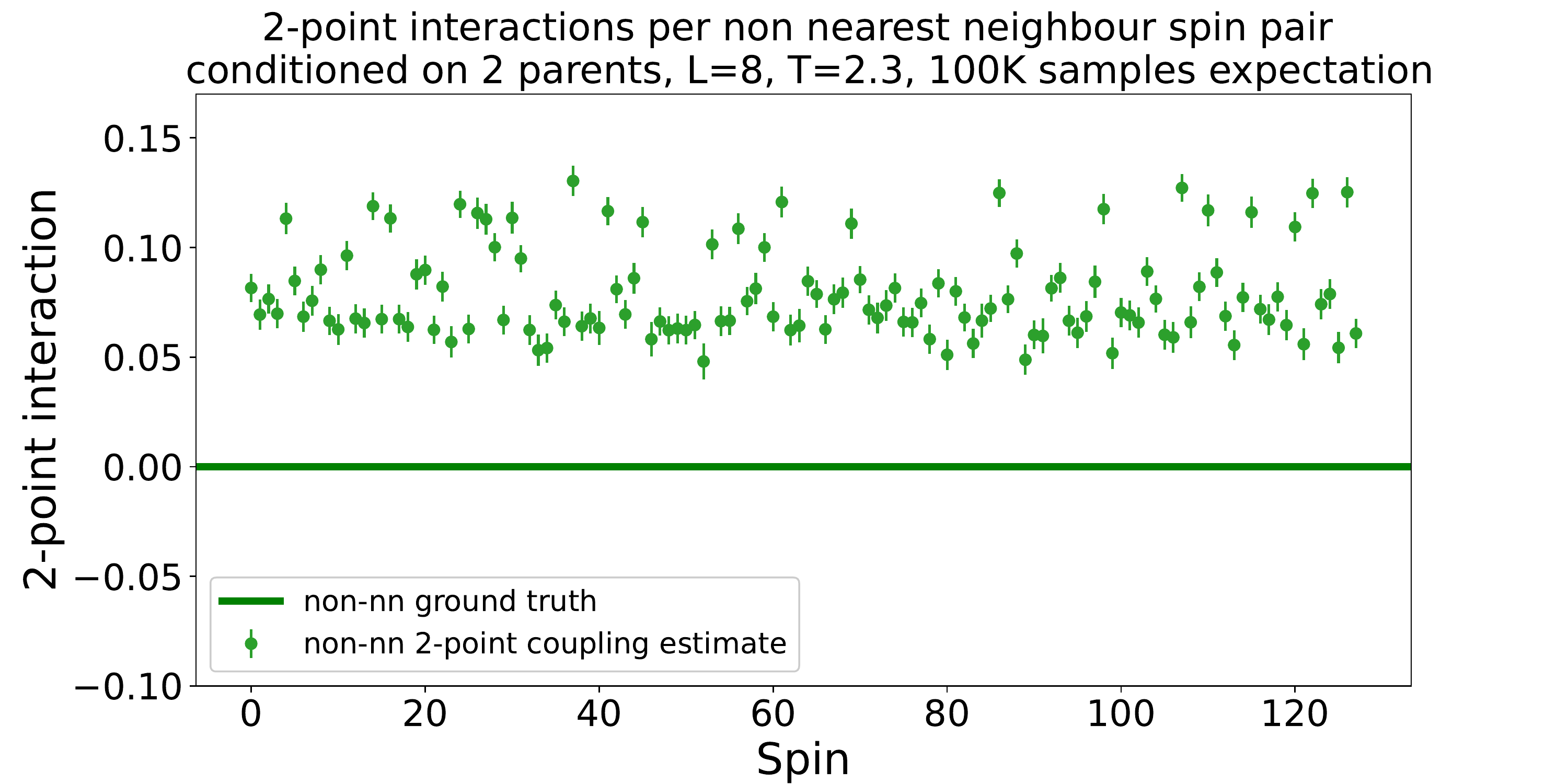}
	\endminipage\hfill
    \caption{Non-nearest neighbour 2-point interactions for Ising configurations near the critical temperature $T=2.3$, 100K samples.
    128 spin pairs are taken as representatives of all 1888 non-nearest neighbour spin pairs. Top: Conditioning on all 8 parents, estimation accurately recovers the ground truth. Bottom: Conditioning on only 2 parents, even though $\chi$-square has accurately detected dependence, results in biased estimates of the interactions. }
    \label{fig:chisq-ignore-condition-on-2parents-only}
\end{center}
\end{figure}

Finally, we condition on $4$ out of the $8$ nearest neighbours, for all the non-nearest neighbour spin pairs, with all $4$ blocking one of the spins from the rest of the system.
In this case the $\chi$-squared test seems to declare independence in most cases.
This is a type II error: failure to reject a false null hypothesis of independence.
We examine the resulting bias on the estimates for the associated $2$-point interactions in Fig.~\ref{fig:chisq-ignore-condition-on-4parents-only}: The level of statistical variation in the data is large enough to compensate for the bias introduced by not conditioning on all the Markovian parents.
In the tests that we have performed, we have observed these features both at cold and hot temperatures as well. 

\begin{figure}[!htb]
\begin{center}
    \minipage{0.45\textwidth}
	\includegraphics[width=\linewidth]{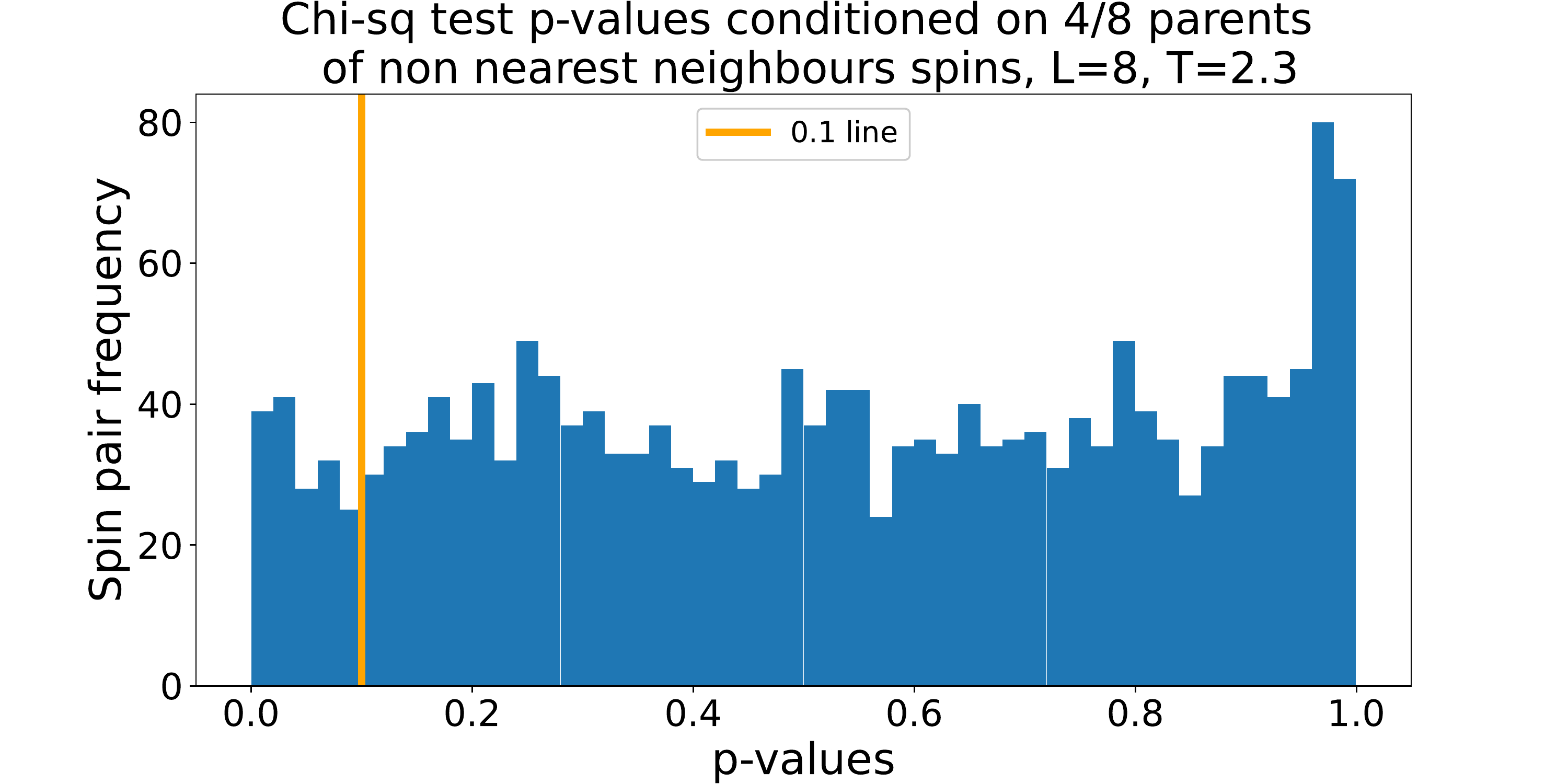}
	\endminipage
	\vspace{0.2cm}
	\minipage{0.45\textwidth}
	\includegraphics[width=\linewidth]{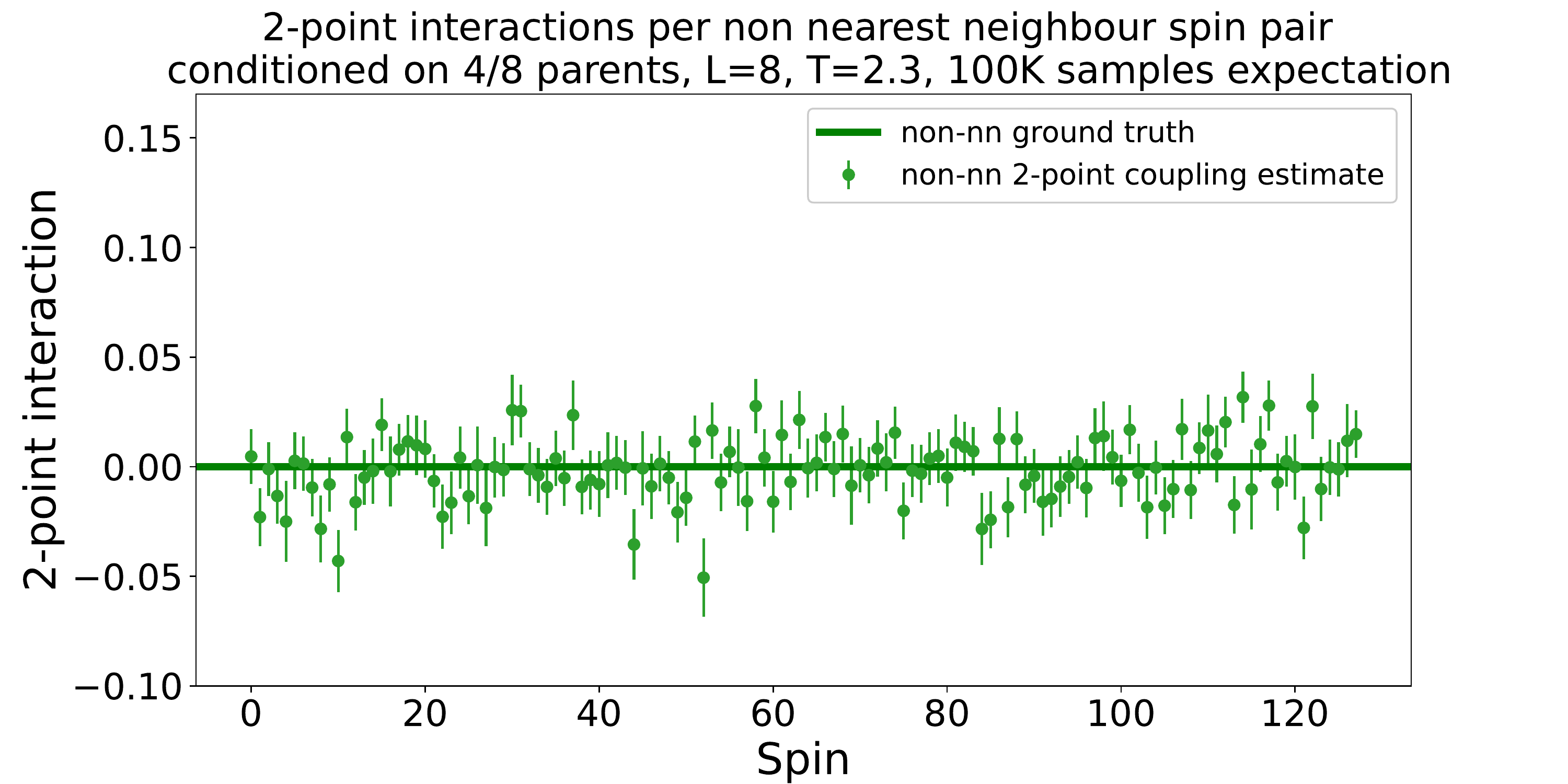}
	\endminipage\hfill
    \caption{Non-nearest neighbour 2-point interactions for Ising configurations near the critical temperature $T=2.3$, 100K samples. 128 spin pairs are taken as representatives of all 1888 non-nearest neighbour spin pairs. Top: Conditioning on 4 out of the total of 8 parents, the $\chi$-squared test is unable to detect dependence. Bottom: Numerical results indicate that when $\chi$-squared does not detect dependence in the data, conditioning on 4 out of the total of 8 parents does not introduce strong bias in estimating the interactions accurately. }
    \label{fig:chisq-ignore-condition-on-4parents-only}
\end{center}
\end{figure}

In summary, when \emph{a priori} knowledge regarding independence amongst variables is not available and has to be derived from the data, one can perform the non-parametric $\chi$-squared test for binary and categorical data.
If $\chi$-squared declares dependence amongst variables, we must ensure to condition on these when estimating the interactions.
If $\chi$-squared declares false independence, potentially due to the level of variance/noise in the data, it is likely to be the case that this missed degree of dependence is not so large as to bias the estimates of $n$-point interaction, again given the level of variance/noise in the data.

\section{Results III: A Hamiltonian with 1-, 2-, 3-, and 4-point interactions}\label{sec:4pt_Hamiltonian}
\subsection{Analytical formulation}
In this section, we consider an Ising-like Hamiltonian in the $\{-1,1\}$ basis with $4$-point couplings.
After transforming to the $\{0,1\}$ basis, this results in a Hamiltonian with non-zero self, $2$-point, $3$-point, and $4$-point couplings.
The setup is as follows.
Consider a $2$-dimensional square lattice of size $L^2$ with periodic boundary conditions, with a spin $\tilde{v}_i$ on each lattice point $i$ taking on values $\tilde{v}_i = \pm 1$.
A \emph{state} is the assignment $\tilde{\bf v}$ of a value $+1$ or $-1$ to each of the $L^2$ spins.
The Boltzmann distribution describes the probability $p(\tilde{\bf v}|T)$ that the system takes on a particular state $\tilde{\bf v}$ at temperature $T$,\ie,
\begin{equation}
    p(\tilde{\bf v}|T) = \frac{1}{\CZ(T)} e^{-E(\tilde{\bf v})},
\end{equation}
where,
\begin{equation}\label{eq:H_4pt_ising}
    E(\tilde{\bf v}) = -\frac{1}{T} \sum_{(i,j)}
        J_{i,j} \tilde{v}_{(i,j)} \tilde{v}_{(i+1,j)} \tilde{v}_{(i,j+1)} \tilde{v}_{(i+1,j+1)}.
\end{equation}
The sum runs over all $L^2$ lattice sites $(i,j) \in \{1,2,\ldots,L\}^2$ and $J_{i,j}$ is the coupling amongst the square of spins $\{\tilde{v}_{(i,j)}, \tilde{v}_{(i+1,j)}, \tilde{v}_{(i,j+1)}, \tilde{v}_{(i+1,j+1)}\}$.

We first solve the inverse problem defined by the Hamiltonian of Eq.~\ref{eq:H_4pt_ising} analytically.
Our non-parametric definition~\ref{def:interaction_mult} of multiplicative self, $2$-point, $3$-point, and $4$-point interaction amongst binary variables immediately recovers the couplings $-8J_{i,j}$, $8J_{i,j}$, $-8J_{i,j}$, and $16J_{i,j}$ respectively from the probability distribution of Eq.~\ref{eq:H_4pt_ising}, after applying $\ln(-)$ and correcting for double counting due to the change of basis $\{-1,1\} \mapsto \{0,1\}$.
To see this, we first apply the transformation $\tilde{v}_{(i,j)} = 2v_{(i,j)}-1$ expressing the values of a spin in terms of $\{0,1\}$ as opposed to $\{-1,1\}$ in order to apply the definition of multiplicative $n$-point interaction of Eq.~\ref{eq:interaction_mult}.
Thus, $\tilde{v}_{(i,j)} = -1$ corresponds to $v_{(i,j)} = 0$, whereas $\tilde{v}_{(i,j)} = 1$ corresponds to $v_{(i,j)} = 1$.
This yields,
\small
\begin{equation*}
\begin{split}
    &J_{i,j}\tilde{v}_{(i,j)} \tilde{v}_{(i+1,j)} \tilde{v}_{(i,j+1)} \tilde{v}_{(i+1,j+1)} = \\
    &J_{i,j}\bigl(2v_{(i,j)}-1\bigr) \bigl(2v_{(i+1,j)}-1\bigr) \bigl(2v_{(i,j+1)}-1\bigr) \bigl(2v_{(i+1,j+1)}-1\bigr),
\end{split}
\end{equation*}
\normalsize
for the contribution to $E({\bf v})$ of a single square of spins with the top left spin at lattice site $(i,j)$.
The interactions may now be computed by taking suitable derivatives of the energy function $E({\bf v})$ in the $\{0,1\}$ basis, whilst putting the remaining spins to zero, and taking care of double counting due to the change of basis.

\subsection{A Hamiltonian with 4-point interactions}
\label{sec:H-with-4pt_numerical}
In this section, we evaluate the performance of our non-parametric formulation of multiplicative interaction on data generated by an Ising-like Hamiltonian with $4$-point couplings in the $\{-1,1\}$ basis. This corresponds to having non-zero self, 2-point, 3-point, and 4-point interactions in the $\{0,1\}$ basis.

\begin{figure}[!htb]
\begin{center}
    \minipage{0.25\textwidth}
	\includegraphics[width=0.9\linewidth]{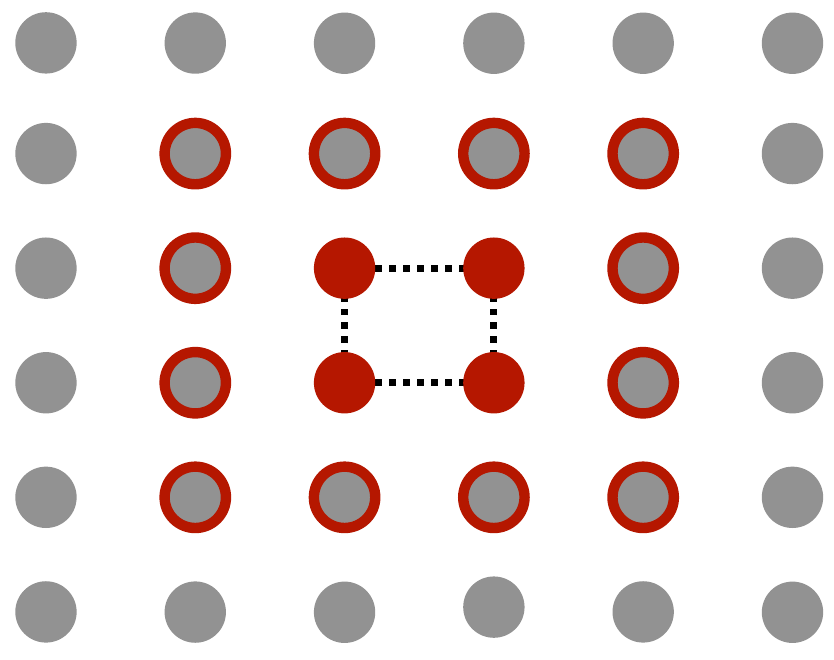}
	\endminipage
    \caption{Nearest neighbour structure in the Ising-like Hamiltonian with 4-point interactions. There are $12$ parents to be conditioned on for estimating the 4-point interaction amongst the quadruple of spins of interest. }
    \label{fig:Markov_structure_H4}
\end{center}
\end{figure}

One million samples were generated using the Metropolis algorithm, at $T=1$ and different coupling constants $0.1,0.125,0.15,0.2,0.25$.
The results for self to 4-point interactions, normalised by the corresponding coupling constant and corrected for change of basis factors, are presented in Fig.~\ref{fig:H4_higher_order_ave}.
As expected, the uncertainty on the estimations increases as we consider higher-order interactions.
Nevertheless, at one million samples, the uncertainty on the average 4-point interaction is approximately less than 10\% in this system.
Reducing the sample sizes from one million to 500K, then to 200K, results in not having sufficient power to estimate the 4-point and the 3-point interactions respectively.
The results for the interactions per pair, triple and quadruple of spins are presented in Fig.~\ref{fig:H4_per_spins_2pt_3pt} and Fig.~\ref{fig:H4_per_spins_4t}

\begin{figure}[!htb]
\begin{center}
    \minipage{0.5\textwidth}
	\includegraphics[width=\linewidth]{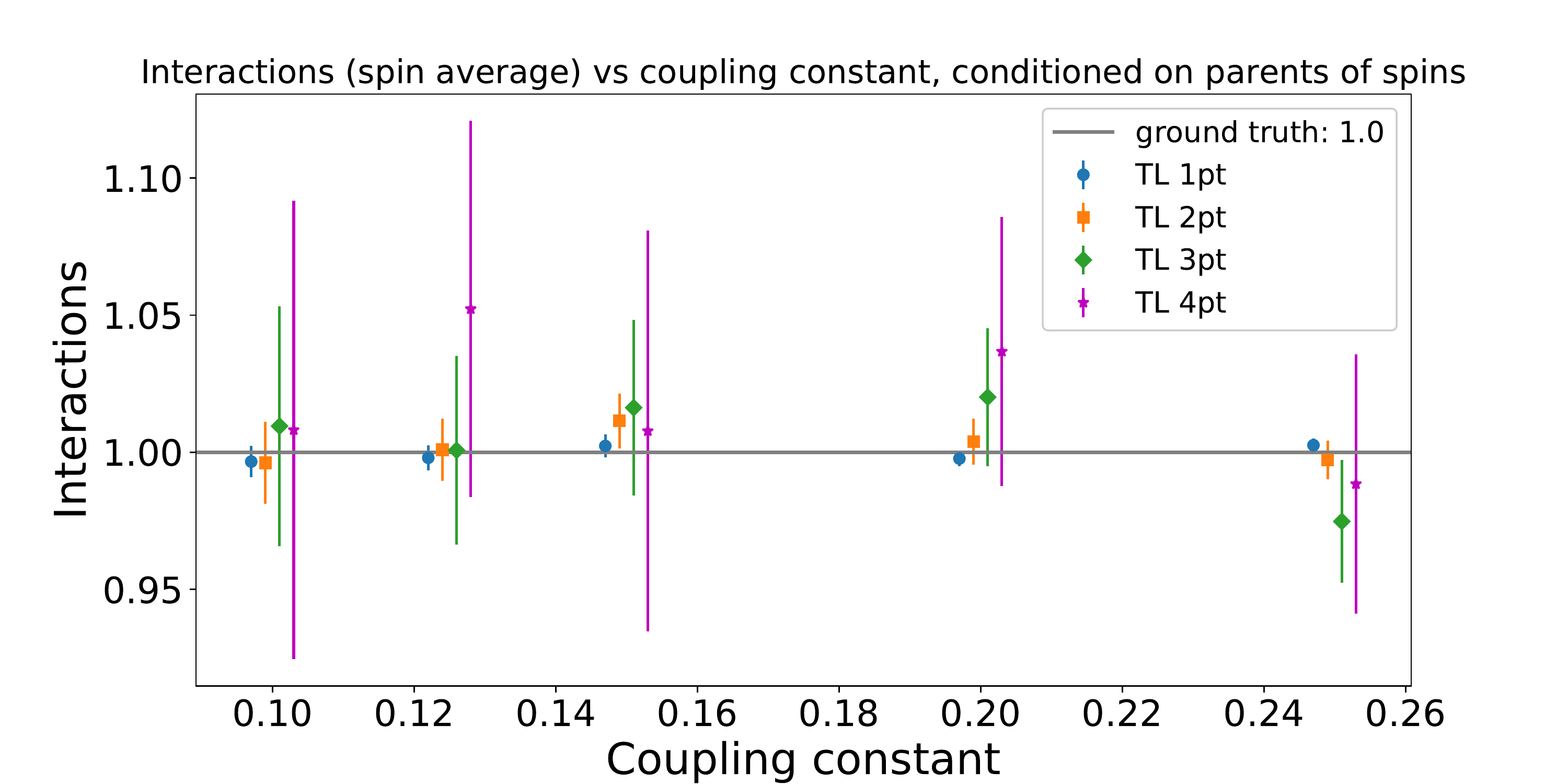}
	\endminipage
    \caption{Estimates of the self to 4-point interactions $I_{ijkl}^m$ averaged across spins and  normalised by various values of coupling constants in the Hamiltonian $0.1,0.125,0.15,0.2,0.25$. Estimations are performed using 1M samples. As the total number of samples used for estimation is lowered, the power to detect higher-order interactions is reduced. }
    \label{fig:H4_higher_order_ave}
\end{center}
\end{figure}

\begin{figure}[!htb]
\begin{center}
    \minipage{0.45\textwidth}
	\includegraphics[width=\linewidth]{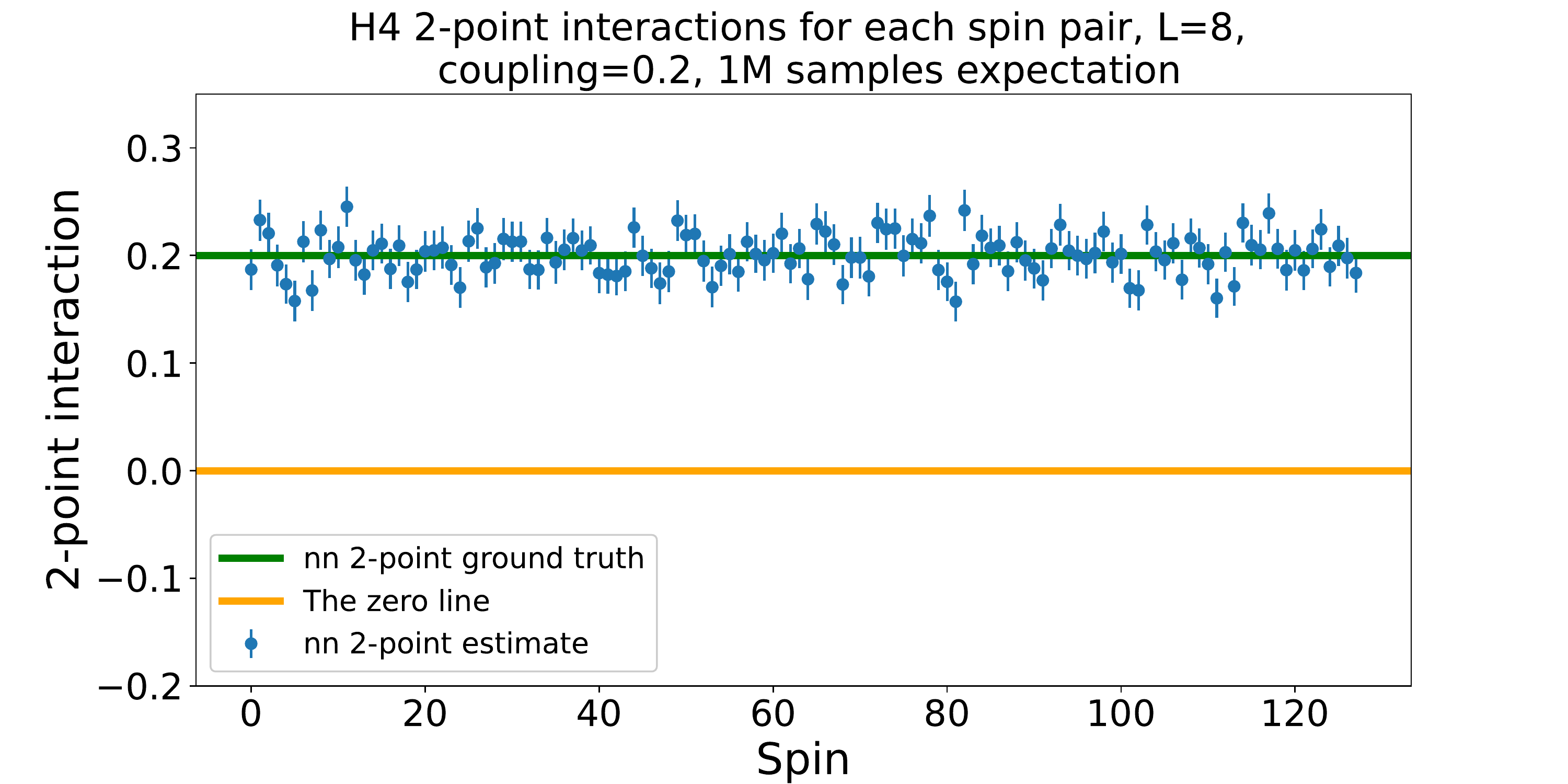}
	\endminipage
	\vspace{0.2cm}
	\minipage{0.45\textwidth}
	\includegraphics[width=\linewidth]{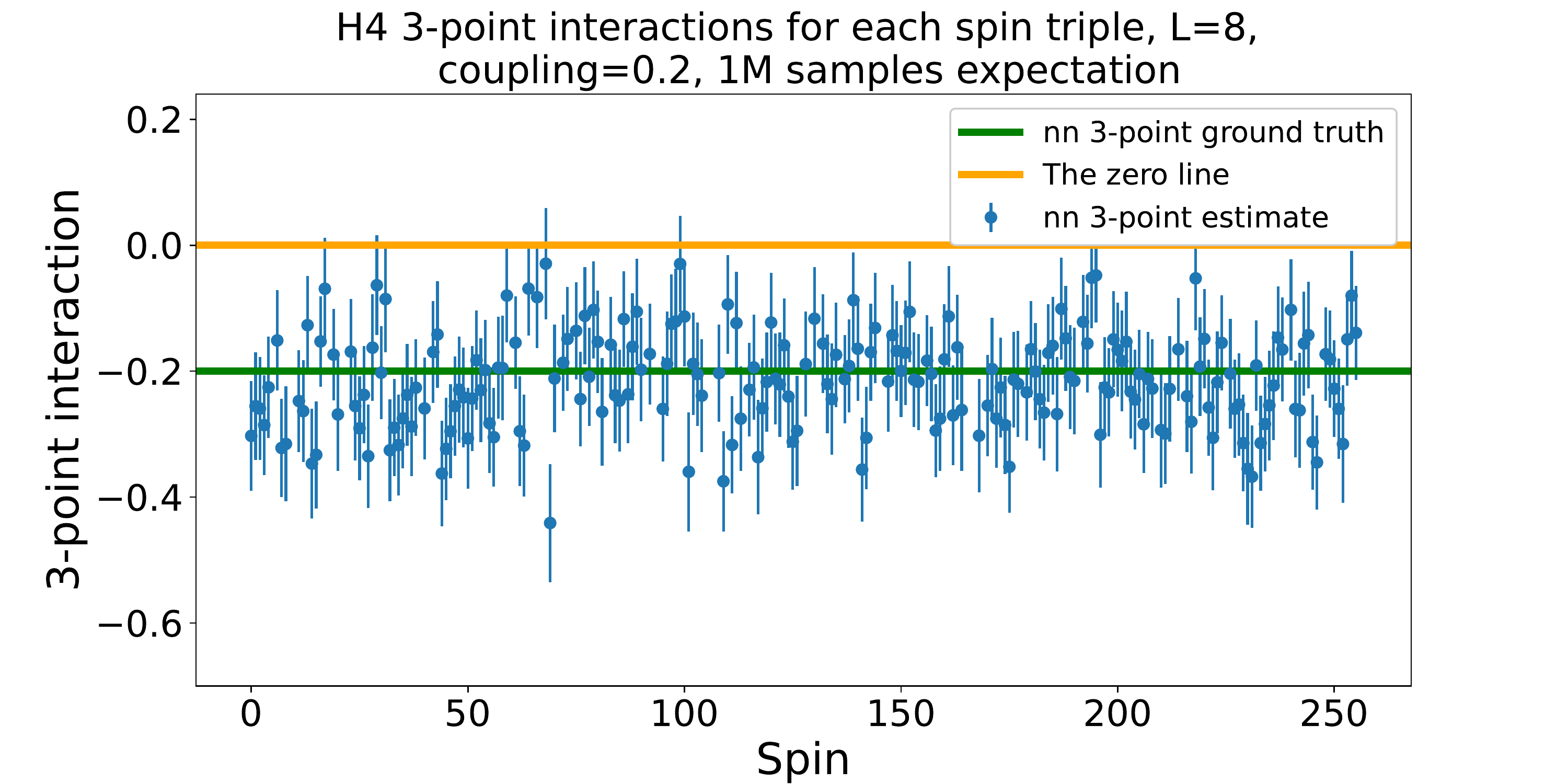}
	\endminipage\hfill
    \caption{2-point (top) and 3-point (bottom) per spin estimates of interactions for the ground truth coupling constant $0.2$. Estimations are performed on 1M samples.}
    \label{fig:H4_per_spins_2pt_3pt}
\end{center}
\end{figure}

\begin{figure}[!htb]
\begin{center}
    \minipage{0.45\textwidth}
	\includegraphics[width=\linewidth]{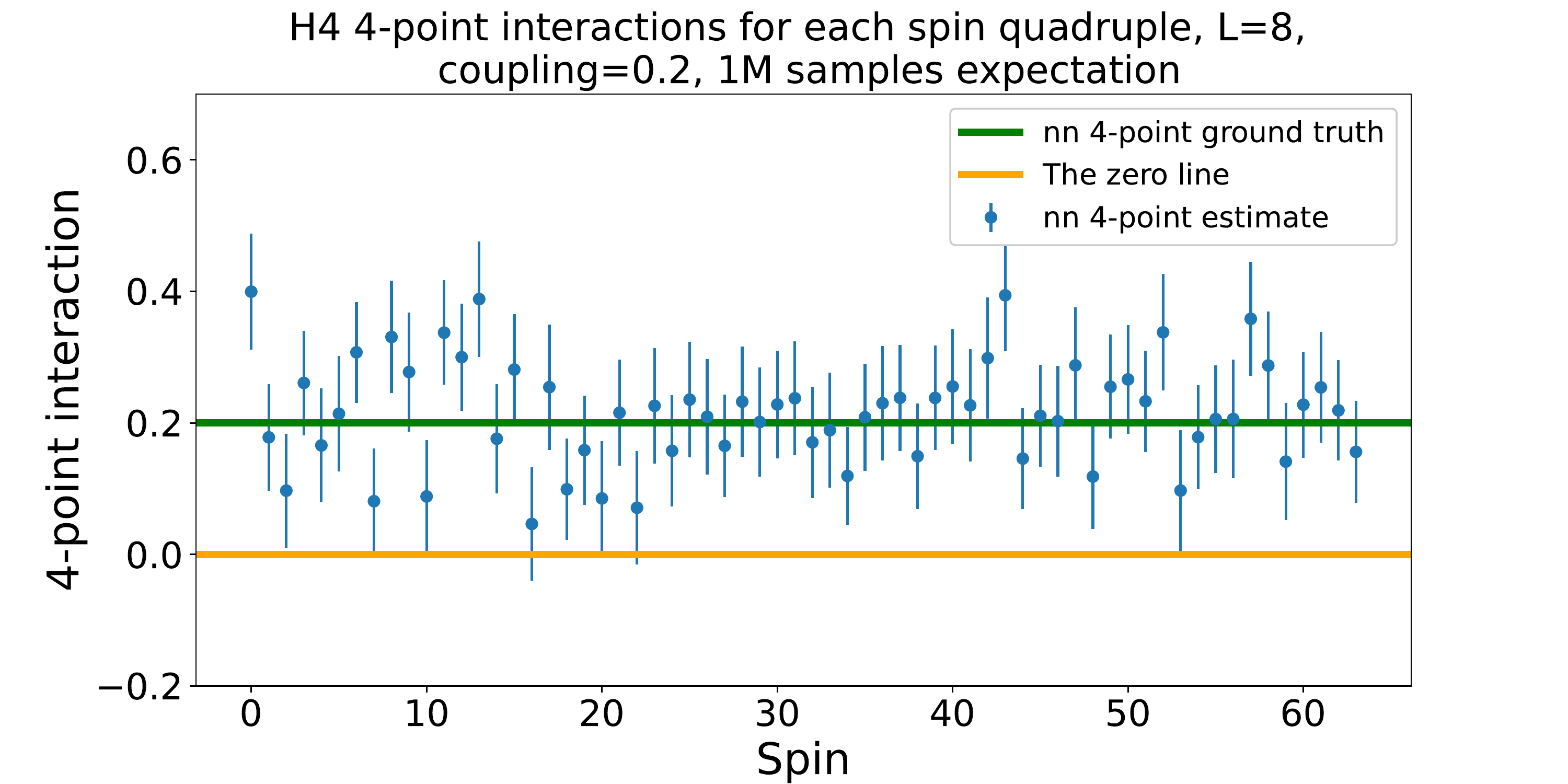}
	\endminipage
    \caption{4-point per spin estimates of interactions for the ground truth coupling constant $0.2$. Estimations are performed using 1M samples. We observe that the variance is large, in the sense that if the ground truth were to be unknown, some of the couplings would be considered as insignificant.}
    \label{fig:H4_per_spins_4t}
\end{center}
\end{figure}

\subsection{Interaction in energy-based models}\label{sec:general-energy-based}
Our non-parametric definition of $n$-point interaction applies to any set of $n$ binary and categorical random variables in any probability distribution $p$.
For example, if the probability distribution is believed to be a Boltzmann distribution, our formulation can be used to estimate all the $n$-point interactions, \ie, the coefficients in the Hamiltonian up to statistics, \eg, as shown in Sec.~\ref{sec:H-with-4pt_numerical} numerically.
In particular, given any parametric form $p_{\theta}$, our formulation yields an analytical, closed form expression for all $n$-point interactions in terms of the parameters $\theta$ of the given model.
For example, the restricted Boltzmann machine was dealt with in Sec.~\ref{sec:RBM_analytic}.
Note, however, that in such energy-based neural networks determining the $n$-point interaction is a two-step procedure: (i) Marginalising of the hidden (latent) variables to obtain the probability distribution in terms of the visible  variables only, and (ii) replacing the probabilities $p$ in Eq.~\ref{eq:interaction_mult} with the parametric form $p_\theta$.
Thanks to the Targeted Learning framework, the last step can be performed directly without the need for asymptotic expansions and resummations.

\section{Conclusions \& future work}
\label{sec:conclusions}
In this work, we have provided a non-parametric solution to the inverse problem of estimating $n$-point interactions amongst binary and categorical random variables directly from data, using the framework of Targeted Learning.
In doing so, no parametric assumptions have to be made, yielding a fully model-independent and unbiased estimator of interaction at all orders. 
We have shown that interaction can naturally be interpreted as a derivative and, more specifically, that $n$-point interactions are inductively interpretable as a \emph{change} in $(n-1)$-point interaction when fixing any one of the $n$ variables.
Under a Markovian assumption, which is satisfied by all energy-based models in statistical physics and machine learning, we have demonstrated that interaction can be efficiently estimated from data by only conditioning on parent variables.
If the parent structure is known, or has been inferred from a non-parametric independence test, one can substantially reduce the sample size required to obtain an accurate estimate.
Furthermore, as the estimator only consists of expectation values over the data, the run time on a local machine is of the order of a few minutes.
We have illustrated the above both analytically and numerically on a $2$-dimensional Ising Hamiltonian, a $4$-point Ising-like Hamiltonian, and the distribution of a restricted Boltzmann machine.
Moreover, we have argued that our formulation can be used to extract closed form expressions of $n$-point interaction in any system of binary and categorical random variables, such as energy-based neural networks, where this coupling cannot directly be read off from a Hamiltonian, \eg, due to multiple hidden nodes.
Finally, we have indicated how our definition of interaction via Targeted Learning has applications in population biomedicine, in particular genome-wide association studies (GWAS), since it both removes the need for parametric assumptions altogether and correctly accounts for molecular interaction effects (epistasis), in contrast to current approaches in the literature.
\\

In future work, we plan to examine the bias-variance trade-off in extracting $n$-point interactions from other generative networks, such as Variational Auto-Encoders (VAE) and Generative Adversarial Networks (GAN).

\acknowledgements
We are most grateful to Mark van der Laan for his suggestions regarding the formulation of $2$-variable interactions using the Targeted Learning framework, in a private conversation at the {\it Causal machine learning masterclass}, the Alan Turing Institute, London.
We are thankful to Luigi Del Debbio for his comments on the numerical results, as well as Andrew Papanastasiou and Abel Jansma for reading and commenting on the manuscript. We are also thankful to Chris Ponting and Neil Clark, for their insights into the biological applicability of our work.
S.V.B. is funded by the Deutsche Forschungsgemeinschaft (DFG, German Research Foundation) under Germany’s Excellence Strategy -- EXC-2047/1 -- 390685813.
A.K. is a cross-disciplinary postdoctoral fellow supported by funding from the University of Edinburgh and Medical Research Council (core grant to the MRC Institute of Genetics and Molecular Medicine).

\appendix

\section{Additive interaction for categorical variables}
\label{app:categorical_variables}

We make the following remarks regarding Eq.~\ref{eq:interaction_add_2_categorical} of additive $2$-point interaction for categorical variables.
\begin{enumerate}
    \item Similar to the notion of interaction in the binary case, the notion of interaction for categorical variables is inherently symmetric under the exchange of the variables $(T_1 \colon t_1 \to t_1')$ and $(T_2 \colon t_2 \to t_2')$, \ie,
        \begin{equation}\label{eq:interaction_symmetric_general}
            I^{a}_{1,2}(t_1t_1';t_2t_2') = I^{a}_{2,1}(t_2t_2';t_1t_1').
        \end{equation}
    \item The interaction between the effect of $T_1 \colon t_1 \to t_1'$ on $Y$ and $T_2 \colon t_2 \to t_2'$ on $Y$ is opposite in sign to the effect of $T_1 \colon t_1' \to t_1$ on $Y$ (we swap $t_1$ and $t_1'$) and $T_2 \colon t_2 \to t_2'$ on $Y$, \ie,
        \begin{equation}
            I^{a}_{1,2}(t_1t_1';t_2t_2') = - I^{a}_{1,2}(t_1't_1;t_2t_2').
        \end{equation}
        For example, the interaction between the effect of \emph{turning on} variable $T_1 \colon 0 \to 1$ on $Y$ and the effect of $T_2 \colon t_2 \to t_2'$ on $Y$, is opposite in sign to the interaction between the effect of \emph{turning off} variable $T_1 \colon 1 \to 0$ on $Y$ and the effect of $T_2 \colon t_2 \to t_2'$ on $Y$.
    \item As a result of the above remark, swapping both categories yields the same interaction
        \begin{equation}
            I^{a}_{1,2}(t_1t_1';t_2t_2') = (-1)^2 I^{a}_{1,2}(t_1't_1;t_2't_2).
        \end{equation}
\end{enumerate}

Finally, the additive $2$-point interaction between categorical variables satisfies the following \emph{transitivity}:
\begin{prop}\label{prop:interaction_transitivity}
    Let $T_1,T_2$ be two categorical variables, let $\{0,1,2\}$ denote the labels of three categories of $T_1$, and let $\{0,1\}$ denote the labels of two categories of $T_2$.
    Then the interactions satisfy transitivity, \ie,
    \begin{equation}\label{eq:interaction_transitivity}
        I^{a}_{1,2}(01;01) + I^{a}_{1,2}(12;01) = I^{a}_{1,2}(02;01).
    \end{equation}
\end{prop}
Heuristically, the result states that the sum of the effect on $Y$ of changing $T_1$ from $0$ to $1$ and then changing $T_1$ from $1$ to $2$, equals the effect on $Y$ of changing $T_1$ from $0$ to $2$ directly.
The same heuristic holds for the interaction with the effect of $T_2 \colon 0 \to 1$ on $Y$ as this effect is the same during all three steps of the procedure.
\begin{proof}
    We define the function $f \colon \{0,1,2\} \times \{0,1\} \to \BR$ as
    \begin{equation}
        f(t_1,t_2) \coloneqq \BE(Y \mid T_1 = t_1, T_2 = t_2, \uT = 0).
    \end{equation}
    We may express the average treatment effect in terms of $f$ as $\ATE_{T_1 \colon t_1 \to t_1'}(Y \mid T_2 = t_2, \uT = 0) = f(t_1',t_2)-f(t_1,t_2)$.
    This leads to the following expression for the interaction in terms of $f$,
    \begin{equation}
        I^{a}_{1,2}(t_1t_1';t_2t_2') = \bigl[f(t_1',t_2')-f(t_1,t_2')\bigr] - \bigl[f(t_1',t_2)-f(t_1,t_2)\bigr].
    \end{equation}
    Eq.~\ref{eq:interaction_transitivity} now follows by writing out both sides:
    \begin{align*}
    \begin{split}
        &I^{a}_{1,2}(01;01) + I^{a}_{1,2}(12;01) = \\
        &\bigl[f(1,1)-f(0,1)\bigr] - \bigl[f(1,0)-f(0,0)\bigr] \\
        + &\bigl[f(2,1)-f(1,1)\bigr] - \bigl[f(2,0)-f(1,0)\bigr] \\
        = &\bigl[f(2,1)-f(0,1)\bigr] - \bigl[f(2,0)-f(0,0)\bigr]
        = I^{a}_{1,2}(02;01).
    \end{split}
    \end{align*}
    This completes the proof.
\end{proof}

As an important corollary, we obtain a \emph{criterion} for linear dependence of the interaction $I^{a}_{1,2}$ on particular labels of the categorical variables.
The precise statement is the following.
\begin{cor}\label{cor:linearity_condition}
    Let $T_1,T_2$ be two categorical variables, let $\{0,1,2\}$ denote the labels of three categories of $T_1$, and let $\{0,1\}$ denote the labels of two categories of $T_2$.
    If
    \begin{equation}\label{eq:linearity_condition}
        I^{a}_{1,2}(01;01) = I^{a}_{1,2}(12;01),
    \end{equation}
    then the interaction $I^{a}_{1,2}(\underline{\hspace{2mm}} ;01)$ between the effect of $T_1$ on $Y$ and the effect of $T_2 \colon 0 \to 1$ on $Y$ depends linearly on the label of the categorical variable $T_1$, in the sense that
    \begin{equation}\label{eq:linearity_relation}
        I^{a}_{1,2}(02;01) = 2 \cdot I^{a}_{1,2}(01;01).
    \end{equation}
    Thus the $2$ of the label $02$ can be taken outside to multiply the interaction leaving the label $01$, hence the term \emph{linear}.
\end{cor}
\begin{proof}
    This follows directly from~Proposition~\ref{prop:interaction_transitivity}.
\end{proof}
A similar statement holds for the interaction conditioned on a particular covariate $W = w$, and when interchanging the roles of $T_1$ and $T_2$ by considering two categories for $T_1$ and three for $T_2$.

This result has a graphical interpretation in terms of the following triangle:
\begin{equation}
\begin{tikzcd}
    & & 2 \\
    0 \arrow[swap]{r}{i(01)} \arrow{urr}{i(02)} & 1 \arrow[swap]{ur}{i(12)} &
\end{tikzcd}
\end{equation}
where we denote the corresponding interaction by $i(t_1t_1') = I^{a}_{1,2}(t_1t_1';01)$ which is represented by the length of the \emph{vertical} component of the arrow.
For example, in the above picture $i(01) = 0$ since the arrow is horizontal, and $i(12) = i(02)$ as the vertical components of both arrows have the same length.
Thus, the transitive relation
\begin{equation}
    i(01)+i(12)=i(02)
\end{equation}
allows us to draw this triangle.
Under the condition of Corollary~\ref{cor:linearity_condition}, the vertical components of the arrow $0 \to 1$ and $1 \to 2$ are equal, \ie, $i(01) = i(12)$, in which case the above triangle is \emph{degenerate}, \ie, a line segment.
In conclusion, the linearity of the dependence on the categorical variable $T_1$ of the interaction $I^{a}_{1,2}(\underline{\hspace{2mm}};01)$ between the effect of $T_1 \colon 0 \to 1$ on $Y$ and the effect of $T_2 \colon 0 \to 1$ on $Y$, in the sense that
\begin{equation}
    i(02) = 2 \cdot i(01),
\end{equation}
corresponds to degeneracy of the above triangle.
This is a geometrical criterion for linearity.
\\

The notion of interaction as in Eq.~\ref{eq:interaction_add_2_categorical} is \emph{independent} of the chosen labels for the categorical random variables $T_1,T_2$ whether they be numbers, farm animals, or names of cabinet ministers.
The \emph{interpretation} of equation Eq.~\ref{eq:linearity_relation} in terms of linearity \emph{depends} on the chosen labels since it forces them to appear in the mathematical formula Eq.~\ref{eq:linearity_relation}.
Naturally, the above discussion admits a direct generalisation to the case of categorical variables describing more than three categories.
In fact, all results are formulated in this general setting already, apart from assigning the particular labels $\{0,1,2\}$ or $\{0,1\}$.


\section{Symmetry of $n$-point interaction}
In this section, we prove the symmetry under any permutation of the variables $X_{i_1},\ldots,X_{i_n}$ of the multiplicative formulation of $n$-point interaction.
\begin{prop}\label{prop:symmetry_interaction_mult}
    Let $K = \{i_1,\ldots,i_n\} \subset \{0,1,\ldots,r\}$ be a subset of indices, and let $\sigma$ be any of the $n!$ permutations of $\{1,2,\ldots,n\}$ that acts on the $n$-tuple $K$ as $\sigma(K) = \sigma(i_1,\ldots,i_n) = \{i_{\sigma(1),\ldots,\sigma(n)}\}$.
    Then we have
    \begin{equation}\label{eq:symmetry_interaction_mult_in_appendix}
        I^{m}_{i_1,\ldots,i_n} = I^{m}_{\sigma(i_1,\ldots,i_n)}.
    \end{equation}
\end{prop}
\begin{proof}
    Let $J \subset K$ be a subset of $j$ indices and recall that $e^{(n)}_{J} = (e_{i_1},\ldots,e_{i_n})$ is the unique $n$-tuple such that $e_{i_l} = 1$ if $i_{l} \in J$ and $e_{i_l} = 0$ otherwise; in particular, this $n$-tuple contains $j$ ones and $n-j$ zeros.
    The same property holds for the $n$-tuple $e^{(n)}_{\sigma(J)}$, where $\sigma$ is any permutation of $K$.
    As a result, it suffices to show that $\sigma$ satisfies
    \begin{equation}
        I^{m}_{\sigma(i_1,\ldots,i_n)}(j) = I^{m}_{i_1,\ldots,i_n}(j),
    \end{equation}
    where
    \begin{equation}
        I^{m}_{i_1,\ldots,i_n} = \prod_{j=0}^{n} I^{m}_{i_1,\ldots,i_n}(j).
    \end{equation}
    \ie, that it fixes the $n+1$ factors $I^{m}_{i_1,\ldots,i_n}(j)$ of $I^{m}_{i_1,\ldots,i_n}$ separately.
    But any permutation of $K = \{i_1,\ldots,i_n\}$ simply permutes all subsets $J \subset K$ of fixed length $\ell(J) = j$ amongst each other.
    This completes the proof.
\end{proof}
As a corollary, we deduce the general permutation symmetry of the additive $n$-point interaction.
\begin{cor}\label{cor:symmetry_interaction_add}
    Let $K = \{i_1,\ldots,i_n\} \subset \{1,2,\ldots,r\}$ be a subset, and let $\sigma$ be any of the $n!$ permutations of $\{1,2,\ldots,n\}$ acting on $K$ as $\sigma(K) = \sigma(i_1,\ldots,i_n) = \{i_{\sigma(1),\ldots,\sigma(n)}\}$.
    The additive $n$-point interaction satisfies
    \begin{equation}
        I^{a}_{i_1,\ldots,i_n} = I^{a}_{\sigma(i_1,\ldots,i_n)}.
    \end{equation}
\end{cor}
\begin{proof}
    For the outcome $Y = -E(\uX)$, this follows directly by combining Eq.~\ref{eq:symmetry_interaction_mult_in_paper} and Eq.~\ref{eq:equivalence_interactions_add_mult}.
    For a general outcome $Y$, it follows by the argument of Prop.~\ref{prop:symmetry_interaction_mult}.
\end{proof}

\section{Hammersley--Clifford Theorem for the Ising model}
\label{app:Hammersley_Clifford}
Recall the $2$-dimensional Ising model of spins $\{v_i\}$ taking on the value $\pm 1$.
As an example, we explicitly establish the Hammersley--Clifford theorem of Sec.~\ref{sec:conditional_independence} in this case by verifying that its Hamiltonian,
\small
\begin{equation}
    p({\bf v}) = \frac{1}{\CZ(T)} e^{-E({\bf v})} \quad \text{where }
        E({\bf v}) = -\sum_{i,j} J_{i,j} v_i v_j,
\end{equation}
\normalsize
from Eq.~\ref{eq:H-ising} is locally, and hence globally, Markovian.
To do so, we denote by $\CN$ the set of all spins in the system, by $\CN_i$ the set of (four) spins neighbouring spin $i$, and we denote by $\CN_{-i}$ the set of all spins in the system apart from spin $i$.
The probability $p$ is locally Markovian if we have the equality,
\small
\begin{equation}\label{eq:local_Markovian_check_Ising}
    p\bigl(v_i = \pm 1 \mid v_{j} \text{ for } j \neq i \bigr) = p\bigl(v_i = \pm 1 \mid v_j \text{ for } j \in \CN_i \bigr),
\end{equation}
\normalsize
for each $i \in \CN$.
Fix a spin $v_0$ and denote its neighbours by $\CN_{0} = \{v_1,v_2,v_3,v_4\}$.
We will check that in the conditional probability on the left hand side of Eq.~\ref{eq:local_Markovian_check_Ising}, one only needs to condition on the spins $v_1,v_2,v_3,v_4$.
Here
\small
\begin{equation*}
\begin{split}
    &p\bigl(v_0,v_j | j \in \CN_{-0}\bigr) =
    \frac{1}{\CZ(T)}e^{-\sum_{i,j \neq 0} J_{i,j} v_i v_j} e^{-v_0 \sum_{i=1}^{4} (J_{0,i} v_i + J_{i,0}v_i)} , \\
    &p\bigl(v_j | j \in \CN_{-0} \bigr) = \frac{1}{\CZ(T)} e^{ -\sum_{i,j \neq 0} J_{i,j} v_i v_j } \\
    \cdot &\Bigl[e^{- \sum_{i=1}^{4} (J_{0,i}v_i + J_{i,0}v_i)} + e^{ \sum_{i=1}^{4} (J_{0,i} v_i + J_{i,0}v_i)} \Bigr].
\end{split}
\end{equation*}
\normalsize
It follows that their ratio, which is by definition the binary probability distribution $p\bigl(v_0 \mid v_{j} \text{ for } j \neq i \bigr)$, is fully determined once one conditions on the four nearest neighbour spins $v_1,v_2,v_3,v_4$ of $v_0$.
This proves the claim.

\section{Linear regression}\label{sec:linear_regression_numerical}
Let us consider the regression model with quadratic and cubic terms, representing additive $2$- and $3$-point interactions amongst the effects of the binary random variables $T_1, T_2$, and $T_3$ on $Y$:
\begin{equation}
\begin{split}
    Y = &\alpha_0 + \alpha_1 T_1 + \alpha_2 T_2 + \alpha_3 T_3 + \alpha_{12} T_1 T_2 \\
    + &\alpha_{13} T_1 T_3 + \alpha_{23} T_2 T_3 + \gamma T_1 T_2 T_3 + \epsilon \ .
\end{split}
\end{equation}
The noise term $\epsilon$ is normally distributed $\mathcal{N}(0,\sigma^2)$ with $\sigma^2=1$. Without loss of generality, the ground truth 3-point interaction $\gamma$ is set to twice the value of the noise, \ie, $\gamma=2$, while the 2-point interactions are set to $\alpha_{12},\alpha_{13},\alpha_{23}=5.0,-2.5,0$ respectively. The zeroth order coefficient $\alpha_0=-1.5$ and the linear coefficient are set to $\alpha_1,\alpha_2,\alpha_3=-2,10,0$. 
We generate $N_s = 40,80,\ldots,1000$ samples with $T_1 \sim \Binom(0.4)$, $T_2 \sim \Binom(0.7)$, $T_3 \sim \Binom(0.5)$, where we have fixed regression coefficients to be as above. 
We then take as input $(Y,T_1,T_2,T_3)$, and compute the expectation values in Eq.~\ref{eq:interaction_add} to estimate the 2-point and 3-point interactions, for varying sample sizes $N_s$, and compare with the ground truth values used to generate the data. 

In order to ensure the estimates are robust, sufficiently many sub-samples have to be available for estimating each of the four conditional expectation values appearing in Eq.~\ref{eq:interaction_add}. 
As with any statistical estimator, having very few samples for one of the conditional expectation values may result in unstable estimates of the expectation value and its variance. This will in turn introduce instabilities in the estimates of the interactions.
See App.~\ref{app:linear_regression} for a comparison of bin sizes for each of the expectation values as the total sample size increases.

The three 2-point interactions and the 3-point interaction amongst variables $T_1,T_2,T_3$ are presented in Fig.~\ref{fig:regression-2pt-3pt}.
The uncertainties on the estimates are derived using statistical bootstrap \cite{efron1979}.
One can readily observe that as the sample size increases, the estimates converge to the correct value with smaller variance as expected. 

\begin{figure}[!htb]
\begin{center}
    \minipage{0.4\textwidth}
	\includegraphics[width=\linewidth]{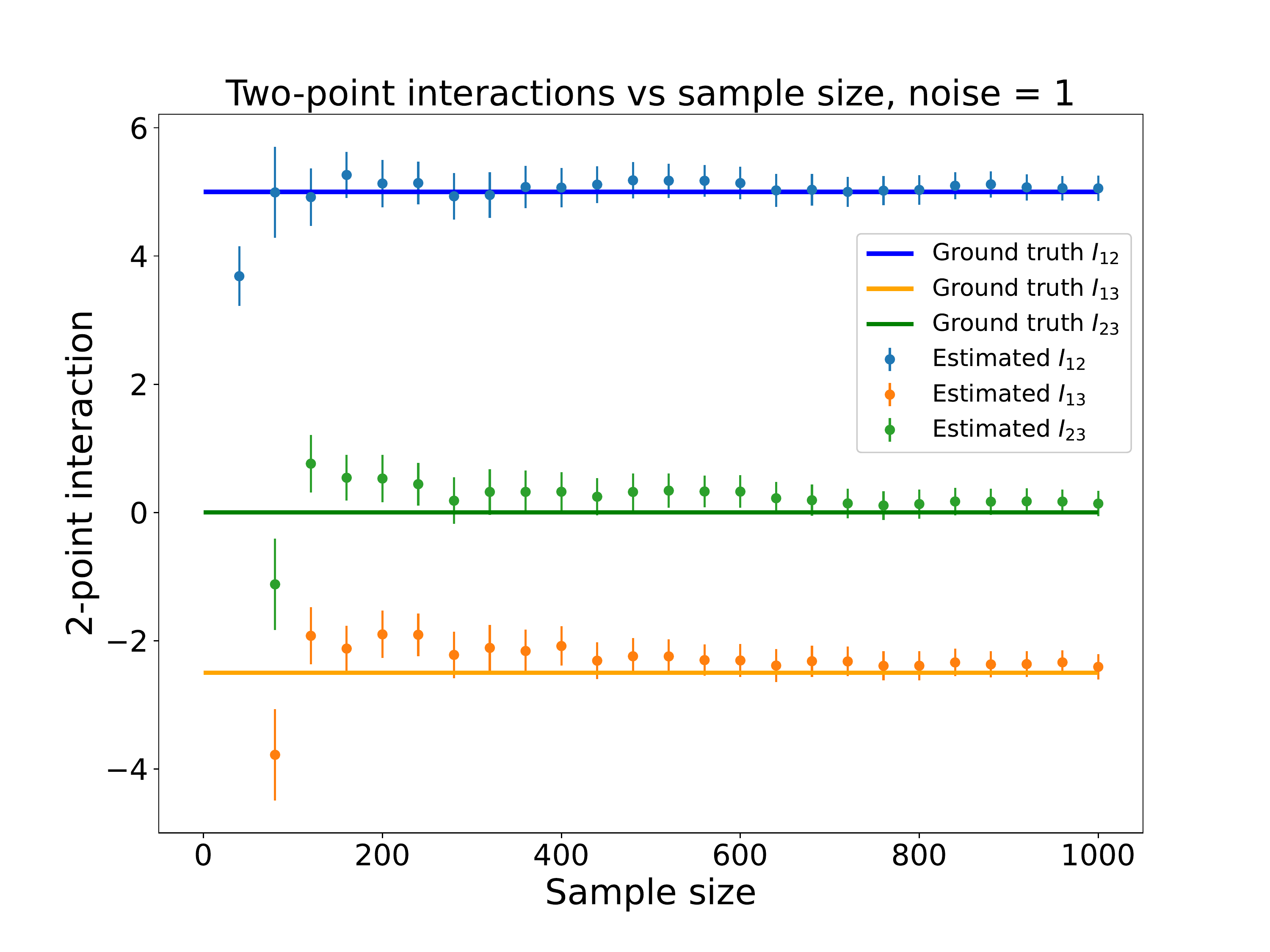}
	\endminipage\hfill
	\minipage{0.4\textwidth}
	\includegraphics[width=\linewidth]{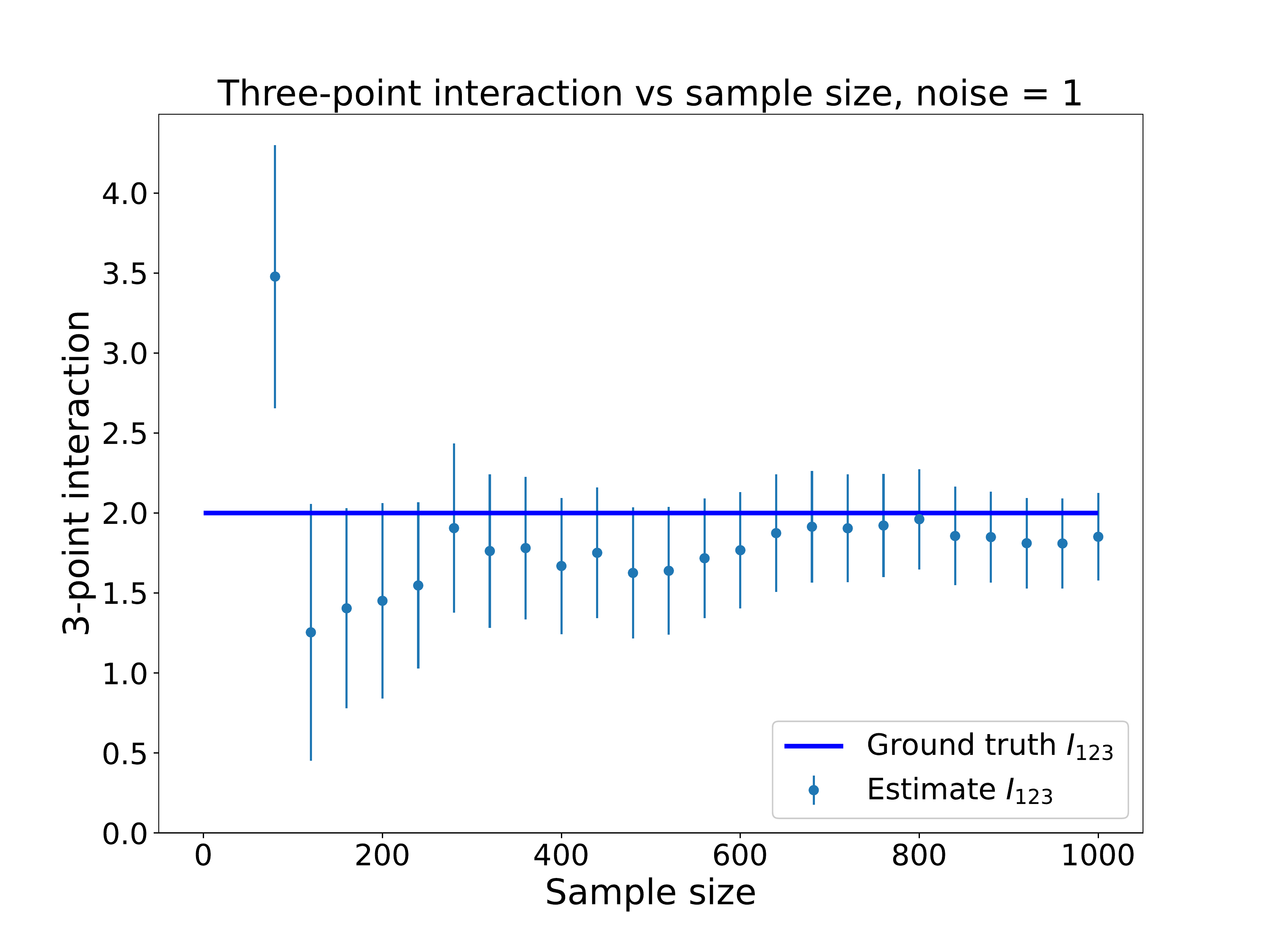}
	\endminipage\hfill
    \caption{Estimates of 2-point (top) and 3-point (bottom) interaction as a function of sample size, with noise $\sigma^2=1$.
    The uncertainties on the estimates are derived using statistical bootstrap.
    See Fig.~\ref{fig:regression_bin} in App.~\ref{app:linear_regression} for a comparison of bin sizes for each of the expectation values as the total sample size increases.}
    \label{fig:regression-2pt-3pt}
\end{center}
\end{figure}


\section{Linear regression: bin sizes as a function of sample size}
\label{app:linear_regression}
In Fig.~\ref{fig:regression_bin} we plot the bin sizes for each of the four expectation values appearing in Eq.~\ref{eq:interaction_add_2} as the sample size grows.
When the total sample size is, \eg, $N_s=40$, some of the conditional expectation values are estimated using one or two samples only and thus are unreliable.  

\begin{figure}
\begin{center}
    \minipage{0.4\textwidth}
	\includegraphics[width=\linewidth]{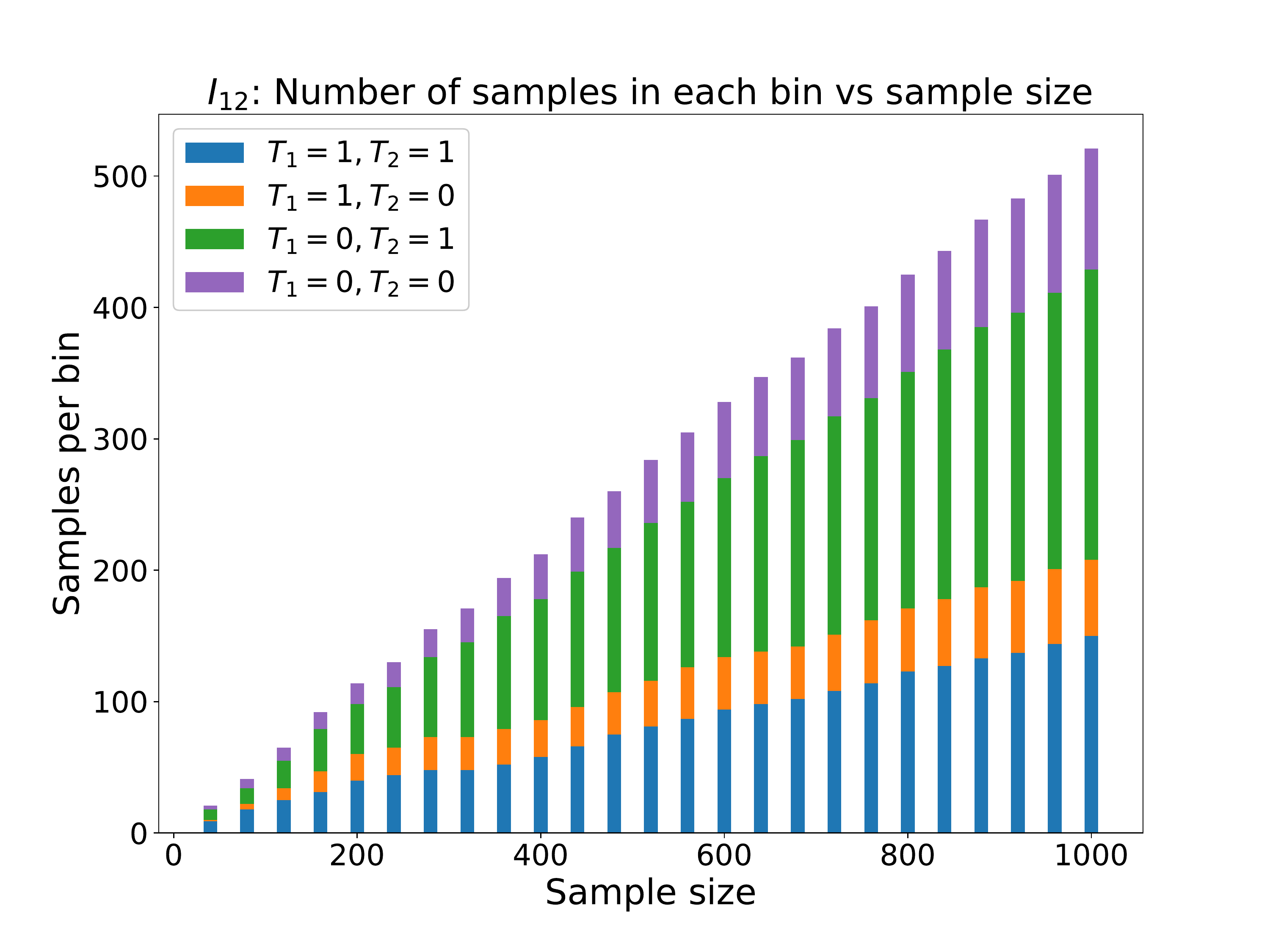}
	\endminipage
	\vspace{0.2cm}
	\minipage{0.4\textwidth}
	\includegraphics[width=\linewidth]{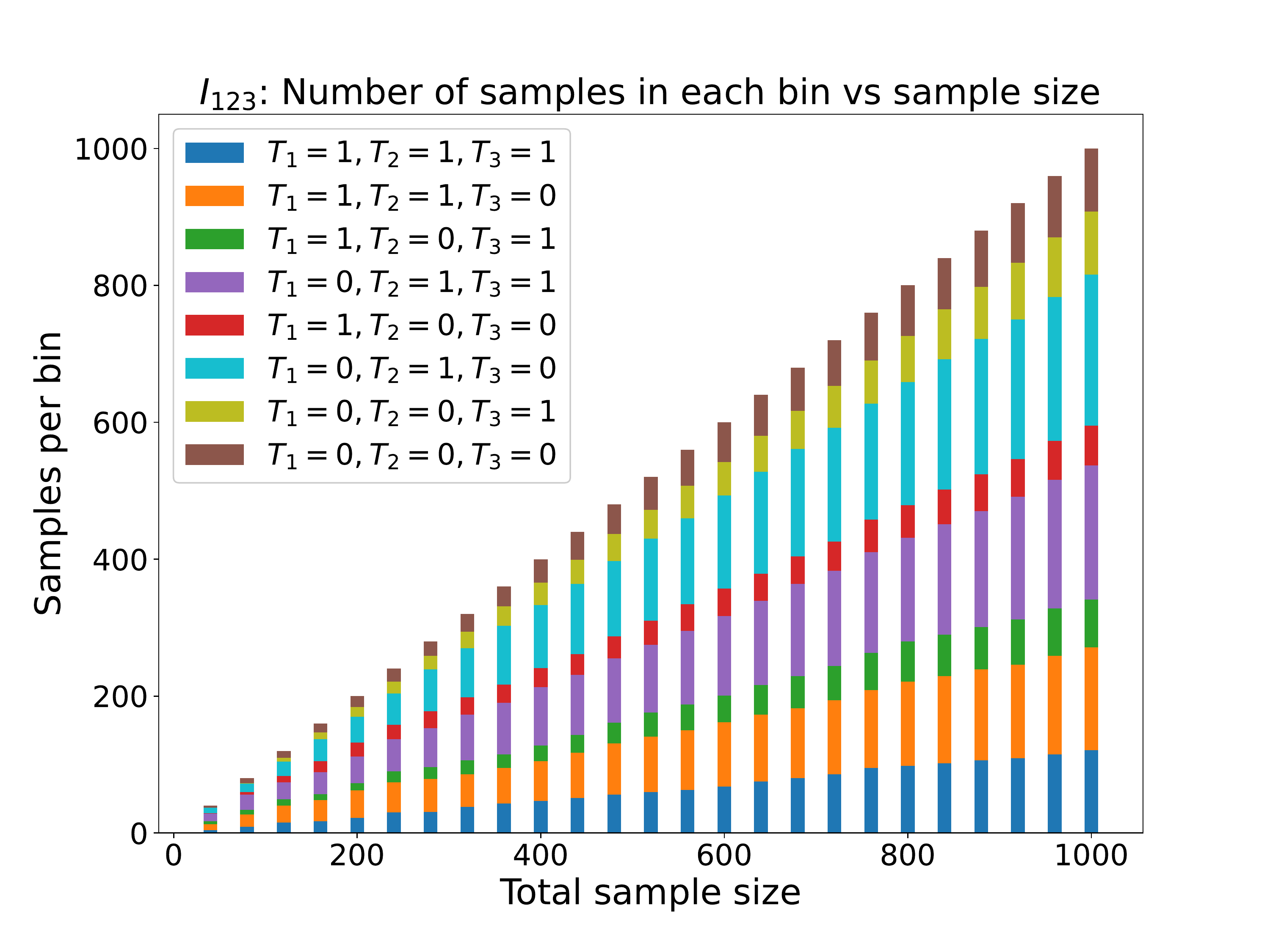}
	\endminipage\hfill
    \caption{Number of samples for each of the expectation values vs total sample size. Top: For the 2-point interaction $I_{12}$. The variables are distributed as $T_1 \sim \Binom(0.4)$ and $T_2 \sim \Binom(0.7)$ so that, \eg, the bin size of $(T_1,T_2) = (1,0)$ is the smallest, whereas the one of $(T_1,T_2) = (0,1)$ is the largest. Bottom: For the 3-point interaction $I_{123}$, where $T_3 \sim \Binom(0.5)$. The legend $T_1=T_2=T_3=1$ and $T_1=T_2=T_3=0$ are placed lowest and highest in the bar plot respectively. }
    \label{fig:regression_bin}
\end{center}
\end{figure}

\section{Interaction estimates per spin pair for the Ising model}
\label{app:spin_pair_exp_L8_T18_T22_T30}

We present the histogram of 2-point interactions amongst all pairs of (non)-nearest neighbours, using Eq.~\ref{eq:interaction_mult_2} for Ising states simulated at temperature $T=1.8$ and $L^2=8^2$.
As follows from Fig.~\ref{fig:ising_two_point_T18_histogram_100K_1M}, as the total sample size increases the two peaks corresponding to zero couplings between non-nearest neighbour pairs and positive couplings at $\frac{1}{2T} \approx 0.28$ corresponding to the nearest neighbour pairs, become more distinct. 

\begin{figure}[!htb]
\begin{center}
    \minipage{0.4\textwidth}
	\includegraphics[width=\linewidth]{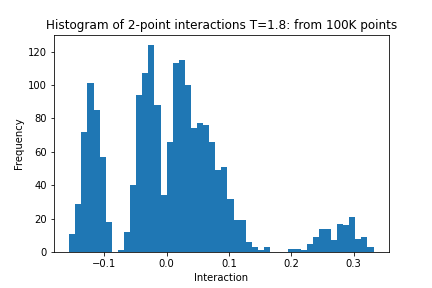}
	\endminipage\hspace{0.03\textwidth}
	\minipage{0.4\textwidth}
	\includegraphics[width=\linewidth]{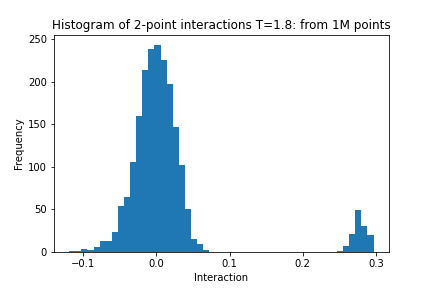}
	\endminipage\hfill
    \caption{Histograms of $100$K (top) and $1$M (bottom) estimates of the $2$-point interaction at $T = 1.8$, in an Ising system of size $L^2 = 8^2$. The interactions are computed directly from the data using the non-parametric multiplicative formulation in Eq.~\ref{eq:interaction_mult_2}. As expected, with larger sample sizes, the peaks corresponding to non-nearest neighbour interactions, around zero, and nearest neighbour interactions, around $\frac{1}{2T} \approx 0.28$, become more distinct with less noise. }
    \label{fig:ising_two_point_T18_histogram_100K_1M}
\end{center}
\end{figure}

The estimates of 2-point couplings for both the nearest neighbour and non-nearest neighbour spin pairs, using 100K (top) and 20K (bottom) sample sizes, are presented in Fig.~\ref{fig:L8_T18_per_spin_indept_and_expect}.
As mentioned in Sec.~\ref{sec:Ising_RBM_numerical}, one can use smaller sample sizes to estimate the couplings at the cost of reduced power.
For colder temperatures and small sample sizes, there may be no states in the $p_{11}$ bin, for the case of non-nearest neighbour spin pairs.
For $T=1.8$ over 20K samples, we have power to accurately estimate all the nearest neighbour couplings, but only have power to accurately estimate approximately 70\% of couplings between non-nearest neighbour pairs.
As expected, increasing the sample size to 100K improves the latter to 99\%.
Note that with real data sets, one may have limitations on the sample size.
For example, as shown in Fig.~\ref{fig:L8_T23_per_spin_indept_and_expect_10K}, the non-parametric estimator, combined with conditional independence amongst the variables, has nevertheless enabled us to obtain accurate estimates using 10K samples only.
In contrast, \eg, the RBM does not train well on Ising data with 10K samples, see~\cite[Fig.~31]{PhysRevB.100.064304}.

Fig.~\ref{fig:cond_ising_two_point_interaction_vs_temp_10K} illustrates the estimates for nearest neighbour interactions vs temperature with 10K total samples using the TL framework.

\begin{figure}[!htb]
\begin{center}
    \minipage{0.45\textwidth}
	\includegraphics[width=\linewidth]{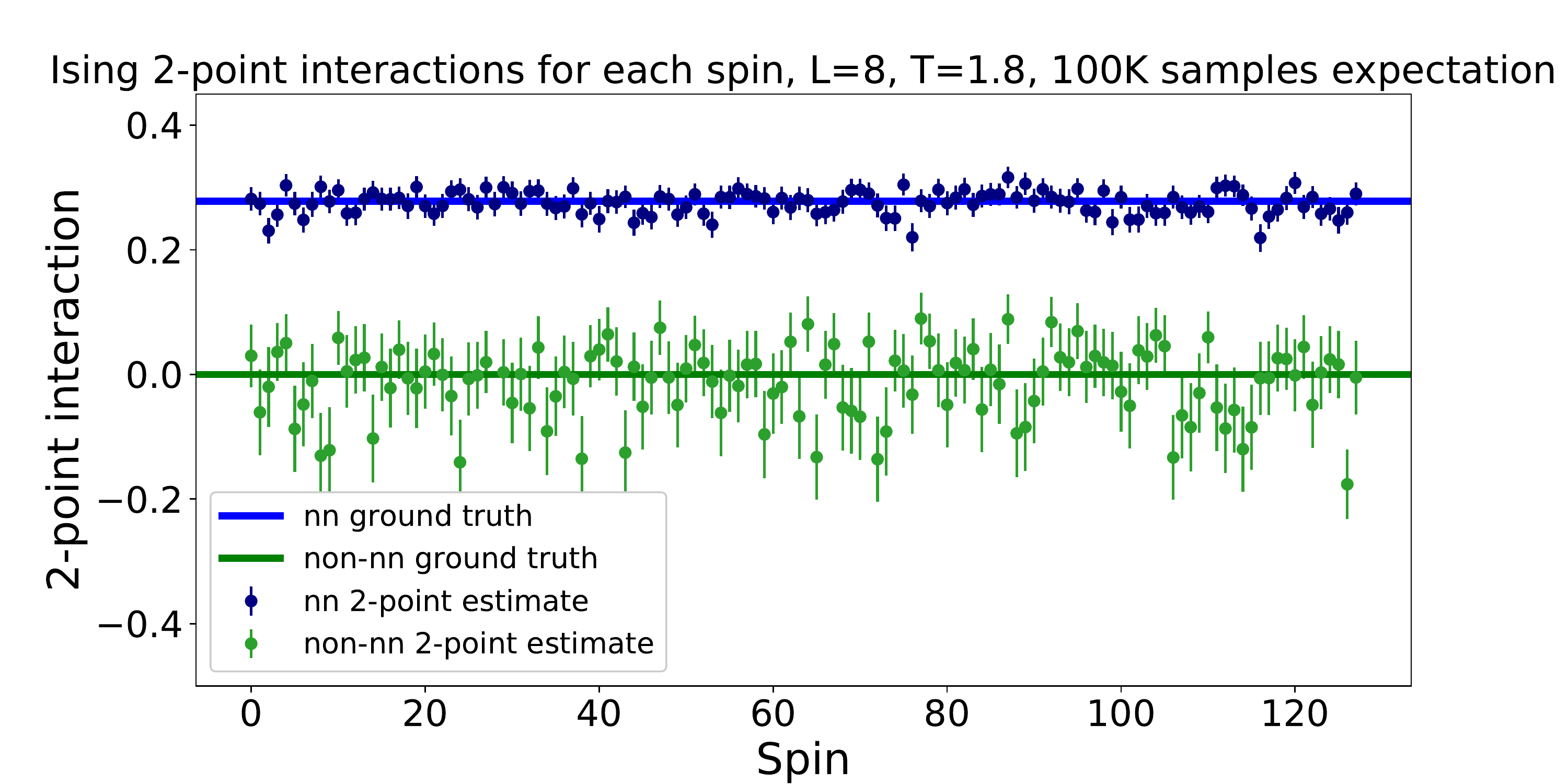}
	\endminipage
	\vspace{0.2cm}
	\minipage{0.45\textwidth}
	\includegraphics[width=\linewidth]{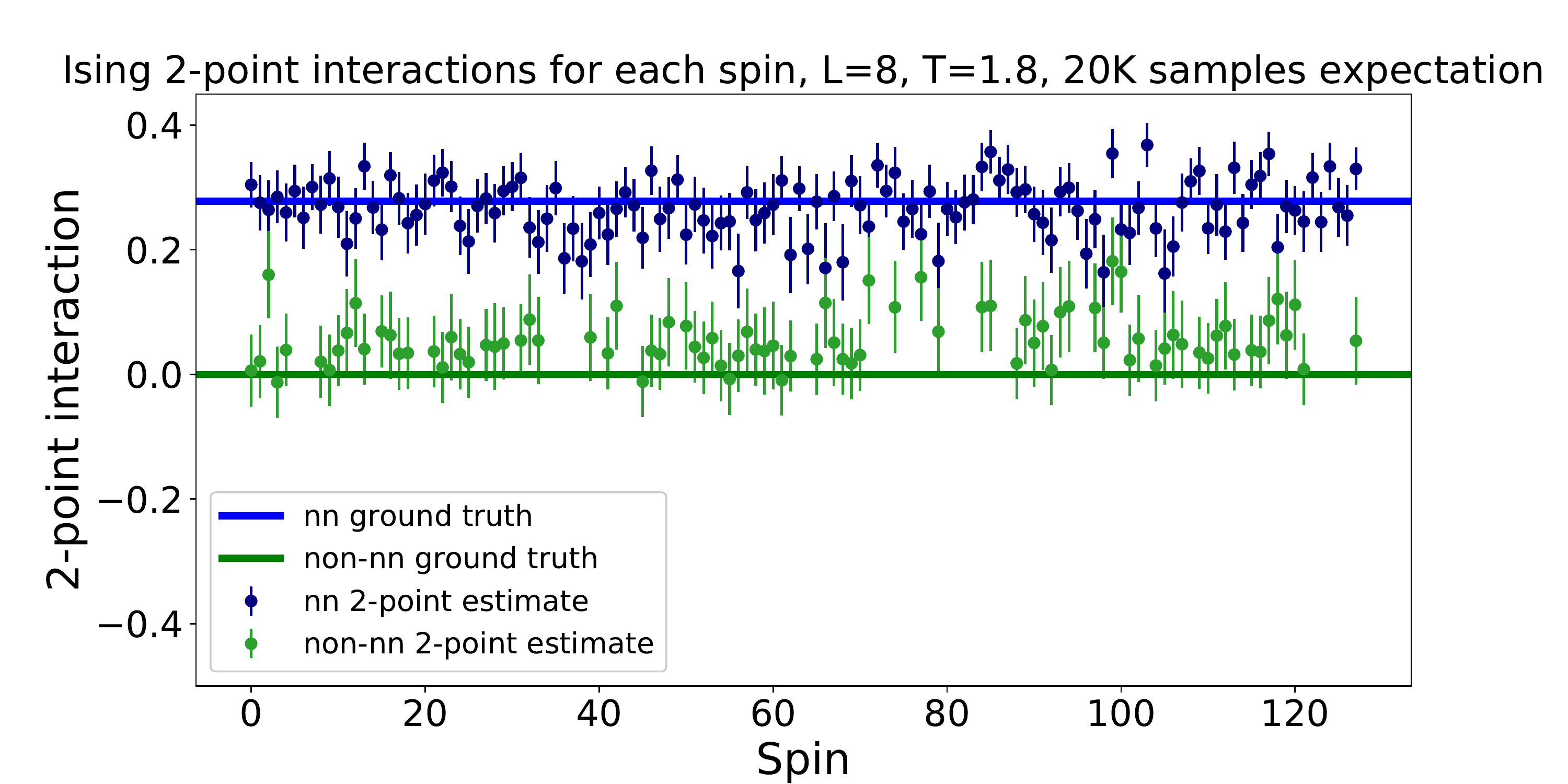}
	\endminipage\hfill
    \caption{$L^2=8^2$, $T=1.8$, with conditioning on the nearest neighbours to estimate $I_{ij}^m$ for both nearest and non-nearest neighbours. In order to reduce clutter, the same number of non-nearest couplings as nearest neighbours are shown (128). No translational invariance is used. Top: The results over 100K samples, using Eq.~\ref{eq:interaction_mult_2_expectation} and statistical bootstrap, as compared to bottom: The results over 20K samples. For the latter, approximately 30\% of spins had no samples in the $p_{11}$ bin. This is due to the fact that it is very rare to find 2 spins having value one, whilst their 8 nearest neighbours all have spin value 0, particularly at cold temperatures, as the total sample size become smaller.}
    \label{fig:L8_T18_per_spin_indept_and_expect}
\end{center}
\end{figure}

\begin{figure}[!htb]
\begin{center}
    \minipage{0.45\textwidth}
	\includegraphics[width=\linewidth]{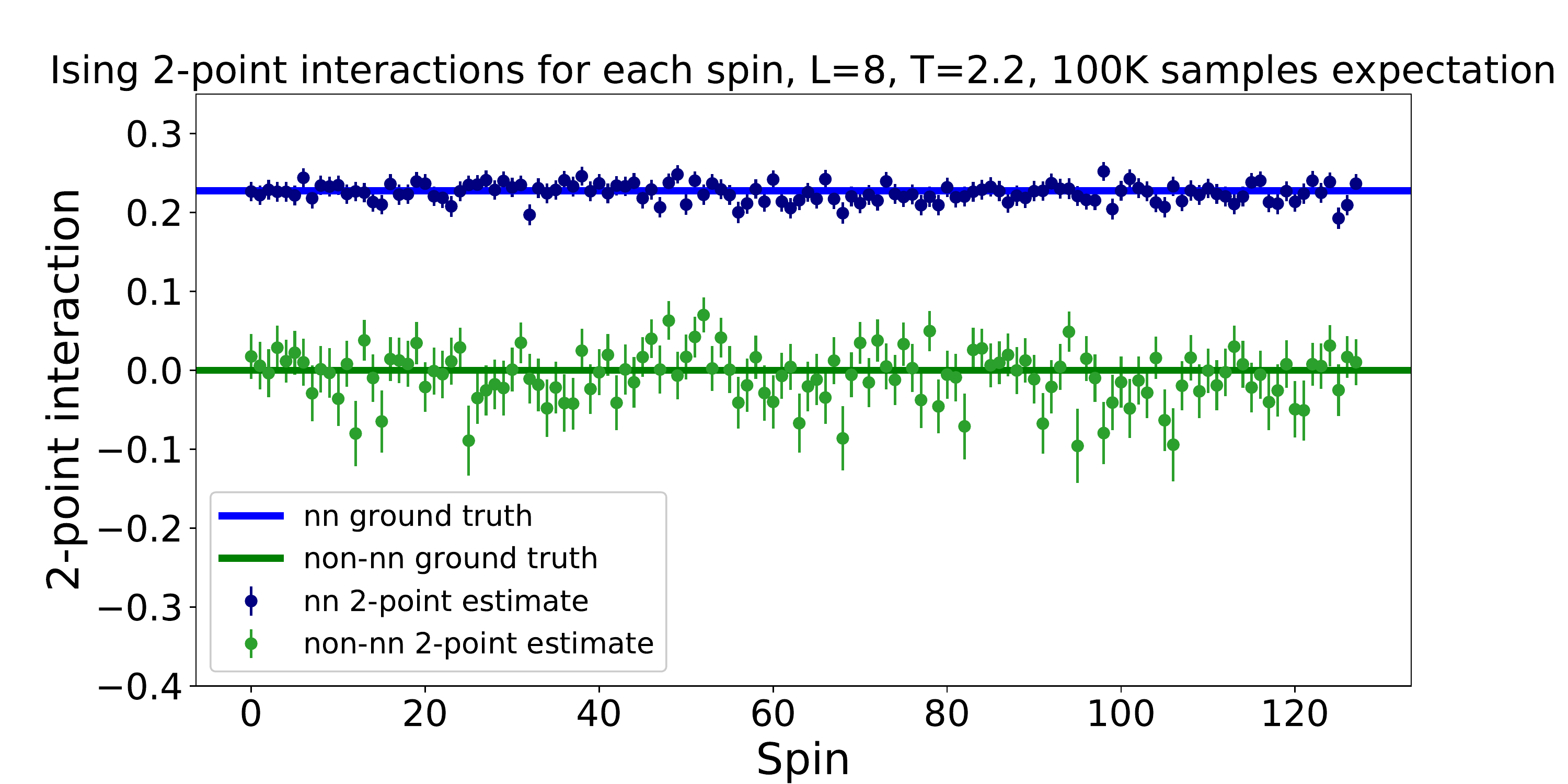}
	\endminipage
	\vspace{0.2cm}
	\minipage{0.45\textwidth}
	\includegraphics[width=\linewidth]{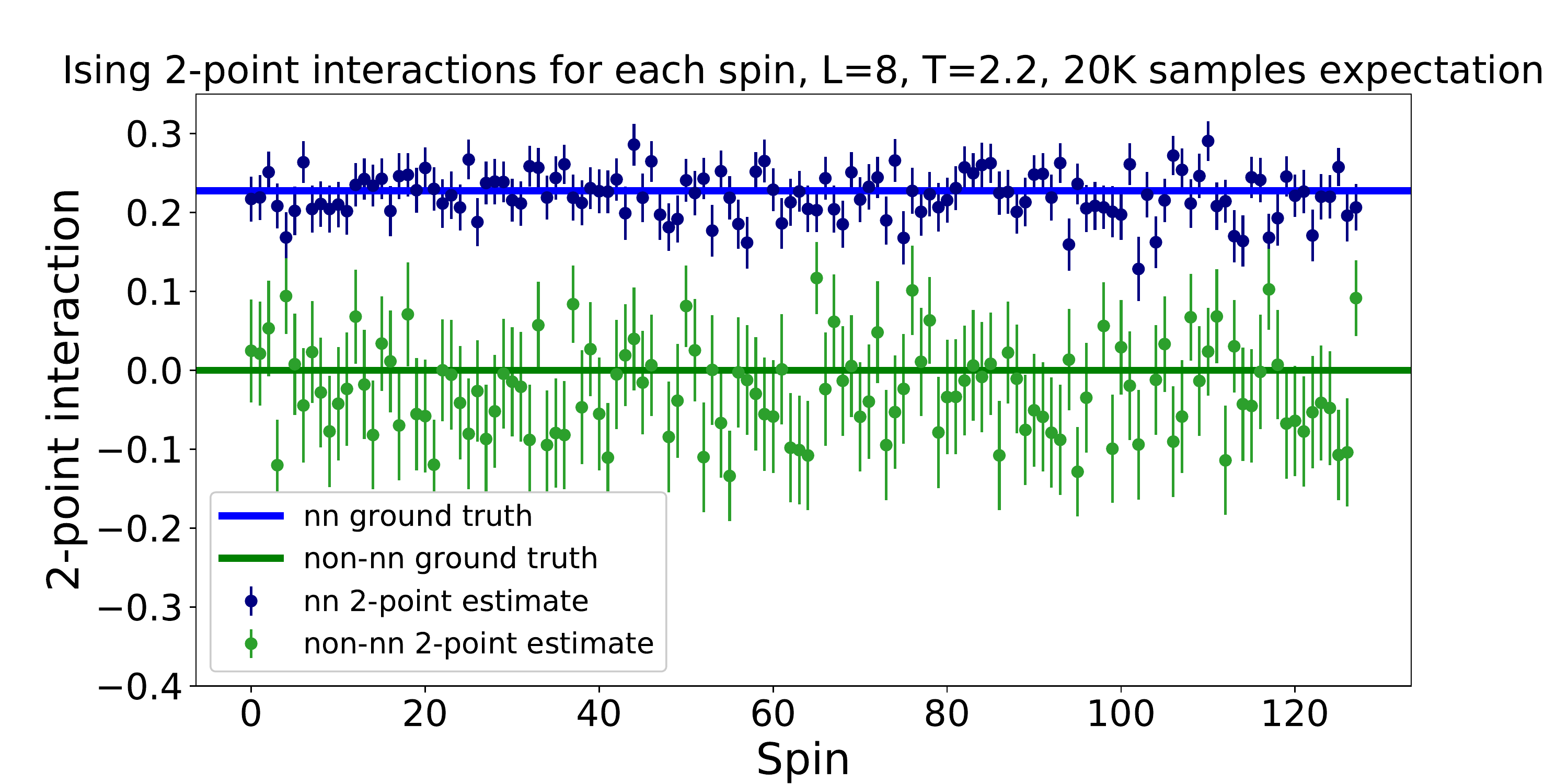}
	\endminipage\hfill
    \caption{$L^2=8^2$, $T=2.2$, with conditioning on the nearest neighbours to estimate $I_{ij}^m$ for both nearest and non-nearest neighbours. In order to reduce clutter, the same number of non-nearest couplings as nearest neighbours are shown (128). No translational invariance is used. Top: The results over 100K samples, using Eq.~\ref{eq:interaction_mult_2_expectation} and statistical bootstrap, as compared to bottom: The results over 20K samples. At 20K samples we have power to accurately estimate approximately 98\% of non-nearest neighbour spin pairs.}
    \label{fig:L8_T22_per_spin_indept_and_expect}
\end{center}
\end{figure}


\begin{figure}[!htb]
\begin{center}
    \minipage{0.45\textwidth}
	\includegraphics[width=\linewidth]{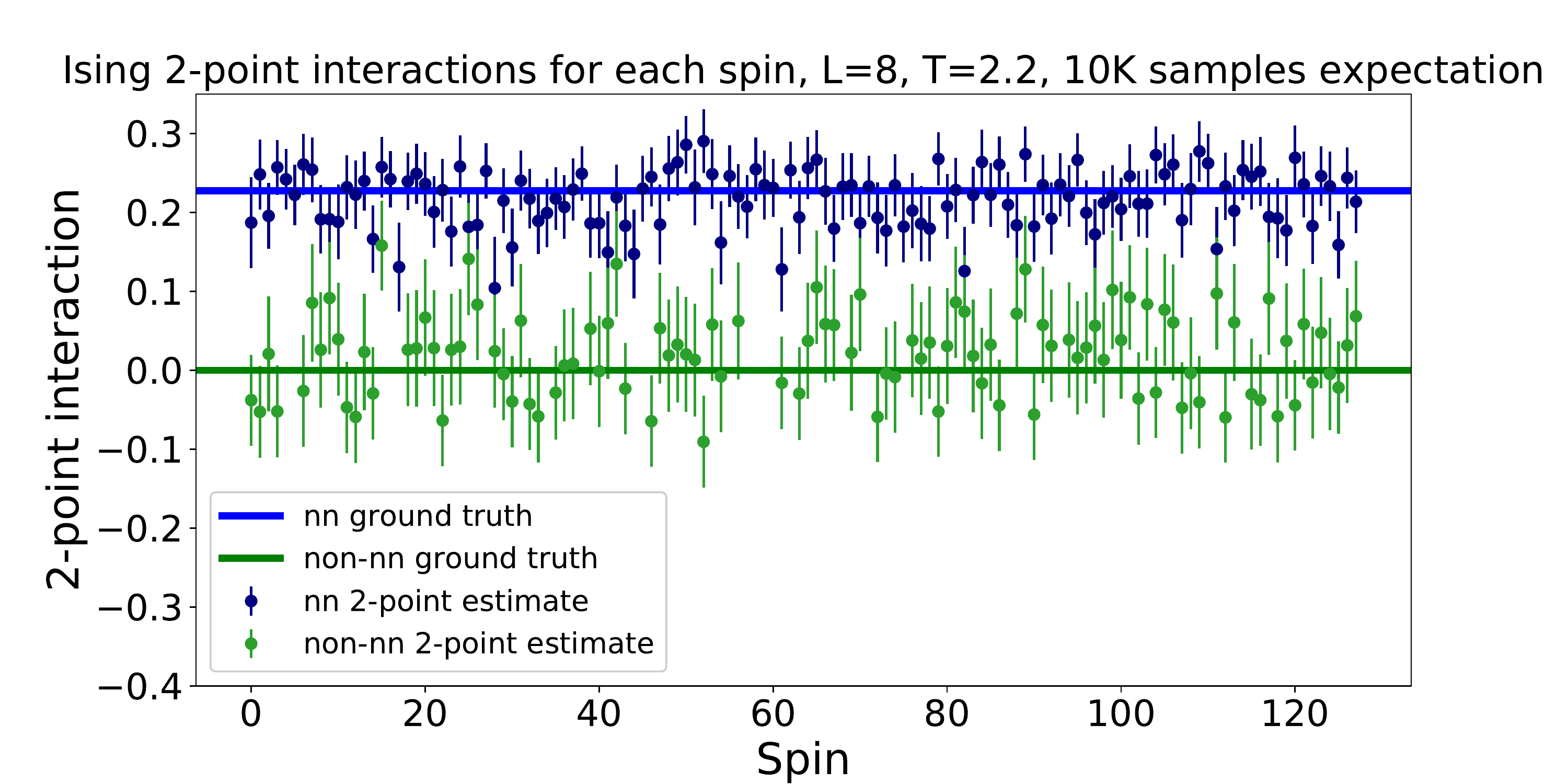}
	\endminipage
    \caption{$L^2=8^2$, $T=2.2$, with conditioning on the nearest neighbours to estimate $I_{ij}^m$ for (non-)nearest neighbours. In order to reduce clutter, the same number of non-nearest couplings as nearest neighbours are shown (128). Similar to the results in Fig.~\ref{fig:L8_T22_per_spin_indept_and_expect}, except the total sample size is now 10K only. There is enough power to accurately estimate $I_{ij}^m$ for all nearest neighbour pairs, and approximately 83\% of the non-nearest neighbour pairs. In contrast, \eg, the RBM does not train on 10K samples, see~\cite[Fig.~31]{PhysRevB.100.064304}.}
    \label{fig:L8_T23_per_spin_indept_and_expect_10K}
\end{center}
\end{figure}

\begin{figure}[!htb]
\begin{center}
    \minipage{0.45\textwidth}
	\includegraphics[width=\linewidth]{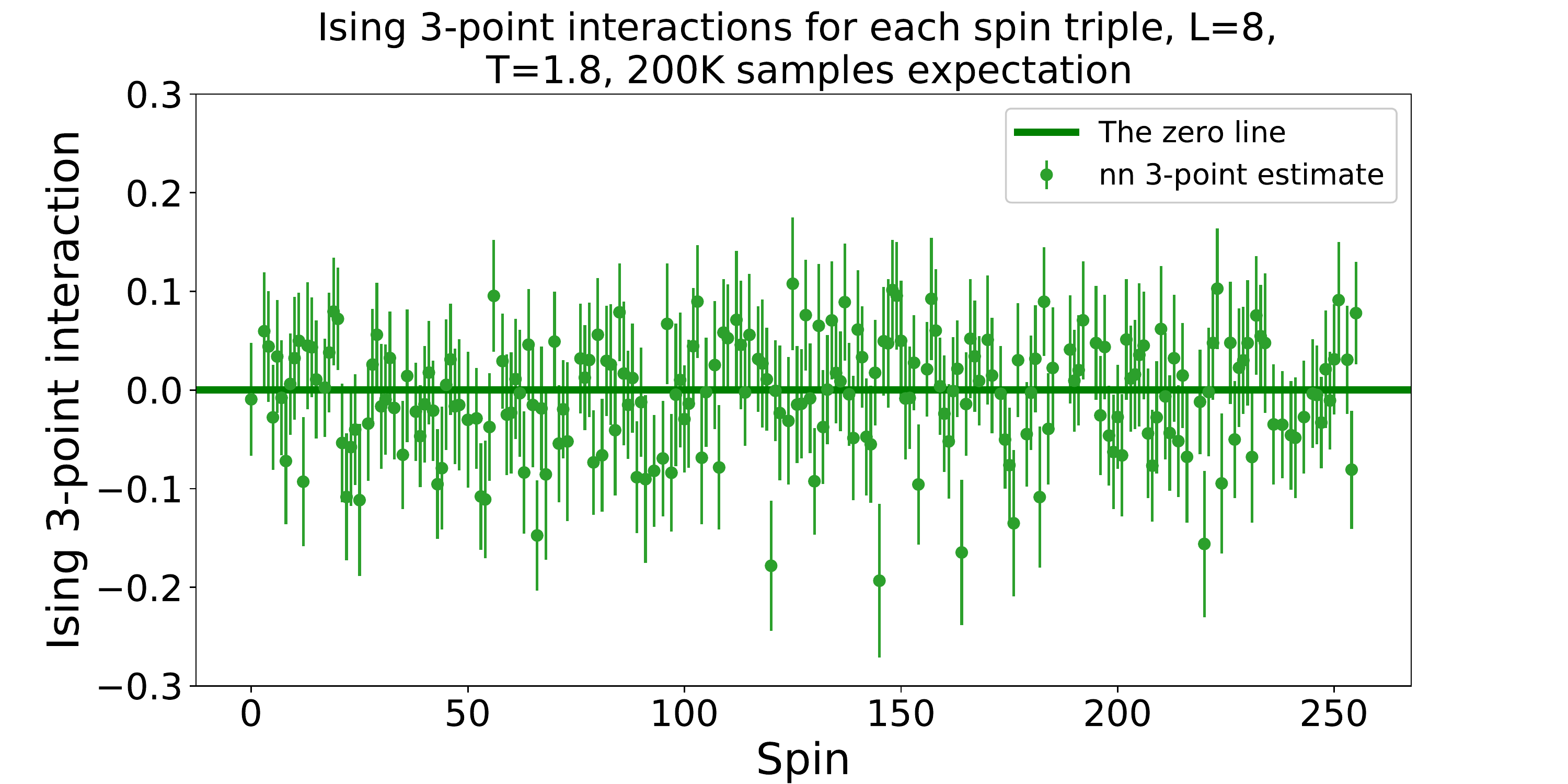}
	\endminipage
	\vspace{0.2cm}
	\minipage{0.45\textwidth}
	\includegraphics[width=\linewidth]{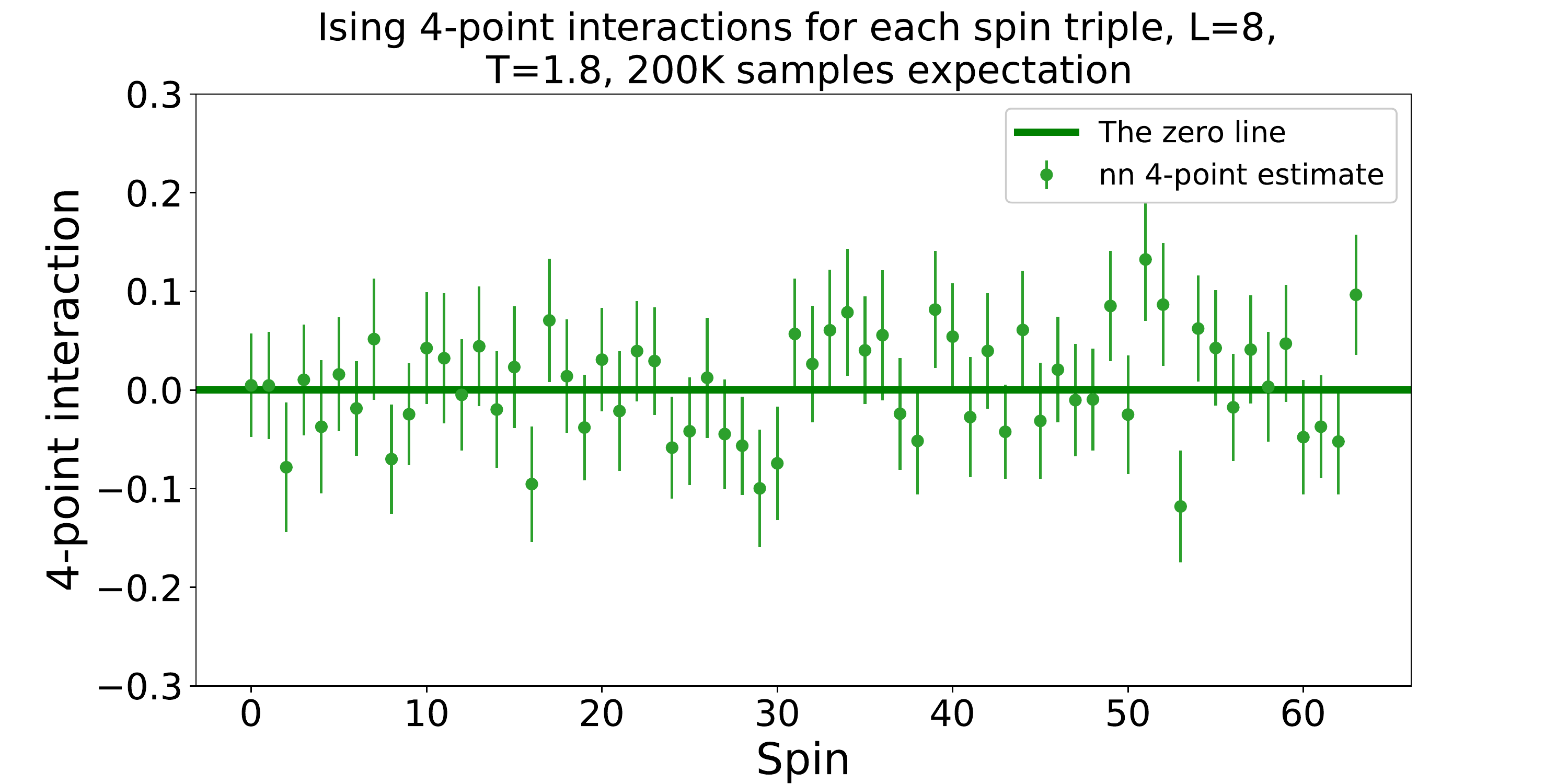}
	\endminipage\hfill
    \caption{$L^2=8^2$, $T=1.8$, with conditioning on the nearest neighbours to estimate 3-point (top) and 4-point (bottom) interaction for the nearest neighbours. Due to the cold temperature, $85\%$ of triples can be estimated, all 4-points are estimated. If 100K samples are used $40\%$ of the 3-points can be estimated, but they are all accurately zero within statistics, similar to the top plot. }
    \label{fig:L8_T18_per_spin_3pt_4pt_cond_200K}
\end{center}
\end{figure}

\begin{figure}[!htb]
\begin{center}
    \minipage{0.50\textwidth}
	\includegraphics[width=\linewidth]{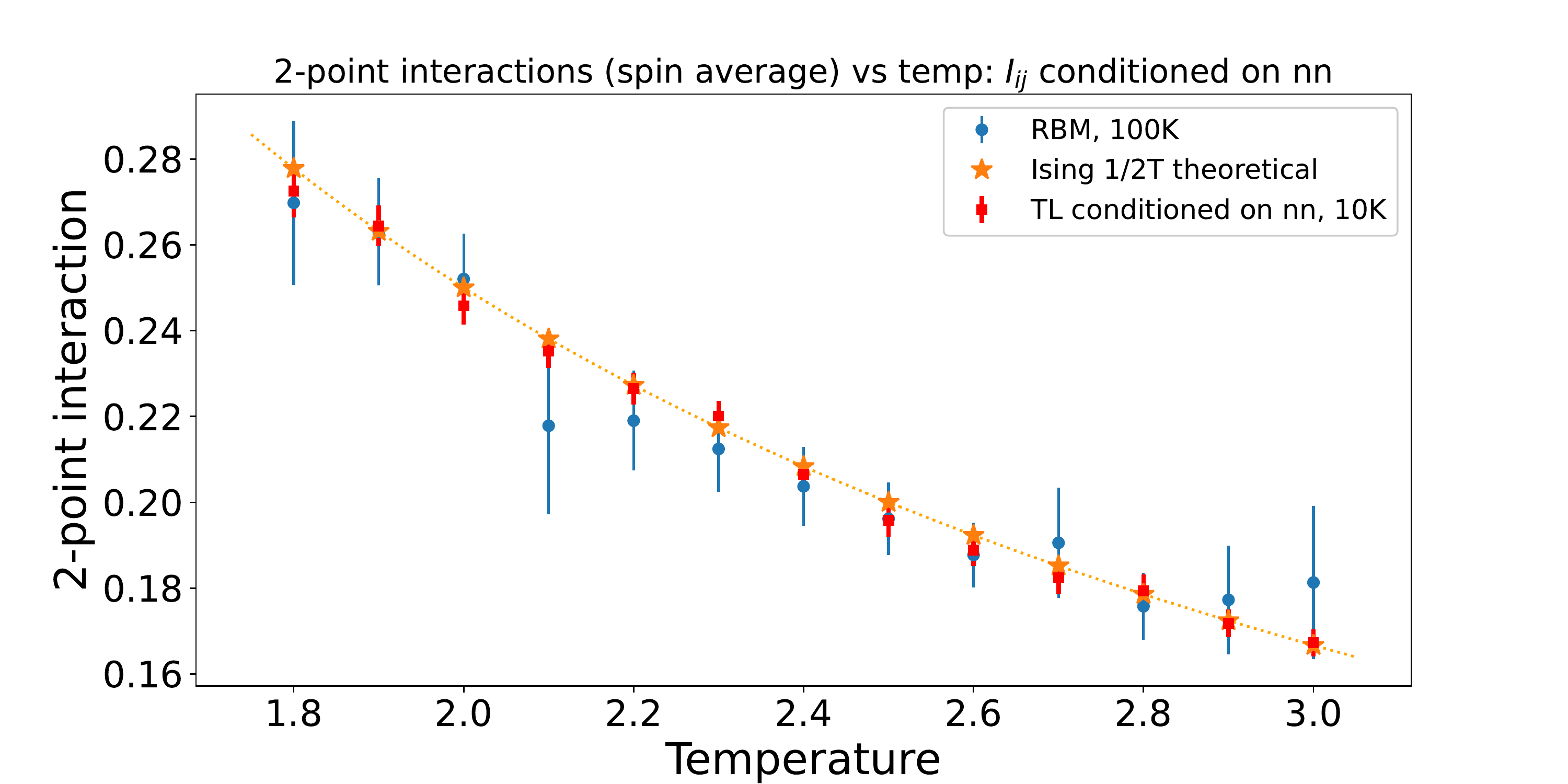}
	\endminipage\hfill
    \caption{Conditioning on the nearest neighbours to estimate $I_{ij}^m$ substantially improves the estimates as compared to Fig.~\ref{fig:ising_two_point_interaction_vs_temp}. The square points are estimations of interactions and their uncertainty using TL with 10K samples. The run time for each estimation using TL is at the order of a few seconds.}
    \label{fig:cond_ising_two_point_interaction_vs_temp_10K}
\end{center}
\end{figure}

\clearpage

\bibliographystyle{unsrt}       
\bibliography{refs}

\end{document}

%% file: macros.tex
\usepackage{tikz-cd}
\usepackage{amsthm}

\newtheorem{prop}{Proposition}[section]

\newtheorem{cor}[prop]{Corollary}

\newtheorem{defn}[prop]{Definition}

\def\ie{{\it i.e.}}
\def\eg{{\it e.g.}}

\newcommand{\ATE}{\operatorname{ATE}}

\newcommand{\uX}{\underline{X}}
\newcommand{\uT}{\underline{T}}

\DeclareMathOperator{\Binom}{Binom}

\newcommand{\BE}{{\mathbb{E}}}

\newcommand{\BN}{{\mathbb{N}}}

\newcommand{\BR}{{\mathbb{R}}}

\newcommand{\CM}{{\mathcal M}}
\newcommand{\CN}{{\mathcal N}}
\newcommand{\CO}{{\mathcal O}}
\newcommand{\CP}{{\mathcal P}}

\newcommand{\CZ}{{\mathcal Z}}